\documentclass[accepted=2019-01-14]{quantumarticle}

\input{graphicalcrypto.sty}

\begin{document}

\title{Graphical Methods in Device-Independent \newline Quantum
  Cryptography}

\date{}

\author{Spencer Breiner,\texorpdfstring{$^{1}$}{} Carl
  A. Miller,\texorpdfstring{$^{1,2}$}{} and Neil
  J. Ross\texorpdfstring{$^{2,3}$}{}\texorpdfstring{\\}{} 
  \small \texorpdfstring{$^1$}{}National Institute of Standards and
  Technology (NIST)\texorpdfstring{\\}{}
  \small Gaithersburg, MD 20899, USA\texorpdfstring{\\}{}
  \small \texorpdfstring{$^2$}{}Joint Center for Quantum Information
  and Computer Science (QuICS)\texorpdfstring{\\}{}
  \small University of Maryland, College Park, MD 20742,
  USA\texorpdfstring{\\}{}
  \small \texorpdfstring{$^3$}{}Department of Mathematics and
  Statistics\texorpdfstring{\\}{}
  \small Dalhousie University, Halifax, NS B3H 4R2, Canada}
  
\maketitle

\begin{abstract}
  We introduce a framework for graphical security proofs in
  device-independent quantum cryptography using the methods of
  categorical quantum mechanics.  We are optimistic that this approach
  will make some of the highly complex proofs in quantum cryptography
  more accessible, facilitate the discovery of new proofs, and enable
  automated proof verification. As an example of our framework, we
  reprove a previous result from device-independent quantum
  cryptography: any linear randomness expansion protocol can be
  converted into an unbounded randomness expansion protocol.  We give
  a graphical proof of this result, and implement part of it in the
  Globular proof assistant.
\end{abstract}

\section{Introduction}
\label{sec:intro}

Graphical methods have long been used in the study of physics and
computation. In physics, this can be traced back at least as far as
Penrose's use of diagrams \cite{Penrose}. During the last decade of
the twentieth century, rigorous methods for graphical reasoning in
monoidal categories were developed by Joyal, Street, and others
\cite{JS91,Selinger2011}. When Abramsky and Coecke proposed monoidal
categories as an alternative foundation for quantum physics
\cite{AC04}, they were able to draw from these technical developments
to introduce an elaborate graphical language for reasoning about
quantum mechanical concepts. Since then, the use of rigorous graphical
methods has been extended widely, ranging from foundations \cite{AC04}
to quantum algorithms \cite{Vic13}, quantum error correction
\cite{CKZH16}, and beyond \cite{C16, CKPict}. The great success of
graphical methods in the quantum sciences is largely due to their
ability to deal with elaborate concepts in a simple way. This is
especially true when compared to the standard methods involving linear
operators acting on Hilbert space.

Quantum cryptography, the study of cryptographic protocols that are
based on quantum mechanical principles (see Section~12.6 of
\cite{NC02} for an introduction) is an ideal candidate for graphical
analysis. Indeed, proofs in quantum cryptography are often long and
complicated even when the central idea of the proof is relatively
clear. Pictures are regularly used as a conceptual aid in discussions
of quantum cryptography, but it would be beneficial for both
accessibility and rigor if proofs themselves could be expressed as
pictures. For this reason, we believe that the field of quantum
cryptography can benefit from the abstract methods of categorical
quantum mechanics \cite{CKPict}.

To our knowledge the use of graphical methods to formalize quantum
cryptography is fairly new, although the literature provides some
useful beginnings.  Graphical security proofs for quantum key
distribution (one of the original problems in the field) have been
presented in \cite{Coecke2010,CoeckeQKD,Hillebrand2011SuperdenseCW},
although these are not yet at the level of security that has been
proved through non-graphical means.  Meanwhile, the literature has a
number of formal treatments of cryptography that are not primarily
based on graphical reasoning. An important example is \cite{CSW14},
which has a focus similar to the current paper.  See also
\cite{Heunen,Maurer11abstractcryptography, Pav14, StayVicary}.

One of the recent major achievements of quantum cryptography is the
development of \textit{device-independent} security proofs
\cite{mayersyao}.  In such proofs, not only the adversary but also the
quantum devices are untrusted, and allowed to exhibit uncharacterized
behavior.  Protocols in this field always include a classical test of
the quantum devices which lead to a ``succeed'' or ``abort'' event at
the end of the protocol, and we must prove that if the protocol
``succeeds,'' then the desired cryptographic task has been securely
carried out.  Device-independence is a desirable level of security for
quantum cryptography (in particular because it accounts for arbitrary
noise or imperfection in the quantum hardware) and it will be our
focus in this paper.

As a specific problem to study, we consider \emph{device-independent
  randomness expansion}.  Research over the last decade has led to a
proof that random numbers can be generated with devices that are
completely untrusted \cite{Arnon-Friedman:2018, Arnon:2016,
  bierhorst2017, colbeck2007quantum, CY13, Dupuis:2016,
  fehr2013security, Knill:2018, MY14-2,MY14-1, pironio2010random,
  pironio2013security, CSW14}.  Precisely, two untrusted quantum
devices, together with a perfectly random seed of length $N$, can be
manipulated by a classical user to generate a perfectly random output
of size $f ( N ) > N$ with negligible error. See \cite{A14} for a
gentle introduction to this topic and \cite{Beacon} for a discussion
of a possible implementation.

Going further, it has been proved that one can take two copies of such
a randomness expansion protocol (using different pairs of devices for
each copy) and cross-feed them to produce an \textit{arbitrarily
  large} quantity of random bits from a fixed-length seed
\cite{CY13,CSW14, MY14-1}.  Proving that randomness expansion can be
extended in this way is not easy (indeed, the paper containing the
first proof of unbounded expansion \cite{coudron2013arXiv} was 34
pages long).  Here, we give a proof of this extension (based on
\cite{CSW14, MY14-1}) using a graphical language. Explicitly, we prove
that any secure protocol for linear randomness expansion implies a
secure protocol for unbounded randomness expansion.  The proof is
based on formal graphical reasoning and is fairly compact.

A great benefit of graphical proofs is they are amenable to computer
verification. The Globular proof assistant software \cite{BKV16}
carries out category-theoretic proofs, using diagrams that are easily
compatible with graphical languages for quantum mechanics. To further
demonstrate the utility of our work we implemented the proof of
unbounded randomness expansion in Globular. The proof is available as
a video at \cite{globularvideo}.

This paper is organized as follows. In Section~\ref{sec:fundamentals}
we formalize a language to deal with quantum processes, building on
the diagrammatic language given \cite{CK15,
  CK16}. Section~\ref{sec:untrusted} provides the formal basis for
quantum cryptographic protocols that are based on untrusted
devices. In Section~\ref{sec:graphicalexpansion} we give the proof
that linear randomness expansion implies unbounded randomness
expansion, and comment on our use of computer-assisted proof software.
We conclude and discuss future work in Section~\ref{sec:conc}.

The recent paper \cite{Kiss-et-al:2017} also addresses graphical
proofs of quantum cryptography, albeit with a different focus
(device-dependent quantum key distribution). The graphical concepts in
the present paper and in \cite{Kiss-et-al:2017} were developed
independently, and we expect that there will be a useful synergy
between the two papers.

\subsection{Our contributions}

The graphical language used in this paper is based primarily on
\cite{CK15, CK16}. We added new elements to enable quantum
cryptographic proofs.  Here are some of the additions:
\begin{enumerate}
\item \textbf{Diagrams for sets of processes.} Whereas a diagram in
  the language of \cite{CK15, CK16} represents a quantum process, we
  expanded this to allow diagrams with uncharacterized elements to
  represent \textit{sets} of processes.  This is especially useful in
  the device-independent context for expressing security statements
  (see Definition~\ref{def:distancesets}).
\item \textbf{Approximations.}  We use the symbol $=_\epsilon$ when
  the state represented by one diagram is approximately equal to the
  state represented by another.  This allows proofs via chains of
  approximations. A similar feature was independently used in
  \cite{Kiss-et-al:2017}.
\item \textbf{Duplication of states.} We define the ``duplication'' of
  a classical-quantum state (see Section~\ref{ssec:sympurification}),
  a convenient shorthand which subsumes both the copying of classical
  states and the purification of quantum states.
\item \textbf{A graphical formalization of device-independence.} We
  give a graphical formalization of what it means for a protocol to be
  device-independent (Section~\ref{disubsec}).
\end{enumerate}

\section{Fundamentals}
\label{sec:fundamentals}

In this section, we cover useful graphical techniques and concepts,
building on \cite{CK15,CK16}.  First, we describe the graphical
language of categorical quantum mechanics. Next, we introduce a notion
of distance between diagrams. Finally, we define the process of
duplicating states. For further details the reader is encouraged to
consult \cite{CK15,CK16} or the more recent \cite{CKPict} for
categorical quantum mechanics, and \cite{Wat17} for notions of
distance relevant to quantum information theory.

\subsection{The graphical language of categorical quantum mechanics}
\label{ssec:language}

Throughout, $\mathcal{O}$ denotes a collection of finite-dimensional
Hilbert spaces, each with a fixed orthonormal basis.  We use Dirac
notation when it is convenient: for any $V \in \mathcal{O}$, we denote
the fixed orthonormal basis by $\{ \left| 1 \right>_V, \left| 2
\right>_V, \ldots \}$ (with the subscript $V$ typically dropped).  For
any $v \in V$, the expression $\left| v \right>$ is simply another way
of writing the vector $v$, and $\left< v \right|$ denotes the dual of
$v$.

We assume that $\mathcal{O}$ is closed under tensor products (i.e., if
$V,W\in \mathcal{O}$ then $V\otimes W = VW\in \mathcal{O}$) and
contains the space $\mathbb{C}^n$ for every non-negative integer
$n$. We sometimes refer to the elements of $\mathcal{O}$ as
\emph{types} or \emph{registers}. The elements of $\mathcal{O}$ are
the objects of the category in which our graphical reasoning takes
place. We refer the reader to \cite{CKPict} for further details. To
simplify the notation in our diagrams, we sometimes write $N$ to
represent the $N$-bit quantum register $\mathbb{C}^{2^N}$.

We say that $f$ is a \emph{process of type $V\to W$} if $f$ is a
linear operator whose domain is $V$ and whose codomain is $W$. It is
represented by a box labeled $f$ whose input and output wires are
labeled with the types $V$ and $W$ as follows.
\begin{equation}
  \mpp{0.5}{\begin{tikzpicture}
    \node (x) at (0,0) {}; 
    \node (y) at (0,2) {}; 
    \node[style=cprocess] (f) at (0,1) {$f$};
    \draw[style=cedge] (x) -- node[right] {$V$} (f.south) {}; 
    \draw[style=cedge] (y) -- node[right] {$W$} (f.north) {}; 
  \end{tikzpicture}}
\end{equation}
Note that diagrams are read from bottom to top. The identity operator
is a process represented by a box-less diagram (below, left). The
composition and tensor product of processes are respectively
represented by the vertical and horizontal composition of diagrams
(below, center and right).
\begin{equation}
  \mpp{0.5}{
    \begin{tikzpicture}
    \node (x) at (0,1.5) {}; 
    \node (y) at (0,0) {}; 
    \draw[style=cedge] (x) -- node[right] {$V$} (y);
    \end{tikzpicture}
  }
    \qquad~\qquad
  \mpp{0.5}{
    \begin{tikzpicture}
      \node (x) at (0,0) {}; 
      \node (z) at (0,3) {}; 
      \node[style=cprocess] (f) at (0,1) {$f$};
      \node[style=cprocess] (g) at (0,2) {$g$};
      \draw[style=cedge] (x) -- node[right] {$U$} (f.south) {}; 
      \draw[style=cedge] (g.south) -- node[right] {$V$} (f.north) {}; 
      \draw[style=cedge] (z) -- node[right] {$W$} (g.north) {}; 
    \end{tikzpicture}
  }
    \qquad~\qquad
  \mpp{0.5}{
    \begin{tikzpicture}
      \node (x1) at (0,0) {}; 
      \node (y1) at (0,2) {}; 
      \node[style=cprocess] (f1) at (0,1) {$f$};
      \draw[style=cedge] (x1) -- node[right] {$V$} (f1.south) {}; 
      \draw[style=cedge] (y1) -- node[right] {$W$} (f1.north) {}; 
      \node (x2) at (1,0) {}; 
      \node (y2) at (1,2) {}; 
      \node[style=cprocess] (f2) at (1,1) {$f'$};
      \draw[style=cedge] (x2) -- node[right] {$V'$} (f2.south) {}; 
      \draw[style=cedge] (y2) -- node[right] {$W'$} (f2.north) {}; 
    \end{tikzpicture}
  }
\end{equation}
Note that two processes can be vertically composed only if their types
are compatible. A process with no input wires is a \emph{state}. A
state $v$ of type $V$ should be interpreted as a vector in $V$ and is
represented by the first diagram on the left below. The conjugate,
transpose, and conjugate-transpose (i.e., adjoint) of $v$ are also
depicted below (second, third, and fourth diagram respectively).
\begin{equation}
  \mpp{0.5}{\begin{tikzpicture}
    \node (b) at (0.8,1) {}; 
    \node[style=cstate] (a) at (0.8,0) {$v$}; 
    \draw[style=cedge] (b) -- node[right] {$V$} (a);
  \end{tikzpicture}
  \qquad
  \begin{tikzpicture}
    \node (a) at (0.8,1) {}; 
    \node[style=cstatexflip] (b) at (0.8,0) {$v$}; 
    \draw[style=cedge] (a) -- node[right] {$V$} (b);
  \end{tikzpicture}
  \qquad
  \begin{tikzpicture}
    \node (e) at (2.8,0) {}; 
    \node[style=ceffectxflip] (f) at (2.8,0.75) {$v$}; 
    \draw[style=cedge] (f) -- node[right] {$V$} (e);
  \end{tikzpicture}
  \qquad
  \begin{tikzpicture}
    \node (c) at (1.8,0) {}; 
    \node[style=ceffect] (d) at (1.8,0.75) {$v$}; 
    \draw[style=cedge] (c) -- node[right] {$V$} (d);
  \end{tikzpicture}}
\end{equation}
A process with no output wires is called an \emph{effect}. A process
with no input and no output is a \emph{number}, and is represented by
a diamond whose label indicates its value.

Let $W$ be a register. Then the diagram
\begin{equation}
\label{spiderW}
\mpp{.5}{\begin{tikzpicture}
	\node[style=spider] (uniform) at (0,0) {}; 
	\draw[style=cedge] (uniform) to[out=180,in=-90] node[left,near end] {$W$} (-1,1) {}; 
	\draw[style=cedge] (uniform) to[out=0,in=-90] node[right,near end] {$W$}(1,1) {}; 
	\draw[style=cedge] (uniform) to node[right,near end] {$W$} (0,1) {}; 
	\end{tikzpicture}}
\end{equation}
denotes the vector $\sum_{i=1}^{\dim W} \left| i \right> \otimes
\left| i \right> \otimes \left| i \right> \in W \otimes W \otimes W$
and is called a \textit{spider}.  Spiders can have an arbitrary number
of outgoing wires (with all of the wires being of the same type, and
the definition extending in the obvious way).  Because they will play
an important role below, we will also introduce the \emph{uniform
  vector}, which is denoted by a gray node:
\begin{equation}
\label{uniformW}
\mpp{.5}{\begin{tikzpicture}
	\node[style=uniform] (uniform) at (0,0) {}; 
	\draw[style=cedge] (uniform) to[out=180,in=-90] node[left,near end] {$W$} (-1,1) {}; 
	\draw[style=cedge] (uniform) to[out=0,in=-90] node[right,near end] {$W$} (1,1) {}; 
	\draw[style=cedge] (uniform) to node[right,near end] {$W$} (0,1) {}; 
	\end{tikzpicture}}
~~ = ~~
\mpp{.5}{\begin{tikzpicture}
	\node[style=spider] (uniform2) at (4,0) {}; 
	\draw[style=cedge] (uniform2) to[out=180,in=-90] node[left,near end] {$W$} (3,1) {}; 
	\draw[style=cedge] (uniform2) to[out=0,in=-90] node[right,near end] {$W$} (5,1) {}; 
	\draw[style=cedge] (uniform2) to node[right,near end] {$W$} (4,1) {}; 
	\node[style=scalar] (r) at (1.5,0.5) {$m^{-1}$};
	\end{tikzpicture}}
\end{equation}
where $m=\dim W$. The diagram on the left-hand side of
(\ref{uniformW}) denotes the vector $\frac{1}{m} \sum_{i=1}^{m} \left|
i \right> \otimes \left| i \right> \otimes \left| i \right> \in W
\otimes W \otimes W$.

A \emph{quantum} type, or quantum register, is an element $Q$ of
$\mathcal{O}$ of the form $Q=V\otimes V$. To graphically distinguish
them for their classical counterparts, quantum states, effects, and
processes are drawn with thick lines as in the diagrams below.
\begin{equation}
  \mpp{0.5}{\begin{tikzpicture}
    \node (b) at (0.8,1) {}; 
    \node[style=qstate] (a) at (0.8,0) {$\Psi$}; 
    \draw[style=qedge] (b) -- node[right] {$Q$} (a);
  \end{tikzpicture}
  \qquad\qquad
  \begin{tikzpicture}
    \node (c) at (1.8,-0.2) {}; 
    \node[style=qeffect] (d) at (1.8,0.75) {$\beta$}; 
    \draw[style=qedge] (c) -- node[right] {$Q$} (d);
  \end{tikzpicture}
  \qquad\qquad
  \begin{tikzpicture}
    \node (x) at (0,0) {}; 
    \node (y) at (0,2) {}; 
    \node[style=qprocess] (f) at (0,1) {$\Sigma$};
    \draw[style=qedge] (x) -- node[right] {$Q$} (f.south) {}; 
    \draw[style=qedge] (y) -- node[right] {$Q'$} (f.north) {}; 
  \end{tikzpicture}}
\end{equation}
A \textit{pure quantum state} $\Psi$ is a state of the form $v \otimes
\overline{v}$ with $\left\| v \right\| = 1$. Graphically, a quantum
state is pure if it satisfies the diagrammatic equality below.
\begin{equation}
  \mpp{0.5}{\begin{tikzpicture}
    \node (b') at (-3.2,1) {};
    \node[style=qstate] (a') at (-3.2,0) {$\Psi$};
    \draw[style=qedge] (b') -- node[right] {$Q$} (a');
  \end{tikzpicture}}
  ~~ = ~~
  \mpp{0.5}{\begin{tikzpicture}
    \node (b) at (0.8,1) {};
    \node (b2) at (0,1) {};
    \node[style=cstate] (a) at (0.8,0) {$v$};
    \node[style=cstatexflip] (a2) at (0,0) {$v$};
    \draw[style=cedge] (b) -- node[right] {$V$} (a);
    \draw[style=cedge] (b2) -- node[left] {$V$} (a2);
  \end{tikzpicture}}
\end{equation}
A \emph{mixed quantum state} of $Q$ is a state of the form $\sum_i v_i
\otimes \overline{v_i}$ satisfying $\sum_i \left\| v_i \right\|^2
=1$. A \emph{subnormalized mixed state} is only required to satisfy
$\sum_i \left\| v_i \right\|^2 \leq 1$. Unless otherwise specified,
the word ``state'' refers to a normalized mixed state.

The above definition can be related to more conventional notation for
quantum states (see, e.g., \cite{Wat17}) as follows. If
\begin{eqnarray}
\label{opexpression}
\sum_{ij} c_{ij} \left| i \right> \left< j \right|
\end{eqnarray}
is a density operator on the vector space $V$ which represents a
quantum state in the conventional sense, then the corresponding
quantum state in the notation that we are using is given by
\begin{eqnarray}
\label{vecexpression}
\sum_{ij} c_{ij} \left| i \right> \otimes \left| j \right> \in V
\otimes V.
\end{eqnarray}
The main difference between (\ref{opexpression}) and
(\ref{vecexpression}) is that (\ref{opexpression}) is a linear map
from $V$ to $V$, while (\ref{vecexpression}) is simply an element of
$Q = V \otimes V$.

A linear map $Q \to \mathbb{C}$ is a \textit{quantum effect} if it
maps all states to the interval $[0,1]$. In other words, $\beta$ is a
quantum effect if
\begin{equation}
  \mpp{0.5}{\begin{tikzpicture}
    \node[style=qstate] (tri) at (0,0) {$\Psi$};
    \node[style=qeffect] (k1) at (0,1.2) {$\beta$};
    \draw[style=qedge] (tri) --node[right] {$Q$} (k1);    \node (b') at (-3.2,1) {};
  \end{tikzpicture}}
  ~~ \in [0,1]
\end{equation}
for all quantum states $\Psi$ of $Q$. Effects correspond to positive
semidefinite operators on $Q$ with operator norm less than or equal to
one.  The effect $Q \to \mathbb{C}$ given by $\sum_{ij} w_{ij} \left|
ij \right> \mapsto \sum_i w_{ii}$ is denoted by the diagram below
\begin{equation}
  \mpp{0.3}{\begin{tikzpicture}
    \node[style=terminal] (end) at (0,0) {};
    \draw[style=qedge] (end) to node[right] {$Q$} (0,-0.7);
  \end{tikzpicture}}
\end{equation}
and corresponds to taking the trace of a linear operator.  If $V$ is a
classical register then
\begin{equation}
  \mpp{0.3}{\begin{tikzpicture}
    \node[style=terminal] (end) at (0,0) {};
    \draw[style=cedge] (end) to node[right] {$V$} (0,-0.7);
  \end{tikzpicture}}
\end{equation}
denotes the linear map $V \to \mathbb{C}$ given by $\sum_i v_i \left|
i \right> \mapsto \sum_i v_i$.

For use in later proofs, we note the obvious fact that uniform states
absorb terminations, i.e.,
\begin{equation}
  \mpp{0.5}{
    \begin{tikzpicture}
    \node[style=uniform] (seed) at (.8,1.2) {};
    \node[style=terminal] (term) at (.4,2) {};
    \node (nonterm) at (1.2,2.2) {};
    \draw[style=cedge] (seed) to[out=180,in=-90] (term.180);
    \draw[style=cedge] (seed) to[out=0,in=-90] (nonterm.270);
  \end{tikzpicture}}
  ~~ = ~~
  \mpp{0.5}{
  \begin{tikzpicture}
    \node[style=uniform] (seed2) at (2.5,1.2) {};
    \draw[style=cedge] (seed2) to (2.5,2);
  \end{tikzpicture}}
\label{eq:spiders-discard}
\end{equation}

A \textit{causal quantum process} $\Sigma$ from a register $Q$ to a
register $R$ is a linear homomorphism such that for any state $\Psi$
of $Q\otimes Q$, the diagram below represents a state.
\begin{equation}
  \mpp{0.5}{\begin{tikzpicture}
    \node[style=qstate] (tri) at (.4,0) {~~~~$\Psi$~~~~~};
    \node[style=qprocess] (k1) at (0,1.3) {$\Sigma$};
    \draw[style=qedge] (k1) --node[left] {$Q$} (k1|-tri.north);
    \draw[style=qedge] (k1) --node[left] {$R$} (0,2);
	\node (out2) at (1,2) {};
    \draw[style=qedge] (out2) to node[right] {$Q$} (out2|-tri.north);
  \end{tikzpicture}}
\end{equation}
A \textit{stochastic} quantum process is one that satisfies the same
condition, with the weaker requirement that the diagram above is a
subnormalized quantum state.  Intuitively, a causal process is a
process that always gives an outcome, while a stochastic process is a
process that may sometimes ``fail'' and give no outcome (hence the
trace of its output state may be smaller than the trace of its input
state).  Unless otherwise specified, the phrase ``quantum process''
refers to a stochastic quantum process.

Diagrammatically, a process is causal if and only if the process of
applying $T$ and then discarding its output wires is equal to the
process of merely discarding the input wires, e.g.,
\begin{equation}
  \mpp{0.5}{
  \begin{tikzpicture}
    \node[style=qprocess] (corrbox) at (-2,0) {~~~~$T$~~~~};
    \node (in1) at (-2.5,-.75) {};
    \node (in2) at (-2,-.75) {};
    \node (in3) at (-1.5,-.75) {};
    \node[style=terminal] (term1) at (-2,.6) {};
    \draw[style=cedge] (in1) to (in1|-corrbox.south);
    \draw[style=cedge] (in2) to (in2|-corrbox.south);
    \draw[style=qedge] (in3) to (in3|-corrbox.south);
    \draw[style=qedge] (term1) to (term1|-corrbox.north);
  \end{tikzpicture}}
  ~~ = ~~
  \mpp{0.5}{
  \begin{tikzpicture}
    \draw[style=cedge] (0.7,-.75) to[in=270, out=90, looseness=1] (0.7,.5);
    \draw[style=cedge] (1.5,-.75) to[in=270, out=90, looseness=1] (1.5,.5);
    \draw[style=qedge] (2.3,-.75) to[in=270, out=90, looseness=1] (2.31,.5);
    \node[style=terminal] at (0.7,.6) {};
    \node[style=terminal] at (1.5,.6) {};
    \node[style=terminal] at (2.3,.6) {};
  \end{tikzpicture}}
  \label{eq:causal}
\end{equation}
Causal quantum processes correspond to completely positive
trace-preserving maps in the conventional notation.

A \textit{pure process} is a quantum process of the form
\begin{equation}
  \mpp{0.5}{\begin{tikzpicture}
    \node (b') at (-3.2,1) {};
    \node (c') at (-3.2,-1) {};
    \node[style=qprocess] (a') at (-3.2,0) {$\Psi$};
    \draw[style=qedge] (b') -- (a');
    \draw[style=qedge] (c') -- (a');
  \end{tikzpicture}}
  ~~ = ~~
  \mpp{0.5}{\begin{tikzpicture}
    \node (b) at (0.8,1) {};
    \node (b2) at (0,1) {};
    \node (c) at (0.8,-1) {};
    \node (c2) at (0,-1) {};
    \node[style=cprocess] (a) at (0.8,0) {$\psi$};
    \node[style=cprocessxflip] (a2) at (0,0) {$\psi$};
    \draw[style=cedge] (b) --  (a);
    \draw[style=cedge] (b2) --  (a2);
    \draw[style=cedge] (c) --   (a);
    \draw[style=cedge] (c2) --   (a2);
  \end{tikzpicture}}
\end{equation}
Note that pure processes map pure subnormalized states to pure
subnormalized states.

If $\Psi$ is a state of $QR$, then we can terminate one of its wires
and obtain a state of $Q$:
\begin{equation}
  \mpp{0.5}{\begin{tikzpicture}
    \node[style=qstate] (tri) at (0,0) {~~~~$\Psi$~~~~~};
    \node (k1) at (-0.25,1.3) {};
    \draw[style=qedge] (k1) --node[left] {$Q$} (k1|-tri.north);
    \node[style=terminal] (end) at (0.5,1) {};
    \draw[style=qedge] (end) to node[right] {$R$} (end|-tri.north);
  \end{tikzpicture}}
\end{equation}
(This corresponds to taking a partial trace in the conventional
framework.)

If $C$ is a register, then a \textit{classical} state is a vector $v
\in C$ of the form $v = \sum_i p_i \left| i \right>$, where $\{ p_i
\}_i$ is a probability distribution.  If $Q$ is a quantum register,
then a \textit{classical-quantum} state of $CQ$ is a state of the form
$\sum_i p_i \left| i \right> \otimes v_i$, where $\{ p_i \}_i$ is a
probability distribution and each $v_i$ is a normalized state of $Q$.
A quantum process on $CQ$ that is \textit{controlled} by the classical
register $C$ is a process of the form $\sum_i \left| i \right> \left<
i \right| \otimes \Phi_i$, where each $\Phi_i$ is a process on $Q$.

In our context, it will be useful to represent \textit{sets} of linear
maps diagrammatically. A diagram in which every element is specified
by an explicit linear map can itself be regarded a linear map
(obtained by composition and tensor product). A diagram in which some
elements are unspecified represents the set of all specifications of
that form. For example, the following diagrams
\begin{equation}
  \mpp{0.5}{\begin{tikzpicture}
    \node (b') at (-3.2,1) {};
    \node[style=qstate] (a') at (-3.2,0) {};
    \draw[style=qedge] (b') to node[right] {$Q$} (a');
  \end{tikzpicture}}
  \qquad\qquad
  \mpp{0.5}{\begin{tikzpicture}
    \node (b') at (-3.2,1) {};
    \node (c') at (-3.2,-1) {};
    \node[style=qprocess] (a') at (-3.2,0) {};
    \draw[style=qedge] (b') to node[right] {$Q$} (a');
    \draw[style=qedge] (c') to node[right] {$R$} (a');
  \end{tikzpicture}}
  \qquad\qquad
  \mpp{0.5}{\begin{tikzpicture}
    \node (a) at (-.3,1) {};
    \node (b) at (.6,1) {};
    \node[style=qstate] (q) at (0,0) {~~~~~~~~~};
    \draw[style=qedge] (b) to node[right] {} (b|-q.north);
    \draw[style=cedge] (a) to node[left] {$\mathbb{C}^2$} (a|-q.north);
  \end{tikzpicture}}
\end{equation}
represent, respectively, the set of all (normalized, mixed) states on
$Q$, the set of all stochastic quantum processes $R\to Q$, and the set
of all classical-quantum states in which the classical register is
$\mathbb{C}^2$ and the quantum register can be arbitrary.  Notice that
when a wire is unlabeled and otherwise unspecified, the
quantification ranges over all (quantum or classical)
registers. Similarly, the diagram
$\mpp{0.3}{\begin{tikzpicture}\node[style=scalar] (sc) at (0,0)
    {}; \end{tikzpicture}}$ represents an arbitrary element of the
interval $[0,1]$.

\subsection{Approximations}
\label{ssec:approximations}

In what follows, we will need to be able to discuss the
\emph{distance} between certain processes. We therefore define a
relation between diagrams which captures the appropriate metric (half
of the diamond norm distance --- see \cite{Wat17} for further
details).

\begin{definition}
\label{def:distanceprocesses}
If $c$, $d$, and $\epsilon$ are real numbers, we write
$c=_{\epsilon}d$ if $\left| c - d \right| \leq \epsilon$. Let $\Sigma$
and $\Sigma'$ be two processes of the same type. Then, we write
$\Sigma =_\epsilon \Sigma'$ if for all states $\Psi$ and effects
$\beta$,
\begin{equation}
  \mpp{0.5}{\scalebox{.8}{
    \begin{tikzpicture}
      \node[style=qeffect] (up) at (0,2) {~~~~~$\beta$~~~~};
      \node[style=qstate] (tri) at (0,0) {~~~~$\Psi$~~~~~};
      \node[style=qprocess] (k1) at (-.25,1) {$\Sigma$};
      \draw[style=qedge] (k1) to (k1|-up.south);
      \draw[style=qedge] (k1) to (k1|-tri.north);
      \draw[style=qedge] (tri.20) to (tri.20|-up.south);
    \end{tikzpicture}
  }}
  ~~\displaystyle =_\epsilon~~
   \mpp{0.5}{\scalebox{.8}{
    \begin{tikzpicture}
      \node[style=qeffect] (up2) at (4,2) {~~~~~$\beta$~~~~};
      \node[style=qstate] (tri2) at (4,0) {~~~~$\Psi$~~~~~};
      \node[style=qprocess] (k3) at (3.75,1) {$\Sigma'$};
      \draw[style=qedge] (k3) to (k3|-tri2.north);
      \draw[style=qedge] (k3) to (k3|-up2.south);
      \draw[style=qedge] (tri2.20) to (tri2.20|-up2.south);
    \end{tikzpicture}
  }}
\end{equation}
\end{definition}
Note that the above definition remains equivalent if the phrase ``for
all states $\Psi$'' is replaced by ``for all pure states $\Psi$.''

It follows from the definition that the notion of approximation
defined above satisfies the triangle inequality (if $\Sigma =_\epsilon
\Sigma'$ and $\Sigma' =_\delta \Sigma''$, then $\Sigma =_{\epsilon +
  \delta} \Sigma''$) and that it is preserved by composition (if
$\Sigma =_\epsilon \Sigma'$, then $\Theta \circ \Sigma =_\epsilon
\Theta \circ \Sigma'$ and $\Sigma \circ \Lambda = \Sigma' \circ
\Lambda$ for all stochastic processes $\Theta, \Lambda$ of appropriate
input and output type).  If we restrict the processes $\Sigma$ and
$\Sigma'$ to be states (i.e., to have no inputs) then $\Sigma
=_{\epsilon}\Sigma'$ if and only if $\Sigma$ and $\Sigma'$ differ by
no more than $2\epsilon$ in trace distance.

We now generalize the notion of distance between processes to a notion
of distance between \emph{sets} of processes. A similar notion of
approximation appears in the non-graphical formalism of
\cite{Maurer11abstractcryptography}.

\begin{definition}
\label{def:distancesets}
Let $A$ and $B$ be two sets of processes, all of which have the same
type. We write $A \subseteq_\epsilon B$ if for every $a \in A$, there
exists $b \in B$ such that $a =_\epsilon b$. Moreover, we write
$A=_\epsilon B$ if $B \subseteq_\epsilon A$ and $A \subseteq_\epsilon
B$.
\end{definition}

These relations also satisfy a triangle inequality: if $A
\subseteq_\epsilon B$ and $B \subseteq_\delta C$, then $A
\subseteq_{\epsilon + \delta} C$. Note that the symbol $=_{\epsilon}$
is used to denote a relation between numbers, a relation between
processes, and a relation between sets of processes. However, it will
always be clear from the context whether numbers, processes, or sets
of processes are being compared so that no ambiguity should arise from
this slight abuse of notation.

\subsection{Duplication}
\label{ssec:sympurification}

We now introduce the notion of \textit{duplicate} states.  Informally
speaking, duplicating a classical-quantum state means copying its
classical component and purifying its quantum component. See section
2.5 of \cite{NC02} for a discussion of the notion of quantum state
purification.

If $\Psi$ is a subnormalized classical-quantum state of a register
$CQ$, then we can write $\Psi = \sum_{ijk} \psi_{ij}^k \left| k
\right> \otimes \left| i \right> \otimes \left| j \right> \in C
\otimes V \otimes V$, where $Q = V \otimes V$.  The matrices $M_k := [
  \psi^k_{ij} ]_{ij}$ are then positive semidefinite.  We can
alternatively express $\Psi$ as
\begin{equation}
\mpp{0.5}{\scalebox{.8}{
		\begin{tikzpicture}
		\node (out1) at (-.5,1.5) {};
		\node (out2) at (.75,1.5) {};
		\node[style=qstate] (psi) at (0, 0) {~~~~~$\Psi$~~~~};
		\draw[style=cedge] (out1) to node[left,near start] {$C$} (out1|-psi.north);
		\draw[style=qedge] (out2) to node[left,near start] {$Q$} (out2|-psi.north);
		\end{tikzpicture}}
}
~~ = ~~ 
\mpp{0.5}{\scalebox{.8}{
		\begin{tikzpicture}
		\node (out1) at (-.5,1) {};
		\node (out2) at (.75,1) {};
		\node[style=qprocess] (psi) at (0, 0) {~~~~~~$P$~~~~~~~};
		\draw[style=cedge] (out1) to node[left,near start] {$C$} (out1|-psi.north);
		\draw[style=qedge] (out2) to node[right,near start] {$Q$} (out2|-psi.north);
		\node[style=spider] (csource) at (-0.5,-1) {};
		\draw[style=cedge] (csource) to node[left] {$C$} (csource|-psi.south);
		\node[style=spider] (qsource) at (0.75, -1) {};
		\draw[style=qedge] (qsource) to node[right] {$Q$} (qsource|-psi.south);
		\end{tikzpicture}}}
\end{equation}
where $P$ is the linear map defined by
\begin{equation}
P ( \left| k \right> \otimes \left| i \right> \otimes \left| j \right>
) = \left| k \right> \otimes \sqrt{M_k} \left| i \right> \otimes
\overline{\sqrt{M_k}} \left| j \right>.
\end{equation}
Then, the \textit{canonical duplicate state} of $\Psi$ is given by
\begin{equation}
\mpp{0.4}{\scalebox{0.8}{
		\begin{tikzpicture}
		\node (out1) at (-1,1.5) {};
		\node (out2) at (-.3,1.5) {};
		\node (out3) at (.3,1.5) {};
		\node (out4) at (1,1.5) {};
		\node[style=qstate] (psi) at (-.25, 0) {~~~~~$\Psi'$~~~~~~};
		\draw[style=cedge] (out1) to node[left, near start] {$C$} (out1|-psi.north);
		\draw[style=qedge] (out2) to node[left, near start] {$Q$} (out2|-psi.north);
		\draw[style=qedge] (out3) to node[right, near start] {$Q$} (out3|-psi.north);
		\draw[style=cedge] (out4) to node[right, near start] {$C$} (out4|-psi.north);
		\end{tikzpicture}}}
~~ := ~~  
\mpp{0.5}{\scalebox{0.8}{
		\begin{tikzpicture}
		\node (out1) at (-.5,1.5) {};
		\node (out2) at (.5,1.5) {};
		\node (out3) at (1.5,1.5) {};
		\node (out4) at (2.5,1.5) {};
		\node[style=qprocess] (psi) at (0, 0) {~~~~~$P$~~~~~~};
		\draw[style=cedge] (out1) to node[left,near start] {$C$} (out1|-psi.north);
		\draw[style=qedge] (out2) to node[left,near start] {$Q$} (out2|-psi.north);
		\draw[style=qedge] (out3) to[out=-90,in=90] node[right,very near start] {$Q$} (1.5,0) to[out=-90,in=0] (1,-.7) to [out=180,in=-90] (out2|-psi.south);
		\draw[style=cedge] (out4) to[out=-90,in=90] node[right,very near start] {$C$} (out4|-psi.south) to[out=-90,in=0] (.75,-1.2) to [out=180,in=-90] (out1|-psi.south);
		\end{tikzpicture}}}
\end{equation}
See equation (\ref{sympure}) in the appendix for an expression for
this state in conventional notation. Note that
\begin{equation}
\label{thecopiedstate}
\mpp{0.5}{\scalebox{0.8}{
		\begin{tikzpicture}
		\node[style=qstate] (psi) at (-.25, 0) {~~~~~~$\Psi'$~~~~~~~};
		\node (out1) at (-1,2) {};
		\node (out2) at (-.3,2) {};
		\draw[style=cedge] (out1) to node[left,near start] {$C$} (out1|-psi.north);
		\draw[style=qedge] (out2) to node[left,near start] {$Q$} (out2|-psi.north);
		\node[style=terminal] (end1) at (0.3,1.5) {};
		\node[style=terminal] (end2) at (1,1.5) {};
		\draw[style=qedge] (end1) to node[right] {$Q$} (end1|-psi.north);
		\draw[style=cedge] (end2) to node[right] {$C$} (end2|-psi.north);
		\end{tikzpicture}}}
~~= ~~
\mpp{0.5}{\scalebox{0.8}{
		\begin{tikzpicture}
		\node[style=qstate] (psi) at (-.25, 0) {~~~~~$\Psi$~~~~~~};
		\node (out1) at (-.5,1.5) {};
		\node (out2) at (.5,1.5) {};
		\draw[style=cedge] (out1) to node[left,near start] {$C$} (out1|-psi.north);
		\draw[style=qedge] (out2) to node[right,near start] {$Q$} (out2|-psi.north);
		\end{tikzpicture}}}
\end{equation}

More generally, a state $\Psi''$ of $CQQC$ is a \textit{duplicate
  state} (or simply \textit{duplication}) of $\Psi$ if
(\ref{thecopiedstate}) holds (with $\Psi'$ replaced by $\Psi''$) and
$\Psi''$ has the form $\Psi'' = \sum_i \left| i \right> \otimes \psi_i
\otimes \left| i \right>$, where each $\psi_i$ is a pure subnormalized
state of $QQ$.  Note that if $\psi = \sum_i p_i \left| i \right>$ is a
classical state (with no quantum component), then there is only one
duplicate of $\psi$, and that is the ``copied'' state $\sum_i p_i
\left| i \right> \otimes \left| i \right>$.

Duplicate states have the following universality property: for any
subnormalized classical-quantum state $\Phi$ of $CQRD$, where $R$ is a
quantum register and $D$ is a classical register, such that
\begin{equation}
\mpp{0.5}{\scalebox{0.8}{
		\begin{tikzpicture}
		\node[style=qstate] (psi) at (-.25, 0) {~~~~~~$\Phi$~~~~~~~};
		\node (out1) at (-1,2) {};
		\node (out2) at (-.3,2) {};
		\draw[style=cedge] (out1) to node[left] {$C$} (out1|-psi.north);
		\draw[style=qedge] (out2) to node[left] {$Q$} (out2|-psi.north);
		\node[style=terminal] (end1) at (0.3,1.5) {};
		\node[style=terminal] (end2) at (1,1.5) {};
		\draw[style=qedge] (end1) to node[right] {$R$} (end1|-psi.north);
		\draw[style=cedge] (end2) to node[right] {$D$} (end2|-psi.north);
		\end{tikzpicture}}}
~~ = ~~
\mpp{0.5}{\scalebox{0.8}{
		\begin{tikzpicture}
		\node[style=qstate] (psi) at (-.25, 0) {~~~~~$\Psi$~~~~~~};
		\node (out1) at (-.5,2) {};
		\node (out2) at (.5,2) {};
		\draw[style=cedge] (out1) to node[left] {$C$} (out1|-psi.north);
		\draw[style=qedge] (out2) to node[right] {$Q$} (out2|-psi.north);
		\end{tikzpicture}}},
\end{equation}
there exists a causal process $\alpha$ from $QC$ to $RD$ satisfying
\begin{equation}
\mpp{0.5}{\scalebox{0.8}{
		\begin{tikzpicture}
		\node[style=qstate] (psi) at (-.25, 0) {~~~~~~$\Phi$~~~~~~~};
		\node (out1) at (-1,2) {};
		\node (out2) at (-.3,2) {};
		\draw[style=cedge] (out1) to node[left] {$C$} (out1|-psi.north);
		\draw[style=qedge] (out2) to node[left] {$Q$} (out2|-psi.north);
		\node (end1) at (0.3,2) {};
		\node (end2) at (1,2) {};
		\draw[style=qedge] (end1) to node[right] {$R$} (end1|-psi.north);
		\draw[style=cedge] (end2) to node[right] {$D$} (end2|-psi.north);
		\end{tikzpicture}}}
~~ = ~~
\mpp{0.5}{\scalebox{0.8}{
		\begin{tikzpicture}
		\node[style=qstate] (psi) at (3.75, 0) {~~~~~~$\Psi'$~~~~~~~};
		\node (out3) at (3,2) {};
		\node (out4) at (3.7,2) {};
		\draw[style=cedge] (out3) to node[left] {$C$} (out3|-psi.north);
		\draw[style=qedge] (out4) to node[left] {$Q$} (out4|-psi.north);
		\node (end3) at (4.3,2) {};
		\node (end4) at (5,2) {};
		\node[style=qprocess] (alpha) at (4.65, 1) {~~~~$\alpha$~~~};
		\draw[style=qedge] (end3) to node[right] {$R$} (end3|-alpha.north);
		\draw[style=cedge] (end4) to node[right] {$D$} (end4|-alpha.north);
		\draw[style=qedge] (end3|-alpha.south) to node[right] {$\scriptstyle Q$} (end3|-psi.north);
		\draw[style=cedge] (end4|-alpha.south) to node[right] {$\scriptstyle C$} (end4|-psi.north);
		\end{tikzpicture}}}
\end{equation}

Additionally, if $\Psi'$ and $\Psi''$ are any two states of $CQQC$
that are both duplicate states of $\Psi$, then there exists a unitary
operator $U$ on $QC$, controlled by the register $C$, such that
\begin{equation}
\label{puruniversal}
\mpp{0.5}{\scalebox{0.8}{
		\begin{tikzpicture}
		\node[style=qstate] (psi) at (-.25, 0) {~~~~~~$\Psi'$~~~~~~~};
		\node (out1) at (-1,2) {};
		\node (out2) at (-.3,2) {};
		\draw[style=cedge] (out1) to node[left] {$C$} (out1|-psi.north);
		\draw[style=qedge] (out2) to node[left] {$Q$} (out2|-psi.north);
		\node (end1) at (0.3,2) {};
		\node (end2) at (1,2) {};
		\draw[style=qedge] (end1) to node[right] {$Q$} (end1|-psi.north);
		\draw[style=cedge] (end2) to node[right] {$C$} (end2|-psi.north);
		\end{tikzpicture}}}
~~=~~
\mpp{0.5}{\scalebox{0.8}{
		\begin{tikzpicture}
		\node[style=qstate] (psi) at (3.75, 0) {~~~~~~$\Psi''$~~~~~~~};
		\node (out3) at (3,2) {};
		\node (out4) at (3.7,2) {};
		\draw[style=cedge] (out3) to node[left] {$C$} (out3|-psi.north);
		\draw[style=qedge] (out4) to node[left] {$Q$} (out4|-psi.north);
		\node (end3) at (4.3,2) {};
		\node (end4) at (5,2) {};
		\node[style=qprocess] (alpha) at (4.65, 1) {~~~~$U$~~~};
		\draw[style=qedge] (end3) to node[right] {$Q$} (end3|-alpha.north);
		\draw[style=cedge] (end4) to node[right] {$C$} (end4|-alpha.north);
		\draw[style=qedge] (end3|-alpha.south) to node[right] {$\scriptstyle Q$} (end3|-psi.north);
		\draw[style=cedge] (end4|-alpha.south) to node[right] {$\scriptstyle C$} (end4|-psi.north);
		\end{tikzpicture}}}
\end{equation}

The following proposition follows from standard techniques and is
proved in Appendix~\ref{approxpropsec}.

\begin{proposition}
  \label{prop:approxprop}
  If $\Psi, \Phi$ are subnormalized states that satisfy $\Psi
  =_\epsilon \Phi$, and $\Psi', \Phi'$ denote their respective
  canonical duplicate states, then $\Psi' =_{\sqrt{2 \epsilon}}
  \Phi'$.  $\Box$
\end{proposition}

The next corollary will be useful in later proofs.

\begin{corollary}
\label{cor:dup}
Suppose that
\begin{equation}
\label{phieq}
\mpp{0.5}{\scalebox{0.8}{
		\begin{tikzpicture}
		\node[style=qstate] (psi) at (-.25, 0) {~~~~~~$\Phi$~~~~~~~};
		\node (out1) at (-1,2) {};
		\node (out2) at (-.3,2) {};
		\draw[style=cedge] (out1) to node[left] {$C$} (out1|-psi.north);
		\draw[style=qedge] (out2) to node[left] {$Q$} (out2|-psi.north);
		\node[style=terminal] (end1) at (0.3,1.5) {};
		\node[style=terminal] (end2) at (1,1.5) {};
		\draw[style=qedge] (end1) to node[right] {$R$} (end1|-psi.north);
		\draw[style=cedge] (end2) to node[right] {$D$} (end2|-psi.north);
		\end{tikzpicture}}}
~~ =_\epsilon ~~
\mpp{0.5}{\scalebox{0.8}{
		\begin{tikzpicture}
		\node[style=qstate] (psi) at (-.25, 0) {~~~~~$\Psi$~~~~~~};
		\node (out1) at (-.5,2) {};
		\node (out2) at (.5,2) {};
		\draw[style=cedge] (out1) to node[left] {$C$} (out1|-psi.north);
		\draw[style=qedge] (out2) to node[right] {$Q$} (out2|-psi.north);
		\end{tikzpicture}}}
\end{equation}
and let $\Psi''$ be a state on $CQQC$ that is a duplicate of $\Psi$.
Then, there exists a causal process $\alpha$ from $QC$ to $RD$
satisfying
\begin{equation}
\label{alphaconclusion}
\mpp{0.5}{\scalebox{0.8}{
		\begin{tikzpicture}
		\node[style=qstate] (psi) at (-.75, 0) {~~~~~~$\Phi$~~~~~~~};
		\node (out1) at (-1.5,2) {};
		\node (out2) at (-.8,2) {};
		\draw[style=cedge] (out1) to node[left] {$C$} (out1|-psi.north);
		\draw[style=qedge] (out2) to node[left] {$Q$} (out2|-psi.north);
		\node (end1) at (-.2,2) {};
		\node (end2) at (.5,2) {};
		\draw[style=qedge] (end1) to node[right] {$R$} (end1|-psi.north);
		\draw[style=cedge] (end2) to node[right] {$D$} (end2|-psi.north);
		\end{tikzpicture}}}
~~=_{\sqrt{2\epsilon}}~~
\mpp{0.5}{\scalebox{0.8}{
		\begin{tikzpicture}
		\node[style=qstate] (psi) at (3.75, 0) {~~~~~~$\Psi''$~~~~~~~};
		\node (out3) at (3,2) {};
		\node (out4) at (3.7,2) {};
		\draw[style=cedge] (out3) to node[left] {$C$} (out3|-psi.north);
		\draw[style=qedge] (out4) to node[left] {$Q$} (out4|-psi.north);
		\node (end3) at (4.3,2) {};
		\node (end4) at (5,2) {};
		\node[style=qprocess] (alpha) at (4.65, 1) {~~~~$\alpha$~~~};
		\draw[style=qedge] (end3) to node[right] {$R$} (end3|-alpha.north);
		\draw[style=cedge] (end4) to node[right] {$D$} (end4|-alpha.north);
		\draw[style=qedge] (end3|-alpha.south) to node[right] {$\scriptstyle Q$} (end3|-psi.north);
		\draw[style=cedge] (end4|-alpha.south) to node[right] {$\scriptstyle C$} (end4|-psi.north);
		\end{tikzpicture}}}
\end{equation}
\end{corollary}

\begin{proof}
Let $\Psi'$ denote the canonical duplicate state of $\Psi$.  By the
property noted in equation (\ref{puruniversal}) above, it suffices to
prove the desired relation (\ref{alphaconclusion}) with $\Psi''$
replaced by $\Psi'$.  By Proposition~\ref{prop:approxprop}, the
canonical duplicate state $\Sigma$ of the state on the left side of
(\ref{phieq}) satisfies $\Sigma =_{\sqrt{2 \epsilon}} \Psi'$.  There
is a causal process $\alpha$ from $QC$ to $RD$ such that
\begin{equation}
\mpp{0.5}{\scalebox{0.8}{
		\begin{tikzpicture}
		\node[style=qstate] (psi) at (-.75, 0) {~~~~~~$\Phi$~~~~~~~};
		\node (out1) at (-1.5,2) {};
		\node (out2) at (-.8,2) {};
		\draw[style=cedge] (out1) to node[left,near start] {$C$} (out1|-psi.north);
		\draw[style=qedge] (out2) to node[left,near start] {$Q$} (out2|-psi.north);
		\node (end1) at (-.2,2) {};
		\node (end2) at (.5,2) {};
		\draw[style=qedge] (end1) to node[right,near start] {$R$} (end1|-psi.north);
		\draw[style=cedge] (end2) to node[right,near start] {$D$} (end2|-psi.north);
		\end{tikzpicture}}}
~~ = ~~
\mpp{0.5}{\scalebox{0.8}{
		\begin{tikzpicture}
		\node[style=qstate] (psi) at (3.75, 0) {~~~~~~$\Sigma$~~~~~~~};
		\node (out3) at (3,2) {};
		\node (out4) at (3.7,2) {};
		\draw[style=cedge] (out3) to node[left,near start] {$C$} (out3|-psi.north);
		\draw[style=qedge] (out4) to node[left,near start] {$Q$} (out4|-psi.north);
		\node (end3) at (4.3,2) {};
		\node (end4) at (5,2) {};
		\node[style=qprocess] (alpha) at (4.65, 1) {~~~~$\alpha$~~~};
		\draw[style=qedge] (end3) to node[right] {$R$} (end3|-alpha.north);
		\draw[style=cedge] (end4) to node[right] {$D$} (end4|-alpha.north);
		\draw[style=qedge] (end3|-alpha.south) to node[right] {$\scriptstyle Q$} (end3|-psi.north);
		\draw[style=cedge] (end4|-alpha.south) to node[right] {$\scriptstyle C$} (end4|-psi.north);
		\end{tikzpicture}}}
\end{equation}
Applying the same process to $\Psi'$ yields a state that is within
distance $\sqrt{2 \epsilon}$ from $\Phi$.
\end{proof}

\section{Untrusted quantum processes}
\label{sec:untrusted}

An \textit{untrusted quantum process} is represented by a diagram in
which all of the quantum processes are unlabeled and all of the
classical processes are labeled.  We discuss an example of such
processes and then give a definition of the more specific class of
\textit{device-independent quantum protocols}.

\subsection{Example: Quantum strategies for nonlocal games}

A nonlocal game is a game played by $k$ parties ($k \geq 2$) in which
the players are given random inputs $X_1, \ldots, X_k$ according to
some fixed joint probability distribution. The players produce outputs
$A_1, \ldots, A_k$ and these outputs are scored as $L ( X_1, \ldots,
X_k, A_1 , \ldots, A_k )$, where $L$ is a deterministic function that
maps to $\{ 0, 1 \}$.  An example (the Clauser-Horne-Shimony-Holt
game) is given below.  The registers $X, A, B, Y$ are classical bit
registers (each isomorphic to $\mathbb{C}^2$).
\begin{equation}
  \mpp{0.5}{\begin{tikzpicture}
    \node[style=uniform] (x) at (-1,2) {};
    \node[style=uniform] (y) at (6,2) {};
    \draw[style=cedge] (0,2) to[in=-90, out=-90, looseness=2]  (x);
    \draw[style=cedge] (5,2) to[in=-90, out=-90, looseness=2] (y) {};
    \node[style=qstate] (qsys) at (2.5,1.5) {};
    \node[style=qprocess] (alicem) at (0,2.3) {};
    \node[style=qprocess] (bobm) at (5,2.3) {};
    \draw[style=cedge] (0,2) to (alicem);
    \draw[style=cedge] (5,2) to  (bobm);
    \draw[style=qedge] (qsys.100) to[out=90, in=-90] (alicem);
    \draw[style=qedge] (qsys.80) to[out=90, in=-90] (bobm);
    \draw[style=cedge] (alicem) to node[left] {$A$}  (0,3);
    \draw[style=cedge] (bobm) to node[right] {$B$} (5,3);
    \draw[style=cedge] (x) to node[left] {$X$} (-1,3.5);
    \draw[style=cedge] (y) to  node[right] {$Y$} (6,3.5);
    \node[style=ceffect] (score) at (2.5,4.2) {\tiny $A \oplus B = X \wedge Y$};
    \draw[style=cedge] (-1,3.5) to[out=90, in=-90, looseness=0.75] (score.220);
    \draw[style=cedge] (0,3) to[out=90, in=-90] (score.260);
    \draw[style=cedge] (5,3) to[out=90, in=-90] (score.280);
    \draw[style=cedge] (6,3.5) to[out=90, in=-90, looseness=0.75] (score.320);
  \end{tikzpicture}}
\end{equation}
The effect at the top denotes the map $\mathbb{C}^{\{0, 1 \}^4} \to
\mathbb{C}$ given by $(p_{xaby}) \mapsto \sum_{a \oplus b = x \wedge
  y} p_{xaby}$.  This game is a common building block for
device-independent protocols (including in particular the randomness
expansion results that we will consider in
section~\ref{sec:graphicalexpansion}).

\subsection{Device-independent quantum protocols}
\label{disubsec}

We are now ready to formalize the notion of a protocol in the
device-independent setting.  Historically, a quantum protocol is
\textit{device-independent} if all of its \textit{quantum} processes
are untrusted and uncharacterized (whereas strictly ``classical''
aspects, such as timing, non-communication, and computation, are still
trusted).  This definition can be traced back to early papers such as
Mayers and Yao \cite{mayersyao} and Ekert \cite{Ekert91}.  There is
some room for interpretation as to exactly which quantum processes are
allowed in device-independence, and we offer a specific formalism
here.  (Our treatment can be compared to the non-graphical
formalization of device-independent protocols in section 4 of
\cite{CSW14}.)

For simplicity our definition is for a $2$-device protocol, but it
could easily be generalized to an $N$-device protocol.

\begin{definition}
  A device-independent protocol with $2$ quantum devices is a diagram
  of the form
  \begin{equation}
  \mpp{0.5}{\begin{tikzpicture}
	\node (seed) at (-1,0) {};
	\node[style=qprocess] (re) at (0.15,1.5) {~~~~~~~$S$~~~~~~~~};
	\draw[style=qedge] (0.2,0) to[in=270, out=90] node[right] {$Q_1$}  (re.280);
	\draw[style=qedge] (1.5,0) to[in=270, out=90] node[right] { $Q_2$}  (re.340);
	\draw[style=cedge] (seed) to[in=270, out=90]  node[left] {$C$} (re.210);
	\draw[style=cedge] (re.150) to[in=270, out=90] node[left] {$C$}  (-1,3.3);
	\node (end) at (0.3,3.3) {};
	\node (end2) at (1.5,3.3) {};
	\draw[style=qedge] (re.80) to[in=270, out=90] node[right] {$Q_1$} (end.180);
	\draw[style=qedge] (re.20) to[in=270, out=90] node[right] {$Q_2$} (end2.180);
	\end{tikzpicture}}
  \end{equation}
  where $S$ is constructed from the following subdiagrams.
  \begin{enumerate}
  \item \textbf{Communication between devices.}  An untrusted process
    transferring information (one way) from one of the two quantum
    registers to the other:
    \begin{equation}
    \mpp{0.5}{\scalebox{0.8}{\begin{tikzpicture}
		\node (starti) at (0,0) {};
		\node (startj) at (3,0) {};
		\node (endi) at (0,3) {};
		\node (endj) at (3,3) {};
		\node[style=qprocess] (proci) at (0,1) {~~};
		\node[style=qprocess] (procj) at (2.9,2) {~~};
		\draw[style=qedge] (starti) to node[left,near start] {$Q_i$} (starti|-proci.south);
		\draw[style=qedge] (startj) to node[right,near start] {$Q_j$} (startj|-procj.south);
		\draw[style=qedge] (proci.45) to[in=-90, out=90] (procj.-120);
		\draw[style=qedge] (endi) to node[left,near start] {$Q_i$} (endi|-proci.north);
		\draw[style=qedge] (endj) to node[right,near start] {$Q_j$} (endj|-procj.north);
		\end{tikzpicture}}}
    \end{equation}
  \item \textbf{Deterministic classical functions.}  A deterministic
    function is applied to the register $C$.
  \begin{equation}
  \mpp{0.5}{\scalebox{0.8}{\begin{tikzpicture}
                  \node (startc) at (0,0) {};
                  \node (startj) at (3,0) {};
                  \node (endc) at (0,2) {};
                  \node (endj) at (3,2) {};
                  \node[style=cprocess] (procc) at (0,1) {$F$};
                  \draw[style=cedge] (startc) to[in=-90, out=90] node[right] {$C$} (procc);
                  \draw[style=cedge] (procc) to node[right] {$C$} (endc);
                  \end{tikzpicture}}}
  \end{equation}
  \item \textbf{Failure.} The value of the classical register $C$ is
    checked to see if it lies in a chosen subset $S$; if it does not,
    the protocol aborts.  (Diagrammatically, this is the linear map
    from $C$ to $C$ given by $\sum_i p_i \left| i \right> \mapsto
    \sum_{i \in S} p_i \left| i \right>$.)
  \item \textbf{Giving input to a device.} A deterministic function is
    applied to $C$ and the result is given to one of the devices.
  \begin{equation}
  \mpp{0.5}{\scalebox{0.8}{\begin{tikzpicture}
                  \node (starti) at (0,0) {};
                  \node (startj) at (3,0) {};
                  \node (endi) at (0,3) {};
                  \node (endj) at (3,3) {};
                  \node[style=cprocess] (proci) at (0,1) {G};
                  \node[style=qprocess] (procj) at (2.9,2) {~~};
                  \draw[style=cedge] (starti) to node[left,near start] {$C$} (starti|-proci.south);
                  \draw[style=qedge] (startj) to node[right,near start] {$Q_j$} (startj|-procj.south);
                  \draw[style=cedge] (proci.45) to[in=-90, out=90] (procj.-120);
                  \draw[style=cedge] (endi) to node[left,near start] {$C$} (endi|-proci.north);
                  \draw[style=qedge] (endj) to node[right,near start] {$Q_j$} (endj|-procj.north);
                  \end{tikzpicture}}}
  \end{equation}
  \item \textbf{Receiving input from a device.} Classical information
    is received from one of the devices.
    \begin{equation}
    \mpp{0.5}{\scalebox{0.8}{\begin{tikzpicture}
		\node (starti) at (0,0) {};
		\node (startj) at (3,0) {};
		\node (endi) at (0,3) {};
		\node (endj) at (3,3) {};
		\node[style=cprocess] (proci) at (0,1) {~~};
		\node[style=cprocess] (procj) at (2.9,2) {H};
		\draw[style=qedge] (starti) to node[left,near start] {$Q_i$} (starti|-proci.south);
		\draw[style=cedge] (startj) to node[right,near start] {$C$} (startj|-procj.south);
		\draw[style=cedge] (proci.45) to[in=-90, out=90] (procj.-120);
		\draw[style=qedge] (endi) to node[left,near start] {$Q_i$} (endi|-proci.north);
		\draw[style=cedge] (endj) to node[right,near start] {$C$} (endj|-procj.north);
		\end{tikzpicture}}}
    \end{equation}
  \end{enumerate}
\end{definition}

Note that every device-independent protocol has two representations:
as a diagram (including some unlabeled elements) and as a set of
processes from $C Q_1 Q_2$ to $C Q_1 Q_2$. We may use the label $S$ to
refer to either representation. Device-independent protocols can be
composed (e.g., the output quantum states of one protocol can be given
as inputs to another, which corresponds to re-using the devices from
the first protocol in the second).

\section{Randomness expansion}
\label{sec:graphicalexpansion}

\subsection{Linear randomness expansion}
\label{lrsubsec}

We can now phrase security results on device-independent randomness
expansion \cite{colbeck2007quantum} in terms of diagrams.  A
device-independent randomness expansion protocol accepts a seed and
returns a larger output.  Security results for such protocols consist
of asserting that if the seed is uniformly random, then except with
negligible probability, the output is also uniformly random.  The
protocols that we consider for randomness expansion consist of
iterating $2$ untrusted devices many times, and sometimes at random
playing a nonlocal game (such as the CHSH game) to test that the
devices are behaving properly (see Figure~2 in \cite{MY14-1}).

A simple way to assert security for a $2$-device randomness expansion
protocol $R$ is to say that replacing the output with a true uniformly
random state has a negligible effect, i.e.,
\begin{equation}
\label{simpleRE}
\mpp{0.5}{\scalebox{0.8}{\begin{tikzpicture}
	\node[style=qstate] (dstate) at (0.4,0) {};
	\node[style=uniform] (seed) at (-0.4,0) {};
	\node[style=qprocess] (re) at (0,1) {~~~~$R$~~~~};
	\draw[style=qedge] (dstate) to (dstate|-re.south);
	\draw[style=cedge] (seed) to node[left] { $N$} (seed|-re.south);
	\node (out1) at (-.4,2.8) {};
	\node[style=terminal] (end1) at (0.4,2) {};
	\draw[style=cedge] (out1) to node[left,near start] {$M$}(out1|-re.north);
	\draw[style=qedge] (end1) to (end1|-re.north);
\end{tikzpicture}}}
~~ =_\delta ~~
\mpp{0.5}{\scalebox{0.8}{\begin{tikzpicture}
	\node[style=qstate] (dstate2) at (3.6,0) {};
	\node[style=uniform] (seed2) at (2.8,0) {};
	\node[style=qprocess] (re2) at (3.2,1) {~~~~$R$~~~~};
	\draw[style=qedge] (dstate2) to (dstate2|-re2.south);
	\draw[style=cedge] (seed2) to node[left] { $N$} (seed2|-re2.south);
	\node[style=uniform] (seed3) at (2.8,2.2) {};
	\node[style=terminal] (end3) at (2.8,1.7) {};
	\node (out2) at (2.8,2.8) {};
	\node[style=terminal] (end2) at (3.6,2) {};
	\draw[style=cedge] (out2) to node[left,near start] { $M$} (seed3);
	\draw[style=cedge] (end3) to (end3|-re2.north);
	\draw[style=qedge] (end2) to (end2|-re2.north);
\end{tikzpicture}}}
\end{equation}
Above we compressed the two device states of $R$ into a single thick
wire, and we are using the labels $M$ and $N$ as a shorthand for
$\mathbb{C}^{2^M}$ and $\mathbb{C}^{2^N}$. But it is preferable to
have a stronger assertion: we wish to know that the output of the
protocol is also approximately uniform when conditioned on the seed
and on any quantum information entangled with the devices.  The
following theorem captures this stronger assertion.

\begin{theorem}[\textbf{Spot-check protocol}]
  \label{startingthm}
  There exist device-independent protocols $R ( 1 ), R ( 2 ) , R ( 3 )
  , \ldots$, where $R ( N )$ has classical input dimension $2^N$ and
  classical output dimension $2^{2N}$, and there exists a function $\delta
  = \delta ( N ) \in 2^{-\Omega ( N ) }$, such that the following
  hold.
  \begin{enumerate}
  \item \textbf{Soundness.}  For any $r \in R ( N )$ and any state $\Gamma$,
  \begin{equation}
  \mpp{0.5}{\scalebox{1.0}{\begin{tikzpicture}
          \node (out1) at (-.2,2.8) {};
          \node (out2) at (-1.4,2.8) {};
          \node (out3) at (1.2,2.8) {};
          \node[style=terminal] (end1) at (0.5,2) {};
          \node[style=qstate] (dstate1) at (0.7,0) {~~$\Gamma$~~};
          \node[style=uniform] (seed1) at (-.8,0) {};
          \node at (-.8,-.4) {$N$};
          \node[style=qprocess] (re1) at (0,1) {~~~~$r$~~~~};
          \draw[style=qedge] (end1|-dstate1.north) to (end1|-re1.south);
          \draw[style=cedge] (seed1) to[out=0,in=-90] (out1|-re1.south);
          \draw[style=cedge] (out1) to node[left,near start] { $2N$}(out1|-re1.north);
          \draw[style=cedge] (out2) to[out=-90,in=90]  (-1.4,1) to[out=-90,in=180] (seed1);
          \draw[style=qedge] (out3) to (out3|-dstate1.north);
          \draw[style=qedge] (end1) to (end1|-re1.north);
  \end{tikzpicture}}}
  ~~ =_\delta ~~
  \mpp{0.5}{\scalebox{1.0}{\begin{tikzpicture}
          \node[style=qstate] (dstate2) at (5,0) {~~$\Gamma$~~};
          \node[style=uniform] (seed2) at (3.4,0) {};
          \node at (3.4,-.4) {$N$};
          \node[style=qprocess] (re2) at (4.2,1) {~~~~$r$~~~~};
          \node (out4) at (4,3) {};
          \node (out5) at (2.8,3) {};
          \node (out6) at (5.4,3) {};
          \node[style=terminal] (end2) at (4.7,1.9) {};
          \node[style=uniform] (seed3) at (4,2.3) {};
          \node[style=terminal] (end3) at (4,1.9) {};
          \draw[style=qedge] (end2|-dstate2.north) to (end2|-re2.south);
          \draw[style=cedge] (seed2) to[out=0,in=-90]  (out4|-re2.south);
          \draw[style=cedge] (out5) to[out=-90,in=90] (2.8,1) to[out=-90,in=180] (seed2);
          \draw[style=qedge] (out6) to (out6|-dstate2.north);
          \draw[style=qedge] (end2) to  (end2|-re2.north);
          \draw[style=cedge] (out4) to node[left,near start] {$2N$} (seed3);
          \draw[style=cedge] (end3) to node[left] {$2N$} (end3|-re2.north);
          \draw[style=qedge] (end2) to (end2|-re2.north);
  \end{tikzpicture}}}
  \label{soundnesspic}
  \end{equation}

\item \textbf{Completeness.}
\begin{equation}
	\mpp{0.5}{\begin{tikzpicture}
		\node[style=scalar] (sc) at (.5, 0) {$1-\delta$};
		\end{tikzpicture}}
		~~ \in ~~
		\mpp{0.4}{\begin{tikzpicture}
		\node[style=qstate] (dstate2) at (3.6,0) {};
		\node[style=uniform] (seed2) at (2.8,0) {};
		\node[style=qprocess] (re2) at (3.2,1) {~~~~$R(N)$~~~~};
		\draw[style=qedge] (dstate2) to (dstate2|-re2.south);
		\draw[style=cedge] (seed2) to node[left] { $N$} (seed2|-re2.south);
		\node[style=terminal] (end3) at (2.8,2) {};
		\node (out2) at (2.8,2.8) {};
		\node[style=terminal] (end2) at (3.6,2) {};
		\draw[style=cedge] (end3) to node[left] { $2N$} (end3|-re2.north);
		\draw[style=qedge] (end2) to (end2|-re2.north);
	\end{tikzpicture}}
\end{equation}
\end{enumerate}
\end{theorem}

\begin{proof}
  This follows from known results \cite{vazirani2012, MY14-1, MY14-2,
    Arnon:2016}.  See Appendix~\ref{justificationapp} for a formal
  explanation.
\end{proof}

``Soundness'' asserts that the protocols $R(N)$ must either produce
random numbers or fail.  ``Completeness'' asserts that there exist
processes which will make $R ( N )$ succeed with probability
approaching $1$.

The following corollary will be a key step in our proof of unbounded
randomness expansion.  Whereas Theorem~\ref{startingthm} assumes that
the state of the devices is destroyed, the next lemma addresses the
case where the device-state is preserved.  In any diagram, let
$\causal$ denote the set of all causal processes (with input and
output types as implied by the diagram).

\begin{lemma}[Spot-Check Lemma]
  \label{lem:spotcheck}
  There exists $\epsilon = \epsilon ( N ) \in 2^{-\Omega ( N ) }$ such that for every
  integer $N \geq 1$,
  \begin{equation}
	\mpp{0.5}{
	\begin{tikzpicture}
		\node (dstate) at (.5,-0.5) {};
		\node (dstateout) at (.5,2.5) {};
		\node (end2) at (-1.5,2.5) {};     
		\node (end3) at (-.3,2.5) {};
		\node[style=uniform] (seed) at (-.9,0) {};
		\node (seedlabel) at (-.9,-.4) {$N$};
		\node[style=qprocess] (re) at (0,1) {~~$R(N)$~~};
		\draw[style=qedge] (dstate) to (dstate|-re.south);
		\draw[style=qedge] (dstateout) to (dstateout|-re.north);
		\draw[style=cedge] (end3) to node[left,near start] {$2N$} (end3|-re.north);
		\draw[style=cedge] (end2) to (end2|-re.south) to[out=-90,in=180] (seed);
		\draw[style=cedge] (end3|-re.south) to[out=-90,in=0] (seed);
		\node (end) at (0.3,2.8) {};
	\end{tikzpicture}
	}
	~~ \subseteq_{\epsilon} ~~
	\mpp{0.5}{
	\begin{tikzpicture}
		\node (endr3) at (2.5,2.8) {};        
		\node (endr1) at (3,2.8) {};
		\node (endr2) at (4.5,2.8) {};
		\node (dstate2) at (5,-0.5) {};
		\node[style=uniform] (seed2) at (3.6,0) {};
		\node (seed2label) at (3.6,-.4) {$N$};
		\node[style=uniform] (seed3) at (3.6,1) {};
		\node (seed3label) at (3.6,.6) {$2N$};
		\node[style=qprocess] (cause) at (4.6,2) {~~~~$\causal$~~~~};
		\node[style=qprocess] (corr) at (4.7,.7) {\color{white} blank};
		\draw[style=cedge] (seed3) to[in=-90, out=180] (endr1);
		\draw[style=cedge] (seed2) to[in=-90, out=180] (endr3);    
		\draw[style=cedge] (endr2|-corr.south) to[in=0, out=-90, looseness=1] (seed2);
		\draw[style=cedge] (seed3) to[in=-90, out=0] (cause.220); 
		\draw[style=qedge] (dstate2) to[in=-90, out=90] (dstate2|-corr.south);
		\draw[style=qedge] (dstate2|-corr.north) to[in=-90, out=90] (dstate2|-cause.south);
		\draw[style=cedge] (corr.120) to[in=-90, out=90, looseness=1] (corr.120|-cause.south);
		\draw[style=qedge] (endr2) to (endr2|-cause.north);  
	\end{tikzpicture}
	}
\end{equation}
\end{lemma}

\begin{proof}
Let $\delta$ be as in Theorem~\ref{startingthm}, let $N$ be a positive
integer, and let $r \in R( N )$.  We have
\begin{equation}
\label{gammapic}
\mpp{0.5}{\scalebox{1.0}{\begin{tikzpicture}
		\node (out1) at (-.2,2.8) {};
		\node (out2) at (-1.4,2.8) {};
		\node (out3) at (1.2,2.8) {};
		\node[style=terminal] (end1) at (0.5,2) {};
		\node[style=qstate] (dstate1) at (0.7,0) {~~$\Gamma$~~};
		\node[style=uniform] (seed1) at (-.8,0) {};
		\node at (-.8,-.4) {$N$};
		\node[style=qprocess] (re1) at (0,1) {~~~~$r$~~~~};
		\draw[style=qedge] (end1|-dstate1.north) to (end1|-re1.south);
		\draw[style=cedge] (seed1) to[out=0,in=-90] (out1|-re1.south);
		\draw[style=cedge] (out1) to node[left,near start] { $2N$}(out1|-re1.north);
		\draw[style=cedge] (out2) to[out=-90,in=90]  (-1.4,1) to[out=-90,in=180] (seed1);
		\draw[style=qedge] (out3) to (out3|-dstate1.north);
		\draw[style=qedge] (end1) to (end1|-re1.north);
		\end{tikzpicture}}}
~~ =_\delta ~~
\mpp{0.5}{\scalebox{1.0}{\begin{tikzpicture}
		\node[style=qstate] (dstate2) at (5,0) {~~$\Gamma$~~};
		\node[style=uniform] (seed2) at (3.4,0) {};
		\node at (3.4,-.4) {$N$};
		\node[style=qprocess] (re2) at (4.2,1) {~~~~$r$~~~~};
		\node (out4) at (4,3) {};
		\node (out5) at (2.8,3) {};
		\node (out6) at (5.4,3) {};
		\node[style=terminal] (end2) at (4.7,1.9) {};
		\node[style=uniform] (seed3) at (4,2.3) {};
		\node[style=terminal] (end3) at (4,1.9) {};
		\draw[style=qedge] (end2|-dstate2.north) to (end2|-re2.south);
		\draw[style=cedge] (seed2) to[out=0,in=-90]  (out4|-re2.south);
		\draw[style=cedge] (out5) to[out=-90,in=90] (2.8,1) to[out=-90,in=180] (seed2);
		\draw[style=qedge] (out6) to (out6|-dstate2.north);
		\draw[style=qedge] (end2) to  (end2|-re2.north);
		\draw[style=cedge] (out4) to node[left,near start] {$2N$} (seed3);
		\draw[style=cedge] (end3) to node[left] {$2N$} (end3|-re2.north);
		\draw[style=qedge] (end2) to (end2|-re2.north);
		\end{tikzpicture}}}
\end{equation}
for any pure state $\Gamma$.  We construct a duplication of the state
on the right side of the above equation.  Let $Q = V \otimes V$ denote
the quantum register represented by the thick input wire received by
the process $r$.  The effect
\begin{equation}
  \begin{tikzpicture}
  \node (in1) at (-.3,-.8) {};
  \node (in2) at (.5,-.8) {};
  \node[style=terminal] (end1) at (-.3,.6) {};
  \node[style=terminal] (end2) at (.5, .6) {};
  \node[style=qprocess] (re) at (0,0) {~~~~$r$~~~~};
  \draw[style=cedge] (in1) to (in1|-re.south);
  \draw[style=qedge] (in2) to (in2|-re.south);
  \draw[style=cedge] (end1) to (end1|-re.north);
  \draw[style=qedge] (end2) to (end2|-re.north);
  \end{tikzpicture}
\end{equation}
is an element in the dual of $\mathbb{C}^N \otimes V \otimes V$ which
can be written as $\sum_{ijk} \sigma_{jk}^i \left< i \right| \otimes
\left< j \right| \otimes \left< k \right|$, where the matrices
$\Sigma^i = [ \sigma^i_{jk} ]_{jk}$ are positive semidefinite.  Let
$b:\mathbb{C}^N\otimes V \otimes V \to \mathbb{C}^N\otimes V\otimes V$
denote the controlled pure process
\begin{equation}
X \mapsto \left( \sum_i \left| i \right> \left< i \right| \otimes
\sqrt{ \Sigma_i} \otimes \overline{\sqrt{\Sigma_i }} \right) X
\end{equation}
Then, the state
\begin{equation}
\scalebox{0.8}{\begin{tikzpicture}
      \node (out1) at (-1.2,2) {};
      \node (out2) at (-.4,2) {};
      \node (out3) at (.4,2) {};
      \node (out4) at (1.2,2) {};
      \node[style=qstate] (dstate2) at (.8,0) {~~~$\Gamma$~~~};
      \node[style=uniform] (seed2) at (-.8,0) {};
      \node[style=qprocess] (re2) at (0,1) {~~~~$b$~~~~};
      \draw[style=qedge] (out3) to (out3|-re2.north);
      \draw[style=qedge] (out3|-re2.south) to (out3|-dstate2.north);
      \draw[style=qedge] (out4) to (out4|-dstate2.north);
      \draw[style=cedge] (out2) to (out2|-re2.north);
      \draw[style=cedge] (out2|-re2.south) to[out=-90,in=0] (seed2);
      \draw[style=cedge] (out1) to[out=-90,in=90] (out1|-re2.south) to[out=-90,in=180] (seed2);
      \end{tikzpicture}
}
\end{equation}
is a duplication of the state
\begin{equation}
\scalebox{0.8}{      \begin{tikzpicture}
	\node (out1) at (-1.2,2) {};
	\node[style=terminal] (out2) at (-.4,2) {};
	\node[style=terminal] (out3) at (.4,2) {};
	\node (out4) at (1.2,2) {};
	\node[style=qstate] (dstate2) at (.8,0) {~~~$\Gamma$~~~};
	\node[style=uniform] (seed2) at (-.8,0) {};
	\node[style=qprocess] (re2) at (0,1) {~~~~$r$~~~~};
	\draw[style=qedge] (out3) to (out3|-re2.north);
	\draw[style=qedge] (out3|-re2.south) to (out3|-dstate2.north);
	\draw[style=qedge] (out4) to (out4|-dstate2.north);
	\draw[style=cedge] (out2) to (out2|-re2.north);
	\draw[style=cedge] (out2|-re2.south) to[out=-90,in=0] (seed2);
	\draw[style=cedge] (out1) to[out=-90,in=90] (out1|-re2.south) to[out=-90,in=180] (seed2);
	\end{tikzpicture}
}
\end{equation}
Likewise, the state
\begin{equation}
\scalebox{0.8}{      \begin{tikzpicture}
      \node (out1) at (-1.5,2.5) {};
      \node (out2) at (-1,2.5) {};
      \node (out3) at (1,2.5) {};
      \node (out4) at (1.5,2.5) {};
      \node (out5) at (-.4,2.5) {};
      \node (out6) at (.4,2.5) {};
      \node[style=uniform] (seed3) at (0,1.75) {};
      \node[style=qstate] (dstate2) at (.8,0) {~~~$\Gamma$~~~};
      \node[style=uniform] (seed2) at (-1,0) {};
      \node[style=qprocess] (re2) at (0,1) {~~~~$b$~~~~};
      \draw[style=qedge] (out3) to[out=-90,in=90] (out6|-re2.north);
      \draw[style=qedge] (out6|-re2.south) to (out6|-dstate2.north);
      \draw[style=qedge] (out4) to (out4|-dstate2.north);
      \draw[style=cedge] (out2) to[out=-90,in=90] (out5|-re2.north);
      \draw[style=cedge] (out5|-re2.south) to[out=-90,in=0] (seed2);
      \draw[style=cedge] (out1) to[out=-90,in=90] (out1|-re2.south) to[out=-90,in=180] (seed2);
      \draw[style=cedge] (out5) to[out=-90,in=180] node[right, near start] {\footnotesize $~2N$} (seed3);
      \draw[style=cedge] (out6) to[out=-90,in=0] (seed3);
      \end{tikzpicture}
}
\end{equation}
is a duplication of the state on the right side of (\ref{gammapic}).
Therefore by Corollary~\ref{cor:dup}, there is a causal process $c$
such that
\begin{equation}
\mpp{0.5}{\scalebox{0.8}{\begin{tikzpicture}
	\node (out1) at (-1.2,2.8) {};
	\node (out2) at (-.4,2.8) {};
	\node (out3) at (.5,2.8) {};
	\node (out4) at (1.2,2.8) {};
	\node[style=qstate] (dstate) at (.8,0) {~~~$\Gamma$~~~};
	\node[style=uniform] (seed) at (-.8,.4) {};
	\node[style=qprocess] (re) at (0,1.4) {~~~~$r$~~~~};
	\draw[style=qedge] (out3) to (out3|-re.north);
	\draw[style=qedge] (out3|-re.south) to (out3|-dstate.north);
	\draw[style=qedge] (out4) to (out4|-dstate.north);
	\draw[style=cedge] (out2) to (out2|-re.north);
	\draw[style=cedge] (out2|-re.south) to[out=-90,in=0] (seed);
	\draw[style=cedge] (out1) to[out=-90,in=90] (out1|-re.south) to[out=-90,in=180] (seed);
\end{tikzpicture}}}
~~ =_{\sqrt{2 \delta }} ~~
\mpp{0.5}{\scalebox{0.8}{\begin{tikzpicture}
	\node (out5) at (3.2,2.8) {};        
	\node (out6) at (3.7,2.8) {};
	\node (out7) at (5.2,2.8) {};
	\node (out8) at (6.7,2.8) {};
	\node[style=uniform] (seed2) at (4.1,0) {};
	\node (seed2label) at (4.1,-.4) {$N$};
	\node[style=uniform] (seed3) at (4.3,1.2) {};
	\node (seed3label) at (4.2,.8) {$2N$};
	\node[style=qprocess] (cause) at (5.2,2) {~~~~~$c$~~~~~};
	\node[style=qprocess] (corr) at (5.5,1) {~~~~$b$~~~~};
	\node[style=qstate] (dstate2) at (6,0) {~~~$\Gamma$~~~};
	\draw[style=cedge] (corr.145) to (corr.145|-cause.south);
	\draw[style=qedge] (corr.45) to (corr.45|-cause.south);
	\draw[style=qedge] (corr.-45) to (corr.-45|-dstate2.north);
	\draw[style=cedge] (corr.-145) to[out=-90,in=0] (seed2);
	\draw[style=cedge] (out5) to[out=-90,in=90] (out5|-corr.south) to[out=-90,in=180] (seed2);
	\draw[style=cedge] (cause.-145) to[out=-90,in=0] (seed3);
	\draw[style=cedge] (out6) to[out=-90,in=90] (out6|-cause.south) to[out=-90,in=180] (seed3);
	\draw[style=qedge] (out7) to (out7|-cause.north);
	\draw[style=qedge] (out7) to (out7|-cause.north);
	\draw[style=qedge] (out7) to (out7|-cause.north);
	\draw[style=qedge] (out8) to (out8|-dstate2.north);
\end{tikzpicture}}}
\end{equation}
The desired result follows (with $\epsilon = \sqrt{2 \delta}$).
\end{proof}

Additionally, we note the following alternative form of the
completeness assertion for $R( N )$.  The next lemma asserts that the
set of all possible classical outputs of the spot-checking protocol
must contain a state that is close to the (normalized) uniform state
on $\mathbb{C}^{2^N}$. This is similar to the use of ``adjustment
completeness error'' in \cite{MY14-1}.

\begin{lemma}
	\label{complemma}
	There exists $\zeta(N) \in 2^{-\Omega ( N ) }$ such that
	\begin{equation}
	\mpp{0.4}{\scalebox{0.8}{\begin{tikzpicture}
			\node[style=uniform] (sc) at (0, 0) {};
			\node (out1) at (0,2.5) {};
			\draw[style=cedge] (sc) to node[left] {$2N$} (out1);
			\end{tikzpicture}}}
	~~ \in_\zeta ~~
	\mpp{0.4}{\scalebox{0.8}{\begin{tikzpicture}
			\node[style=qprocess] (re) at (2.5,1) {~~~~$R$~~~~};
			\node (out2) at (2.1,2.5) {};
			\node[style=terminal] (end3) at (3,2) {};
			\node[style=qstate] (dstate) at (3,0) {};
			\node[style=uniform] (seed) at (2.1,0) {};
			\draw[style=qedge] (dstate) to (dstate|-re.south);
			\draw[style=cedge] (seed) to node[left,near start] { $N$} (seed|-re.south);
			\draw[style=qedge] (end3) to (end3|-re.north);
			\draw[style=cedge] (out2) to node[left,near start] { $2N$} (out2|-re.north);
			\end{tikzpicture}}}
	\end{equation}
\end{lemma}

\begin{proof}
Combining the soundness and completeness claims in \Cref{startingthm},
we have
\begin{equation}
\mpp{0.4}{\scalebox{0.8}{\begin{tikzpicture}
		\node[style=uniform] (sc) at (0, 0) {};
		\node (out1) at (0,2.5) {};
		\node[style=scalar] at (-1.5,1) {$1-\delta$};
		\draw[style=cedge] (sc) to node[right] {$2N$} (out1);
		\end{tikzpicture}}}
~~ \in_\delta ~~
\mpp{0.4}{\scalebox{0.8}{\begin{tikzpicture}
		\node[style=qprocess] (re) at (2.5,1) {~~~~$R$~~~~};
		\node (out2) at (2.1,2.5) {};
		\node[style=terminal] (end3) at (3,2) {};
		\node[style=qstate] (dstate) at (3,0) {};
		\node[style=uniform] (seed) at (2.1,0) {};
		\draw[style=qedge] (dstate) to (dstate|-re.south);
		\draw[style=cedge] (seed) to node[left,near start] { $N$} (seed|-re.south);
		\draw[style=qedge] (end3) to (end3|-re.north);
		\draw[style=cedge] (out2) to node[left,near start] { $2N$} (out2|-re.north);
		\end{tikzpicture}}}
\end{equation}
The desired result follows, with $\zeta = 2 \delta$.
\end{proof}

\subsection{Unbounded randomness expansion}

Now we discuss a graphical proof of unbounded (rather than linear)
randomness expansion. We would like to apply the spot-checking
protocol and lemma repeatedly, in order to obtain unbounded randomness
expansion. Naively stacking $R(N)$ operations atop one another does
not work. Intuitively, this is because the results of
subsection~\ref{lrsubsec} only apply if the initial seed is
independent of the state of the devices in the protocol $R(N )$.  If
we reuse devices in two successive iterations of the protocol, that
independence assumption may not hold. However, it was observed in
\cite{CSW14,CY13, MY14-1} that we can still obtain unbounded
randomness by employing two pairs of devices and alternating which
pair is employed in the protocol. In effect, one proves that the
second application of the protocol wipes out any correlation with the
first device, which may then be used in the third step to erase
correlation with the second, and so on.

\begin{definition}
  \label{def:sn}
  For any integer $N \geq 1$, let $S ( N)$ denote the set of processes
  given by
  \[
  \mpp{0.5}{
  \begin{tikzpicture}
    \node[style=qprocess] (scp2) at (-2,0) {~~$S(N)$~~};
    \node (in1) at (-2.5,-1) {};
    \node (in2) at (-2,-1) {};
    \node (in3) at (-1.5,-1) {};
    \node (out1) at (-2.5,1) {};
    \node (out2) at (-2,1) {};
    \node (out3) at (-1.5,1) {};
    \draw[style=cedge] (in1) to (in1|-scp2.south);
    \draw[style=qedge] (in2) to (in2|-scp2.south);
    \draw[style=qedge] (in3) to (in3|-scp2.south);
    \draw[style=cedge] (out1) to (out1|-scp2.north);
    \draw[style=qedge] (out2) to (out2|-scp2.north);
    \draw[style=qedge] (out3) to (out3|-scp2.north);
    \node (eq) at (-0.2,0) {$:=$};
    \node[style=qprocess] (scp3) at (.8,-0.9) {$R(N)$};
    \node[style=qprocess] (scp4) at (2,.4) {$R(2N)$};
    \node (in4) at (.5,-1.8) {};
    \node (in5) at (1.1,-1.8) {};
    \node (in6) at (2.4,-1.8) {};
    \node (out4) at (.5,1.8) {};
    \node (out5) at (1.1,1.8) {};
    \node (out6) at (2.4,1.8) {};
    \draw[style=cedge] (in4) to (in4|-scp3.south);
    \draw[style=qedge] (in5) to (in5|-scp3.south);
    \draw[style=qedge] (in6) to (in6|-scp4.south);
    \draw[style=qedge] (out5) to (out5|-scp3.north);
    \draw[style=qedge] (out6) to (out6|-scp4.north);
    \draw[style=cedge] (out4) to[out=-90,in=90] (scp4.120);
    \draw[style=cedge] (out4|-scp3.north) to[out=90,in=-90] (scp4.-120);
  \end{tikzpicture}}
  \]
where $R(N)$ denotes the process set from Theorem~\ref{startingthm}.
Let $S_k(N)$ denote the composition $S(4^{k-1}N) \circ S(4^{k-2}N)
\circ \ldots \circ S(N)$.
\end{definition}

We prove the following lemma (which will be a building block for a
later induction proof).  For the remainder of this section, let
$\epsilon = \epsilon (M )$ be the error function from
Lemma~\ref{lem:spotcheck}.
\begin{lemma}
  \label{lem:sm}
  For every integer $M \geq 1$, we have
\begin{equation}
  \mpp{0.5}{\scalebox{.8}{
    \begin{tikzpicture}
      \node (in1) at (0,-0.5) {};
      \node (in2) at (.5,-0.5) {};
      \node (out1) at (-1.5,3) {};     
      \node (out2) at (-.5,3) {};
      \node (out3) at (0,3) {};
      \node (out4) at (.5,3) {};
      \node[style=uniform] (seed) at (-1,0) {};
      \node (seedlabel) at (-.9,-.4) {$M$};
      \node[style=qprocess] (re) at (0,1.2) {~~$S(M)$~~};
      \draw[style=qedge] (in1) to (in1|-re.south);
      \draw[style=qedge] (in2) to (in2|-re.south);
      \draw[style=cedge] (out1) to (out1|-re.south) to[out=-90,in=180] (seed);
      \draw[style=cedge] (out2) to node[left,near start] {$4M$} (out2|-re.north);
      \draw[style=cedge] (out2|-re.south) to[out=-90,in=0] (seed);
      \draw[style=qedge] (out3) to (out3|-re.north);
      \draw[style=qedge] (out4) to (out4|-re.north);
    \end{tikzpicture}
  }}
  ~~\displaystyle \subseteq_{\epsilon'}~~
   \mpp{0.5}{\scalebox{.8}{
    \begin{tikzpicture}
    \node (out5) at (2.3,3) {};        
    \node (out6) at (2.8,3) {};
    \node (out7) at (3.2,3) {};
    \node (out8) at (4.6,3) {};
    \node (in3) at (4.7,-0.5) {};
    \node (in4) at (5.1,-.5) {};
    \node[style=uniform] (seed2) at (3.6,0) {};
    \node (seed2label) at (3.6,-.4) {$M$};
    \node[style=uniform] (seed3) at (3.6,1) {};
    \node (seed3label) at (3.6,.6) {$4M$};
    \node[style=qprocess] (cause) at (4.6,2.2) {~~~~$\causal$~~~~};
    \node[style=qprocess] (corr) at (4.7,.7) {\color{white} blank};
    \draw[style=cedge] (seed3) to[in=-90, out=180] (out6);
    \draw[style=cedge] (seed2) to[in=-90, out=180] (out5);    
     \draw[style=cedge] (corr.-140) to[in=0, out=-90, looseness=1] (seed2);
     \draw[style=cedge] (seed3) to[in=-90, out=0] (cause.-140); 
     \draw[style=qedge] (in3) to[in=-90, out=90] (in3|-corr.south);
     \draw[style=qedge] (in4) to[in=-90, out=90] (in4|-corr.south);
     \draw[style=cedge] (in3|-corr.north) to[in=-90, out=90] (in3|-cause.south);
     \draw[style=qedge] (in4|-corr.north) to[in=-90, out=90] (in4|-cause.south);
     \draw[style=qedge] (corr.140) to[in=-90,out=90,looseness=.8] (out7);
    \draw[style=qedge] (out8) to (out8|-cause.north);
    \end{tikzpicture}
  }}
  \label{eq:lem1}
\end{equation}
where $\epsilon' ( M ) = \epsilon (M) + \epsilon ( 2 M )$.
\end{lemma}

\begin{proof}
  We first use \Cref{def:sn} to expand the left-hand side of
  \Cref{eq:lem1} and then apply \Cref{lem:spotcheck} to the occurrence
  of $R(M)$ in the resulting diagram.
\begin{equation}
  \mpp{0.5}{\scalebox{0.8}{
    \begin{tikzpicture}
      \node (in1) at (0,-0.5) {};
      \node (in2) at (.5,-0.5) {};
      \node (out1) at (-1.5,3.5) {};     
      \node (out2) at (-.5,3.5) {};
      \node (out3) at (0,3.5) {};
      \node (out4) at (.5,3.5) {};
      \node[style=uniform] (seed) at (-1,.3) {};
      \node (seedlabel) at (-1,-.1) {$M$};
      \node[style=qprocess] (re) at (0,1.5) {~~$S(M)$~~};
      \draw[style=qedge] (in1) to (in1|-re.south);
      \draw[style=qedge] (in2) to (in2|-re.south);
      \draw[style=cedge] (out1) to (out1|-re.south) to[out=-90,in=180] (seed);
      \draw[style=cedge] (out2) to node[left,near start] {$4M$} (out2|-re.north);
      \draw[style=cedge] (out2|-re.south) to[out=-90,in=0] (seed);
      \draw[style=qedge] (out3) to (out3|-re.north);
      \draw[style=qedge] (out4) to (out4|-re.north);
    \end{tikzpicture}
  }}
    ~~ = ~~
  \mpp{0.5}{\scalebox{0.8}{
    \begin{tikzpicture}
      \node[style=qprocess] (scp3) at (3.5,.7) {$R(N)$};
      \node[style=qprocess] (scp4) at (4.7,2) {$R(2N)$};
      \node[style=uniform] (seed2) at (2.6,-.2) {};
      \node (in3) at (3.8,-.5) {};
      \node (in4) at (4.8,-.5) {};
      \node (out5) at (2,3.5) {};
      \node (out6) at (3.3,3.5) {};
      \node (out7) at (3.8,3.5) {};
      \node (out8) at (4.8,3.5) {};
      \draw[style=qedge] (in3) to (in3|-scp3.south);
      \draw[style=qedge] (out7) to (out7|-scp3.north);
      \draw[style=qedge] (in4) to (in4|-scp4.south);
      \draw[style=qedge] (out8) to (out8|-scp4.north);
      \draw[style=cedge] (out5) to (out5|-scp3.south) to[out=-90,in=180] (seed2);
      \draw[style=cedge] (out6|-scp3.south) to[out=-90,in=0] (seed2);
      \draw[style=cedge] (out6) to[out=-90,in=90] (scp4.120);
      \draw[style=cedge] (out6|-scp3.north) to[out=90,in=-90] (scp4.-120);
    \end{tikzpicture}
  }}
    ~~ \subseteq_{\epsilon(M)} ~~
  \mpp{0.5}{\scalebox{0.8}{
    \begin{tikzpicture}
      \node[style=qprocess] (corr) at (0,1.2) {\color{white} blank};
      \node[style=qprocess] (cause) at (0,2.7) {~~~~$\causal$~~~~};
      \node[style=qprocess] (re) at (1,4.2) {$R(2M)$};
      \node[style=uniform] (seed) at (-1.2,.4) {};
      \node (seedlabel) at (-1.2,0) {$M$};
      \node[style=uniform] (seed2) at (-1,1.9) {};
      \node (seed2label) at (-1,1.5) {$2M$};
      \node (in1) at (.3,0) {};
      \node (in2) at (1.3,0) {};
      \node (out1) at (-2,5.3) {};     
      \node (out2) at (-1,5.3) {};
      \node (out3) at (0,5.3) {};
      \node (out4) at (1.3,5.3) {};
      \draw[style=qedge] (in1) to (in1|-corr.south);
      \draw[style=qedge] (in1|-corr.north) to (in1|-cause.south);
      \draw[style=qedge] (in2) to (in2|-re.south);
      \draw[style=cedge] (out1) to (out1|-corr.south) to[out=-90,in=180] (seed);
      \draw[style=cedge] (corr.-120) to[out=-90,in=0] (seed);
      \draw[style=cedge] (corr.120) to (corr.120|-cause.south);
      \draw[style=cedge] (out2) to[out=-90,in=90] node[left,near start] {$4M$~~} (re.135);
      \draw[style=cedge] (re.-135) to[out=-90,in=90] (-1.5,3) to[out=-90,in=180] (seed2);
      \draw[style=cedge] (cause.-150) to[out=-90,in=0] (seed2);
      \draw[style=qedge] (out3) to (cause);
      \draw[style=qedge] (out4) to (out4|-re.north);
    \end{tikzpicture}
  }}
\end{equation}
  We can now move the spider of type $2M$ along its wire to place it
  below the occurrence of $R(2M)$ and apply \Cref{lem:spotcheck} again.
\begin{equation}
\mpp{0.5}{\scalebox{0.8}{
		\begin{tikzpicture}
		\node[style=qprocess] (corr) at (0,1.2) {\color{white} blank};
		\node[style=qprocess] (cause) at (0,2.7) {~~~~$\causal$~~~~};
		\node[style=qprocess] (re) at (1,4.2) {$R(2M)$};
		\node[style=uniform] (seed) at (-1.2,.4) {};
		\node (seedlabel) at (-1.2,0) {$M$};
		\node[style=uniform] (seed2) at (-1,1.9) {};
		\node (seed2label) at (-1,1.5) {$2M$};
		\node (in1) at (.3,0) {};
		\node (in2) at (1.3,0) {};
		\node (out1) at (-2,5.3) {};     
		\node (out2) at (-1,5.3) {};
		\node (out3) at (0,5.3) {};
		\node (out4) at (1.3,5.3) {};
		\draw[style=qedge] (in1) to (in1|-corr.south);
		\draw[style=qedge] (in1|-corr.north) to (in1|-cause.south);
		\draw[style=qedge] (in2) to (in2|-re.south);
		\draw[style=cedge] (out1) to (out1|-corr.south) to[out=-90,in=180] (seed);
		\draw[style=cedge] (corr.-120) to[out=-90,in=0] (seed);
		\draw[style=cedge] (corr.120) to (corr.120|-cause.south);
		\draw[style=cedge] (out2) to[out=-90,in=90] node[left,near start] {$4M$~~} (re.135);
		\draw[style=cedge] (re.-135) to[out=-90,in=90] (-1.5,3) to[out=-90,in=180] (seed2);
		\draw[style=cedge] (cause.-150) to[out=-90,in=0] (seed2);
		\draw[style=qedge] (out3) to (cause);
		\draw[style=qedge] (out4) to (out4|-re.north);
		\end{tikzpicture}
}}
~~ = ~~  
\mpp{0.5}{\scalebox{0.8}{
		\begin{tikzpicture}
		\node[style=qprocess] (corr) at (0,1.2) {\color{white} blank};
		\node[style=qprocess] (cause) at (0,2.4) {~~~~$\causal$~~~~};
		\node[style=qprocess] (re) at (2,4.2) {$R(2M)$};
		\node[style=uniform] (seed) at (-1.2,.4) {};
		\node (seedlabel) at (-1.2,0) {$M$};
		\node[style=uniform] (seed2) at (1,3.4) {};
		\node (seed2label) at (1,3) {$2M$};
		\node (in1) at (.3,0) {};
		\node (in2) at (2.3,0) {};
		\node (out1) at (-2,5.3) {};     
		\node (out2) at (-1,5.3) {};
		\node (out3) at (0,5.3) {};
		\node (out4) at (2.3,5.3) {};
		\draw[style=qedge] (in1) to (in1|-corr.south);
		\draw[style=qedge] (in1|-corr.north) to (in1|-cause.south);
		\draw[style=qedge] (in2) to (in2|-re.south);
		\draw[style=cedge] (out1) to (out1|-corr.south) to[out=-90,in=180] (seed);
		\draw[style=cedge] (corr.-120) to[out=-90,in=0] (seed);
		\draw[style=cedge] (corr.120) to (corr.120|-cause.south);
		\draw[style=cedge] (out2) to[out=-90,in=90] node[left,near start] {$4M~$~~} (re.135);
		\draw[style=cedge] (cause.-150) to[out=-90,in=0] (-.9,1.6) to[out=180,in=-90] (-1.5,2.5) to[out=90,in=180] (seed2);
		\draw[style=cedge] (re.-135) to[out=-90,in=0] (seed2);
		\draw[style=qedge] (out3) to (cause);
		\draw[style=qedge] (out4) to (out4|-re.north);
		\end{tikzpicture}
}}
~~ \subseteq_{\epsilon ( 2 M )} ~~
\mpp{0.5}{\scalebox{0.8}{
		\begin{tikzpicture}
		\node[style=qprocess] (corr) at (0,1.2) {\color{white} blank};
		\node[style=qprocess] (cause) at (0,2.4) {~~~~$\causal$~~~~};
		\node[style=uniform] (seed) at (-1.2,.4) {};
		\node (seedlabel) at (-1.2,0) {$M$};
		\node[style=uniform] (seed2) at (1,3.4) {};
		\node (seed2label) at (1,3) {$2M$};
		\node[style=qprocess] (corr2) at (2,4.2) {\color{white} blank};
		\node[style=qprocess] (cause2) at (1.8,5.4) {~~~~$\causal$~~~~};
		\node[style=uniform] (seed3) at (.8,4.6) {};
		\node (seed3label) at (.8,4.2) {$4M$};
		\node (in1) at (.3,0) {};
		\node (in2) at (2.3,0) {};
		\node (out1) at (-2,6.5) {};     
		\node (out2) at (-1,6.5) {};
		\node (out3) at (0,6.5) {};
		\node (out4) at (1.8,6.5) {};
		\draw[style=qedge] (in1) to (in1|-corr.south);
		\draw[style=qedge] (in1|-corr.north) to (in1|-cause.south);
		\draw[style=qedge] (in2) to (in2|-corr2.south);
		\draw[style=qedge] (in2|-corr2.north) to (in2|-cause2.south);
		\draw[style=cedge] (out1) to (out1|-corr.south) to[out=-90,in=180] (seed);
		\draw[style=cedge] (corr.-120) to[out=-90,in=0] (seed);
		\draw[style=cedge] (corr.120) to (corr.120|-cause.south);
		\draw[style=cedge] (out2) to[out=-90,in=180] (seed3);
		\draw[style=cedge] (cause2.-135) to[out=-90,in=0] (seed3);
		\draw[style=cedge] (cause.-150) to[out=-90,in=0] (-.9,1.6) to[out=180,in=-90] (-1.5,2.5) to[out=90,in=180] (seed2);
		\draw[style=cedge] (corr2.-135) to[out=-90,in=0] (seed2);
		\draw[style=cedge] (corr2.135) to (corr2.135|-cause2.south);
		\draw[style=qedge] (out3) to (cause);
		\draw[style=qedge] (out4) to (cause2);
		\end{tikzpicture}
}}
\end{equation}
In the rightmost diagram above, the three lower boxes form a
stochastic process, and this completes the proof.
\end{proof}

\begin{lemma}
  \label{lem:inductionk}
  For all integers $N, k \geq 1$, we have
\begin{equation}
\label{eq:lem2}
\mpp{0.5}{
	\begin{tikzpicture}
	\def\outY{1}
	\def\inY{-1}
	\def\lineA{-2.1}
	\def\lineB{-.7}
	\def\lineC{0}
	\def\lineD{.7}
	\node[style=qprocess] (scp) at (0,0) {~~~$S_k(N)$~~~};
	\node[style=uniform] (seed) at ({((\lineA+\lineB)/2)},-1) {};
	\node at ({((\lineA+\lineB)/2)},-1.4) {$N$};
	\node (in1) at (\lineC,\inY) {};
	\node (in2) at (\lineD,\inY) {};
	\node (out1) at (\lineA,\outY) {};     
	\node (out2) at (\lineB,\outY) {};
	\node (out3) at (\lineC,\outY) {};
	\node (out4) at (\lineD,\outY) {};
	\draw[style=qedge] (in1) to (in1|-scp.south);
	\draw[style=qedge] (in2) to (in2|-scp.south);
	\draw[style=cedge] (out1) to (out1|-scp.south) to[out=-90,in=180] (seed);
	\draw[style=cedge] (out2|-scp.south) to[out=-90,in=0] (seed);
	\draw[style=cedge] (out2) to (out2|-scp.north);
	\draw[style=qedge] (out3) to (out3|-scp.north);
	\draw[style=qedge] (out4) to (out4|-scp.north);
	\end{tikzpicture}
}
~~ \subseteq_\gamma ~~
\mpp{0.5}{
	\begin{tikzpicture}
	\node[style=qprocess] (corr) at (.2,-.8) {~~~\color{white} blank~~~};
	\node[style=qprocess] (cause) at (.2,.8) {~~~~$\causal$~~~~};
	\node[style=uniform] (seed) at (-1.5,-1.8) {};
	\node (seedlabel) at (-1.5,-2.2) {$N$};
	\node[style=uniform] (seed2) at (-1.2,0) {};
	\node (seed2label) at (-1.2,-.4) {$4^k N$};
	\def\outY{1.5}
	\def\inY{-2}
	\def\lineA{-2.5}
	\def\lineB{-1.8}
	\def\lineC{.2}
	\def\lineD{.7}
	\node (in1) at (\lineC,\inY) {};
	\node (in2) at (\lineD,\inY) {};
	\node (out1) at (\lineA,\outY) {};     
	\node (out2) at (\lineB,\outY) {};
	\node (out3) at (-1.2,\outY) {};
	\node (out4) at (\lineC,\outY) {};
	\draw[style=cedge] (out1) to (out1|-corr.south) to[out=-90,in=180] (seed);
	\draw[style=cedge] (corr.-145) to[out=-90,in=0] (seed);
	\draw[style=cedge] (out2) to (out2|-cause.south) to[out=-90,in=180] (seed2);
	\draw[style=qedge] (in1) to (in1|-corr.south);
	\draw[style=cedge] (cause.-150) to[out=-90,in=0] (seed2);
	\draw[style=qedge] (in1) to (in1|-corr.south);
	\draw[style=cedge] (in1|-corr.north) to (in1|-cause.south);
	\draw[style=qedge] (out4) to (out4|-cause.north);
	\draw[style=qedge] (in2) to (in2|-corr.south);
	\draw[style=qedge] (in2|-corr.north) to (in2|-cause.south);
	\draw[style=qedge] (corr.145) to[out=90,in=-90] (out3);
	\end{tikzpicture}
}
\end{equation}
where $\gamma = \gamma ( N, k) =\sum_{i=0}^{2k-1} \epsilon ( 2^ i N)$.
\end{lemma}

\begin{proof}
  By induction on $k$. The case $k=1$ follows from \Cref{lem:sm} by
  setting $M=N$. So suppose that $k>1$. Then we can expand the
  left-hand side of \Cref{eq:lem2} and apply the induction hypothesis
  to obtain
\begin{equation}
\mpp{0.5}{\scalebox{.8}{
		\begin{tikzpicture}
		\def\outY{2.5}
		\def\inY{-1}
		\def\lineA{-2.1}
		\def\lineB{-.7}
		\def\lineC{0}
		\def\lineD{.7}
		\node[style=qprocess] (scp) at (0,0) {~~~$S_{k-1}(N)$~~~};
		\node[style=qprocess] (scp2) at (0,1.5) {~~~$S(4^{k-1}N)$~~~};
		\node[style=uniform] (seed) at ({((\lineA+\lineB)/2)},-1) {};
		\node at ({((\lineA+\lineB)/2)},-1.4) {$N$};
		\node (in1) at (\lineC,\inY) {};
		\node (in2) at (\lineD,\inY) {};
		\node (out1) at (\lineA,\outY) {};     
		\node (out2) at (\lineB,\outY) {};
		\node (out3) at (\lineC,\outY) {};
		\node (out4) at (\lineD,\outY) {};
		\draw[style=cedge] (out1) to (out1|-scp.south) to[out=-90,in=180] (seed);
		\draw[style=cedge] (out2|-scp.south) to[out=-90,in=0] (seed);
		\draw[style=cedge] (out2) to (out2|-scp2.north);
		\draw[style=cedge] (out2|-scp.north) to (out2|-scp2.south);
		\draw[style=qedge] (in1) to (in1|-scp.south);
		\draw[style=qedge] (in1|-scp.north) to (in1|-scp2.south);
		\draw[style=qedge] (in2) to (in2|-scp.south);
		\draw[style=qedge] (in2|-scp.north) to (in2|-scp2.south);
		\draw[style=qedge] (out3) to (out3|-scp2.north);
		\draw[style=qedge] (out4) to (out4|-scp2.north);
		\end{tikzpicture}
}}
~~ \subseteq_\alpha ~~
\mpp{0.5}{\scalebox{.8}{
		\begin{tikzpicture}
		\node[style=qprocess] (corr) at (.2,-.8) {~~~\color{white} blank~~~};
		\node[style=qprocess] (cause) at (.2,.8) {~~~~$\causal$~~~~};
		\node[style=qprocess] (scp) at (-1.2,2.2) {~~~$S(4^{k-1}N)$~~~};
		\node[style=uniform] (seed) at (-1.5,-1.8) {};
		\node (seedlabel) at (-1.5,-2.2) {$N$};
		\node[style=uniform] (seed2) at (-1.2,0) {};
		\node (seed2label) at (-1.2,-.4) {$4^{k-1} N$};
		\def\outY{3.5}
		\def\inY{-2}
		\def\lineA{-2.8}
		\def\lineB{-2}
		\def\lineC{-.4}
		\def\lineD{.7}
		\node (in1) at (.2,\inY) {};
		\node (in2) at (\lineD,\inY) {};
		\node (out1) at (\lineA,\outY) {};     
		\node (out2) at (\lineB,\outY) {};
		\node (out3) at (-1.2,\outY) {};
		\node (out4) at (-.4,\outY) {};
		\draw[style=cedge] (out1) to (out1|-corr.south) to[out=-90,in=180] (seed);
		\draw[style=cedge] (corr.-145) to[out=-90,in=0] (seed);
		\draw[style=cedge] (out2) to (out2|-scp.north);
		\draw[style=cedge] (out2|-scp.south) to (out2|-cause.south) to[out=-90,in=180] (seed2);
		\draw[style=qedge] (in1) to (in1|-corr.south);
		\draw[style=cedge] (cause.-150) to[out=-90,in=0] (seed2);
		\draw[style=qedge] (in1) to (in1|-corr.south);
		\draw[style=cedge] (in1|-corr.north) to (in1|-cause.south);
		\draw[style=qedge] (corr.145) to[out=90,in=-90] (scp.south);
		\draw[style=qedge] (scp) to[out=90,in=-90] (out3);
		\draw[style=qedge] (out4) to (out4|-scp.north);
		\draw[style=qedge] (out4|-scp.south) to[out=-90,in=90] (cause.north);
		\draw[style=qedge] (in2) to (in2|-corr.south);
		\draw[style=qedge] (in2|-corr.north) to (in2|-cause.south);
		\end{tikzpicture}
}}
~~ \subseteq_\beta ~~
\mpp{0.5}{\scalebox{.8}{
		\begin{tikzpicture}
		\node[style=qprocess] (corr) at (.2,-.8) {~~~\color{white} blank~~~};
		\node[style=qprocess] (cause) at (.2,.8) {~~~~$\causal$~~~~};
		\node[style=uniform] (seed) at (-1.8,-1.8) {};
		\node (seedlabel) at (-1.8,-2.2) {$N$};
		\node[style=uniform] (seed2) at (-1.2,0) {};
		\node (seed2label) at (-1.2,-.4) {$4^{k-1} N$};
		\def\outY{4.8}
		\def\inY{-2}
		\def\lineA{-3.5}
		\def\lineB{-1.7}
		\def\lineC{-.4}
		\def\lineD{.7}
		\node (in1) at (.2,\inY) {};
		\node (in2) at (\lineD,\inY) {};
		\node (out1) at (\lineA,\outY) {};     
		\node (out2) at (\lineB,\outY) {};
		\node (out3) at (-2.5,\outY) {};
		\node (out4) at (-.7,\outY) {};
		\node[style=qprocess] (corr2) at (-1.2,2.2) {~~~\color{white} blank~~~};
		\node[style=qprocess] (cause2) at (-1.2,3.8) {~~~~$\causal$~~~~};
		\node[style=uniform] (seed3) at (-2.4,3) {};
		\node (seed3label) at (-2.4,2.6) {$4^k N$};
		\def\lineBB{-3}
		\node (out22) at (\lineBB,\outY) {};
		\node (out33) at (-1.2,\outY) {};
		\draw[style=cedge] (out1) to (out1|-corr.south) to[out=-90,in=180] (seed);
		\draw[style=cedge] (corr.-145) to[out=-90,in=0] (seed);
		\draw[style=cedge] (out2|-corr2.south) to (out2|-cause.south) to[out=-90,in=180] (seed2);
		\draw[style=cedge] (out22) to (out22|-cause2.south) to[out=-90,in=180] (seed3);
		\draw[style=cedge] (out2|-cause2.south) to[out=-90,in=0] (seed3);
		\draw[style=qedge] (in1) to (in1|-corr.south);
		\draw[style=cedge] (cause.-150) to[out=-90,in=0] (seed2);
		\draw[style=qedge] (in1) to (in1|-corr.south);
		\draw[style=cedge] (in1|-corr.north) to (in1|-cause.south);
		\draw[style=qedge] (corr.145) to[out=90,in=-90] (corr2.south);
		\draw[style=qedge] (corr.145) to[out=90,in=-90] (corr2.south);
		\draw[style=cedge] (corr2.north) to[out=90,in=-90] (cause2.south);
		\draw[style=qedge] (out33) to (cause2.north);
		\draw[style=qedge] (out4|-corr2.south) to[out=-90,in=90] (cause.north);
		\draw[style=qedge] (out4|-corr2.north) to (out4|-cause2.south);
		\draw[style=qedge] (out2|-corr2.north) to[out=90,in=-90] (out3);
		\draw[style=qedge] (in2) to (in2|-corr.south);
		\draw[style=qedge] (in2|-corr.north) to (in2|-cause.south);
		\end{tikzpicture}
}}
\end{equation}
where $\alpha = \epsilon ( M ) + \epsilon ( 2 M ) + \ldots + \epsilon
( 2^{2k-3} M) $ and $\beta = \epsilon ( 2^{2k-2} M ) + \epsilon (
2^{2k-1} M )$.  The desired result follows.
\end{proof}

\begin{theorem}[Soundness for Unbounded Expansion]
  \label{URE}
  There is a function $\lambda = \lambda ( N ) \in 2^{-\Omega ( N ) }$
  such that for any $N, k \geq 1$,
\begin{equation}
\label{eq:unbounded}
\mpp{0.5}{
	\begin{tikzpicture}
	\def\outY{1.5}
	\def\lineA{-.8}
	\def\lineB{0}
	\def\lineC{.8}
	\def\lineD{1.6}
	\node (out1) at (\lineA,\outY) {};
	\node (out2) at (\lineD,\outY) {};
	\node[style=qstate] (enemy) at (.5,-1.5) {\rule{0pt}{1ex}\rule{8ex}{0pt}};
	\node[style=qprocess] (scp) at (0,0) {~~~~~$S_k(N)$~~~~~};
	\node[style=uniform] (seed1) at (\lineA,-1.5) {};
	\node at (\lineA,-1.9) {$N$};  
	\node[style=terminal] (term1) at (\lineB,1) {};
	\node[style=terminal] (term2) at (\lineC,1) {};
	\draw[style=cedge] (seed1) to (seed1|-scp.south);
	\draw[style=qedge] (term1|-enemy.north) to (term1|-scp.south);
	\draw[style=qedge] (term2|-enemy.north) to (term2|-scp.south);
	\draw[style=qedge] (out2|-enemy.north) to (out2);
	\draw[style=qedge] (term1) to (term1|-scp.north);
	\draw[style=qedge] (term2) to (term2|-scp.north);
	\draw[style=cedge] (out1) to (out1|-scp.north);
	\end{tikzpicture}
}
~~ \subseteq_{\lambda} ~~
\mpp{0.5}{
	\begin{tikzpicture}
	\def\outY{.5}
	\def\lineA{-.8}
	\def\lineB{0}
	\def\lineC{.8}
	\def\lineD{1.6}
	\node (out1) at (\lineA,\outY) {};
	\node (out2) at (\lineD,\outY) {};
	\node[style=scalar] (sc) at (-1.5,-.6) {};
	\node[style=qstate] (enemy) at (.5,-1.5) {\rule{0pt}{1ex}\rule{8ex}{0pt}};
	\node[style=uniform] (seed1) at (\lineA,-1.5) {};
	\node at (\lineA,-1.9) {$4^k N$};  
	\node (end) at (\lineC,\outY) {};
	\draw[style=cedge] (seed1) to (out1);
	\draw[style=qedge] (end|-enemy.north) to (end);
	\end{tikzpicture}
}
\end{equation}
\end{theorem}

\begin{proof}
  By the property noted in \Cref{eq:spiders-discard}, we can add a
  terminated branch to the left-hand side of \Cref{eq:unbounded} so
  that we can apply \Cref{lem:inductionk}.
  \begin{equation}
\mpp{0.5}{\scalebox{.8}{
		\begin{tikzpicture}
		\def\outY{1.5}
		\def\lineZ{-2}
		\def\lineA{-.8}
		\def\lineB{0}
		\def\lineC{.8}
		\def\lineD{1.6}
		\node (out1) at (\lineA,\outY) {};
		\node (out2) at (\lineD,\outY) {};
		\node[style=qstate] (enemy) at (.5,-1.5) {\rule{0pt}{1ex}\rule{8ex}{0pt}};
		\node[style=qprocess] (scp) at (0,0) {~~~~~$S_k(N)$~~~~~};
		\node[style=uniform] (seed1) at ({(\lineA+\lineZ)/2},-1.5) {};
		\node at ({(\lineA+\lineZ)/2},-1.9) {$N$};  
		\node[style=terminal] (term1) at (\lineB,1) {};
		\node[style=terminal] (term2) at (\lineC,1) {};
		\node[style=terminal] (term3) at (\lineZ,1) {};
		\draw[style=cedge] (term3) to (term3|-scp.south) to[out=-90,in=180] (seed1);
		\draw[style=cedge] (out1|-scp.south) to[out=-90,in=0] (seed1);
		\draw[style=qedge] (term1|-enemy.north) to (term1|-scp.south);
		\draw[style=qedge] (term2|-enemy.north) to (term2|-scp.south);
		\draw[style=qedge] (out2|-enemy.north) to node[right,near end] {$X$} (out2);
		\draw[style=qedge] (term1) to (term1|-scp.north);
		\draw[style=qedge] (term2) to (term2|-scp.north);
		\draw[style=cedge] (out1) to node[left,near start] {$4^k N$}(out1|-scp.north);
		\end{tikzpicture}
}}
~~ \subseteq_{\lambda} ~~
\mpp{0.5}{\scalebox{.8}{
		\begin{tikzpicture}
		\node[style=qprocess] (corr) at (.2,-.8) {~~~\color{white} blank~~~};
		\node[style=qprocess] (cause) at (.2,.8) {~~~~$\causal$~~~~};
		\node[style=uniform] (seed) at (-1.5,-1.8) {};
		\node (seedlabel) at (-1.5,-2.2) {$N$};
		\node[style=uniform] (seed2) at (-1.2,0) {};
		\node (seed2label) at (-1.2,-.4) {$4^k N$};
		\node[style=qstate] (enemy) at (.5,-2) {\rule{0pt}{1ex}\rule{8ex}{0pt}};
		\def\outY{1.5}
		\def\inY{-2}
		\def\lineA{-2.5}
		\def\lineB{-1.8}
		\def\lineC{.2}
		\def\lineD{.7}
		\def\lineE{1.5}
		\node (in1) at (\lineC,\inY) {};
		\node (in2) at (\lineD,\inY) {};
		\node[style=terminal] (out1) at (\lineA,\outY) {};     
		\node (out2) at (\lineB,\outY+.5) {};
		\node[style=terminal] (out3) at (-1.2,\outY) {};
		\node[style=terminal] (out4) at (\lineC,\outY) {};
		\node (out5) at (\lineE,\outY+.5) {};
		\draw[style=cedge] (out1) to (out1|-corr.south) to[out=-90,in=180] (seed);
		\draw[style=cedge] (corr.-145) to[out=-90,in=0] (seed);
		\draw[style=cedge] (out2) to (out2|-cause.south) to[out=-90,in=180] (seed2);
		\draw[style=cedge] (cause.-150) to[out=-90,in=0] (seed2);
		\draw[style=qedge] (in1|-enemy.north) to (in1|-corr.south);
		\draw[style=cedge] (in1|-corr.north) to (in1|-cause.south);
		\draw[style=qedge] (out4) to (out4|-cause.north);
		\draw[style=qedge] (in2|-enemy.north) to (in2|-corr.south);
		\draw[style=qedge] (in2|-corr.north) to (in2|-cause.south);
		\draw[style=qedge] (out3) to[out=-90,in=90] (corr.145);
		\draw[style=qedge] (out5) to node[right,near start] {$X$} (out5|-enemy.north);
		\end{tikzpicture}
}}
\end{equation}

  We next apply causality (see the discussion of \Cref{eq:causal})
  followed by \Cref{eq:spiders-discard}.
\begin{equation}
\mpp{0.5}{\scalebox{.8}{
		\begin{tikzpicture}
		\node[style=qprocess] (corr) at (.2,-.8) {~~~\color{white} blank~~~};
		\node[style=qprocess] (cause) at (.2,.8) {~~~~$\causal$~~~~};
		\node[style=uniform] (seed) at (-1.5,-1.8) {};
		\node (seedlabel) at (-1.5,-2.2) {$N$};
		\node[style=uniform] (seed2) at (-1.2,0) {};
		\node (seed2label) at (-1.2,-.4) {$4^k N$};
		\node[style=qstate] (enemy) at (.5,-2) {\rule{0pt}{1ex}\rule{8ex}{0pt}};
		\def\outY{1.5}
		\def\inY{-2}
		\def\lineA{-2.5}
		\def\lineB{-1.8}
		\def\lineC{.2}
		\def\lineD{.7}
		\def\lineE{1.5}
		\node (in1) at (\lineC,\inY) {};
		\node (in2) at (\lineD,\inY) {};
		\node[style=terminal] (out1) at (\lineA,\outY) {};     
		\node (out2) at (\lineB,\outY+.5) {};
		\node[style=terminal] (out3) at (-1.2,\outY) {};
		\node[style=terminal] (out4) at (\lineC,\outY) {};
		\node (out5) at (\lineE,\outY+.5) {};
		\draw[style=cedge] (out1) to (out1|-corr.south) to[out=-90,in=180] (seed);
		\draw[style=cedge] (corr.-145) to[out=-90,in=0] (seed);
		\draw[style=cedge] (out2) to (out2|-cause.south) to[out=-90,in=180] (seed2);
		\draw[style=cedge] (cause.-150) to[out=-90,in=0] (seed2);
		\draw[style=qedge] (in1|-enemy.north) to (in1|-corr.south);
		\draw[style=cedge] (in1|-corr.north) to (in1|-cause.south);
		\draw[style=qedge] (out4) to (out4|-cause.north);
		\draw[style=qedge] (in2|-enemy.north) to (in2|-corr.south);
		\draw[style=qedge] (in2|-corr.north) to (in2|-cause.south);
		\draw[style=qedge] (out3) to[out=-90,in=90] (corr.145);
		\draw[style=qedge] (out5) to node[right,near start] {$X$} (out5|-enemy.north);
		\end{tikzpicture}
}}
~~ = ~~
\mpp{0.5}{\scalebox{.8}{
		\begin{tikzpicture}
		\node[style=qprocess] (corr) at (.2,-.8) {~~~\color{white} blank~~~};
		\node (cause) at (.2,.8) {};
		\node[style=uniform] (seed) at (-1.5,-1.8) {};
		\node (seedlabel) at (-1.5,-2.2) {$N$};
		\node[style=uniform] (seed2) at (-1.2,0) {};
		\node (seed2label) at (-1.2,-.4) {$4^k N$};
		\node[style=qstate] (enemy) at (.5,-2) {\rule{0pt}{1ex}\rule{8ex}{0pt}};
		\def\outY{1.5}
		\def\inY{-2}
		\def\lineA{-2.5}
		\def\lineB{-1.8}
		\def\lineC{.2}
		\def\lineD{.7}
		\def\lineE{1.5}
		\node (in1) at (\lineC,\inY) {};
		\node (in2) at (\lineD,\inY) {};
		\node[style=terminal] (out1) at (\lineA,\outY) {};     
		\node (out2) at (\lineB,\outY+.5) {};
		\node[style=terminal] (out3) at (-1.2,\outY) {};
		\node[style=terminal] (out4a) at (\lineC-.6,\outY-.5) {};
		\node[style=terminal] (out4b) at (\lineC,\outY-.5) {};
		\node[style=terminal] (out4c) at (\lineC+.6,\outY-.5) {};
		\node (out5) at (\lineE,\outY+.5) {};
		\draw[style=cedge] (out1) to (out1|-corr.south) to[out=-90,in=180] (seed);
		\draw[style=cedge] (corr.-145) to[out=-90,in=0] (seed);
		\draw[style=cedge] (out2) to (out2|-cause.south) to[out=-90,in=180] (seed2);
		\draw[style=cedge] (out4a) to (out4a|-cause.south) to[out=-90,in=0] (seed2);
		\draw[style=qedge] (in1|-enemy.north) to (in1|-corr.south);
		\draw[style=cedge] (in1|-corr.north) to (out4b);
		\draw[style=qedge] (out4c|-enemy.north) to (out4c|-corr.south);
		\draw[style=qedge] (out4c) to (out4c|-corr.north);
		\draw[style=qedge] (out3) to[out=-90,in=90] (corr.145);
		\draw[style=qedge] (out5) to node[right,near start] {$X$} (out5|-enemy.north);
		\end{tikzpicture}
}}
~~ = ~~
\mpp{0.5}{\scalebox{.8}{
		\begin{tikzpicture}
		\def\outY{.2}
		\def\inY{-2}
		\def\lineA{-2.5}
		\def\lineB{-1.5}
		\def\lineC{.2}
		\def\lineD{.7}
		\def\lineE{1.5}
		\node[style=qprocess] (corr) at (.2,-.8) {~~~\color{white} blank~~~};
		\node (cause) at (.2,.8) {};
		\node[style=uniform] (seed) at (\lineC-.6,-1.8) {};
		\node (seedlabel) at (\lineC-.6,-2.2) {$N$};
		\node[style=uniform] (seed2) at (\lineB,-1.8) {};
		\node (seed2label) at (\lineB,-2.2) {$4^k N$};
		\node[style=qstate] (enemy) at (.7,-2) {\rule{0pt}{1ex}\rule{8ex}{0pt}};
		\node (in1) at (\lineC,\inY) {};
		\node (in2) at (\lineD,\inY) {};
		\node (out2) at (\lineB,\outY+.5) {};
		\node[style=terminal] (out4a) at (\lineC-.6,\outY-.2) {};
		\node[style=terminal] (out4b) at (\lineC,\outY-.2) {};
		\node[style=terminal] (out4c) at (\lineC+.6,\outY-.2) {};
		\node (out5) at (\lineE,\outY+.5) {};
		\draw[style=cedge] (seed) to (seed|-corr.south);
		\draw[style=cedge] (out2) to (seed2);
		\draw[style=qedge] (out4a) to (out4a|-corr.north);
		\draw[style=qedge] (in1|-enemy.north) to (in1|-corr.south);
		\draw[style=cedge] (in1|-corr.north) to (out4b);
		\draw[style=qedge] (out4c|-enemy.north) to (out4c|-corr.south);
		\draw[style=qedge] (out4c) to (out4c|-corr.north);
		\draw[style=qedge] (out5) to node[right,near start] {$X$} (out5|-enemy.north);
		\end{tikzpicture}
}}
\end{equation}
The set of processes described by the rightmost diagram consist of a
uniform state of dimension $4^k N$ together with a subnormalized state
of $X$, as desired.

Let $\lambda ( N ) = \sum_{i=0}^\infty \epsilon ( 2^i N )$, which
upper bounds the error term $\gamma ( N, k )$.  Note that for any
nonnegative function $f$ on the set $\{ 0, 1, 2, \ldots, \}$
\begin{eqnarray}
\label{implication}
f \in 2^{-\Omega (N ) } & \Longrightarrow & \sum_{i=0}^{\infty} f (
2^i N ) \in 2^{-\Omega ( N )}.
\end{eqnarray}
Thus $\lambda \in 2^{-\Omega ( N ) }$.  This completes the proof.
\end{proof}  

Theorem~\ref{URE} asserts soundness for $S_k(N)$.  Completeness for
$S_k(N)$ follows from Lemma~\ref{complemma} and the implication
(\ref{implication}) stated above.

\subsection{Formalization}
\label{formalsubsec}

In addition to their intuitive appeal, the graphical structures of
categorical quantum mechanics are amenable to computer
formalization. In the long term, this will be critically important for
managing the complexity of medium- and large-scale security proofs. In
this respect, computers can play a number of roles including
validation and verification, copying and reuse, and proof search and
discovery.

As part of our investigations, we have produced a computer-verified
skeleton version of the main sequence of steps from our proof, for the
case $k = 2$. We used the Globular proof assistant \cite{BKV16}. The
reader can find and explore the proof object at \cite{globularfile},
and a video of the diagrammatic moves involved in the proof is
available at \cite{globularvideo}. The Globular proof assistant
provides a system for creating string diagram proofs, based on the
perspective of higher-dimensional re-writing. In this project we used
the system to prototype our arguments, and found the tool quite useful
despite a few rough edges.

In Globular, one begins by declaring generators, atomic components
which can be joined together into more complex diagrams. These
generators come in several dimensions; strings in dimension 1
(classical, quantum), processes in dimension 2 (e.g., the spot-check
protocols), and equations in dimension 3 (e.g., the spot-check lemma,
the causality principle). More general protocols are complex
two-dimensional diagrams. A proof becomes a three-dimensional diagram,
though we often think of it as a ``movie'' whose frames are iterated
slices of the 3-D diagram, tracing through the diagrammatic moves of
the proof.

We used Globular to prototype the proof of Theorem \ref{URE}. In
particular, we used it to prototype and validate the general strategy
of our proof in the case $k=2$. Figure \ref{globpics} shows part of
the sequence of Globular diagrams in our proof of Theorem \ref{URE}
for $k=2$; the entire sequence involved 60 steps.  (The sequence
includes four applications of \Cref{lem:spotcheck} followed by the
final application of causality.)

\begin{figure}[h]
  \begin{center}
    \begin{tabular}{ccccm{.5cm}ccm{.5cm}c}
      \includegraphics[height=3cm]{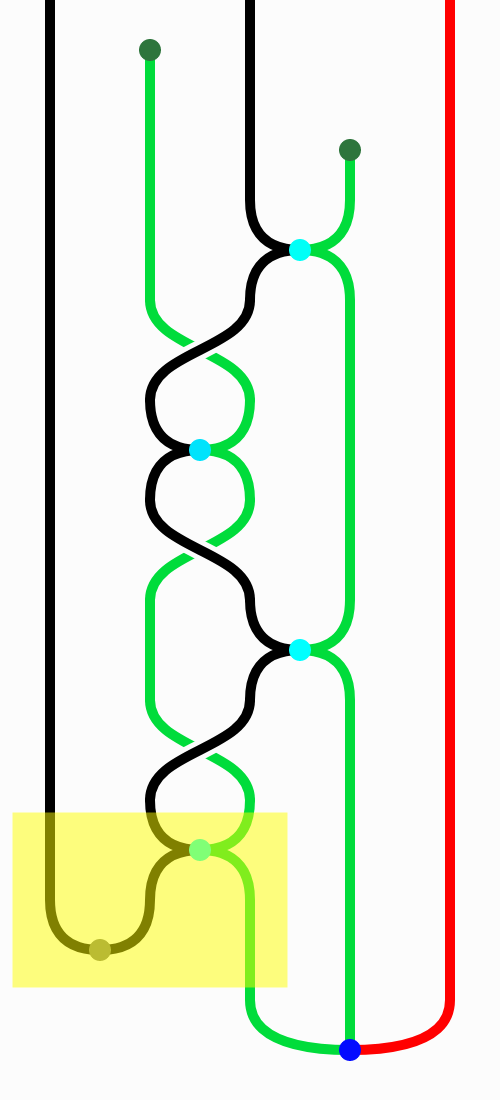}
      &
        \includegraphics[height=3cm]{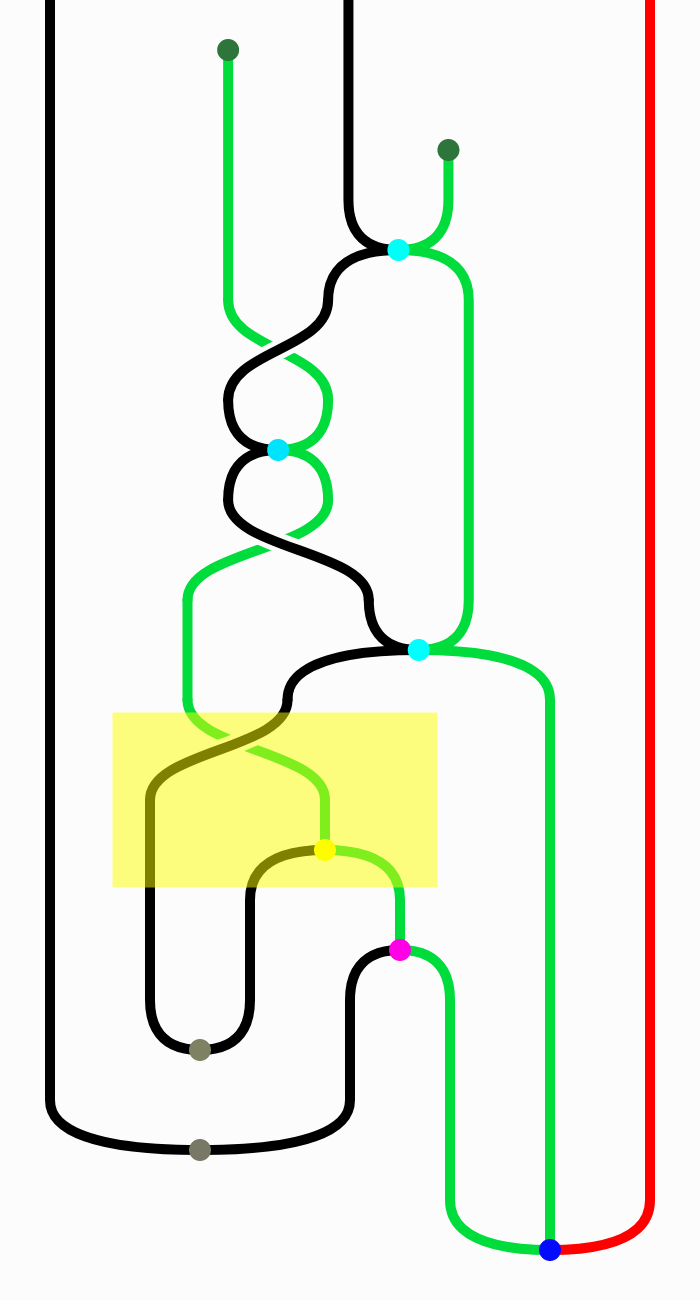}
      &
        \includegraphics[height=3cm]{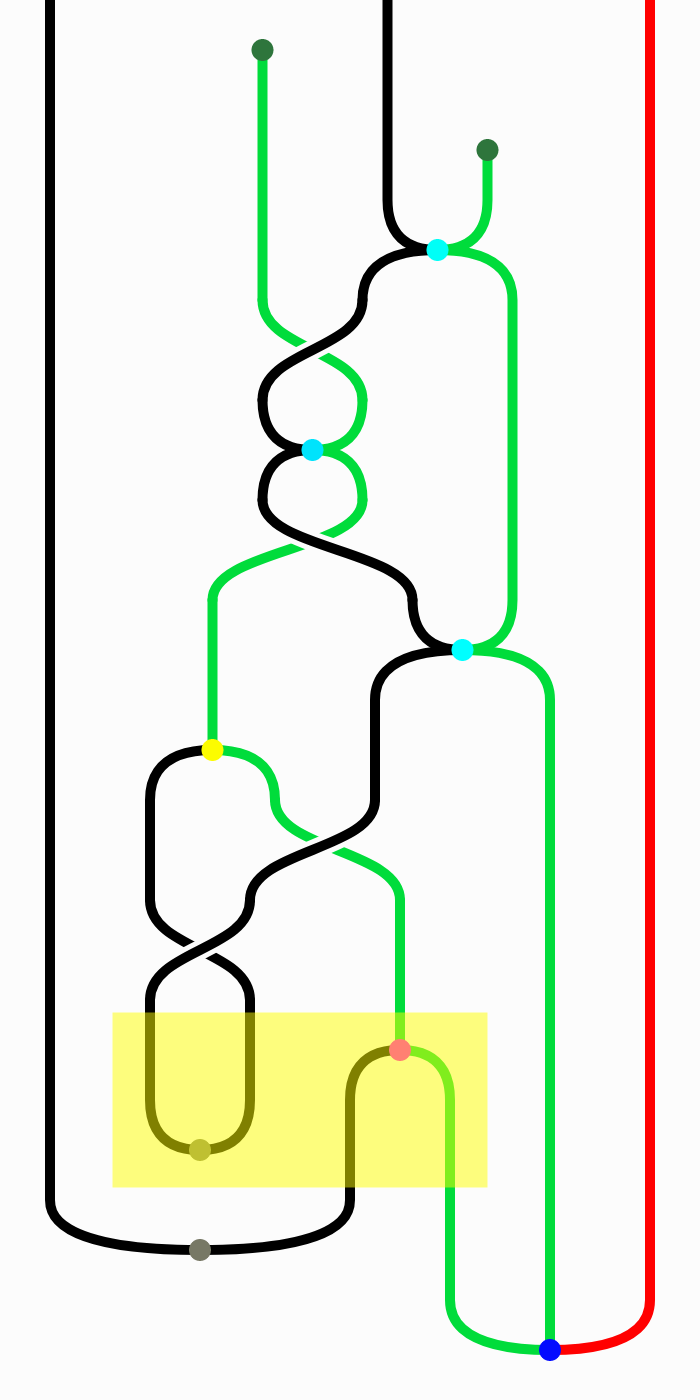}
      &
        \includegraphics[height=3cm]{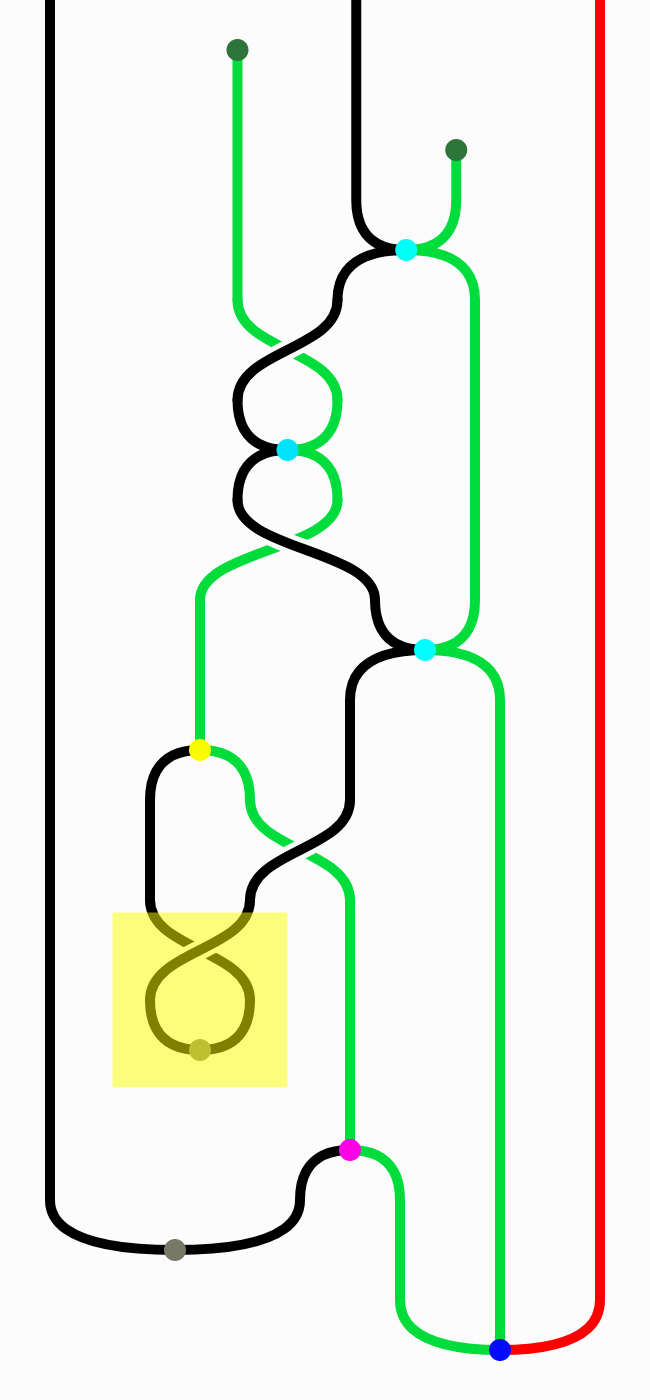}
      & ... &
        \includegraphics[height=3cm]{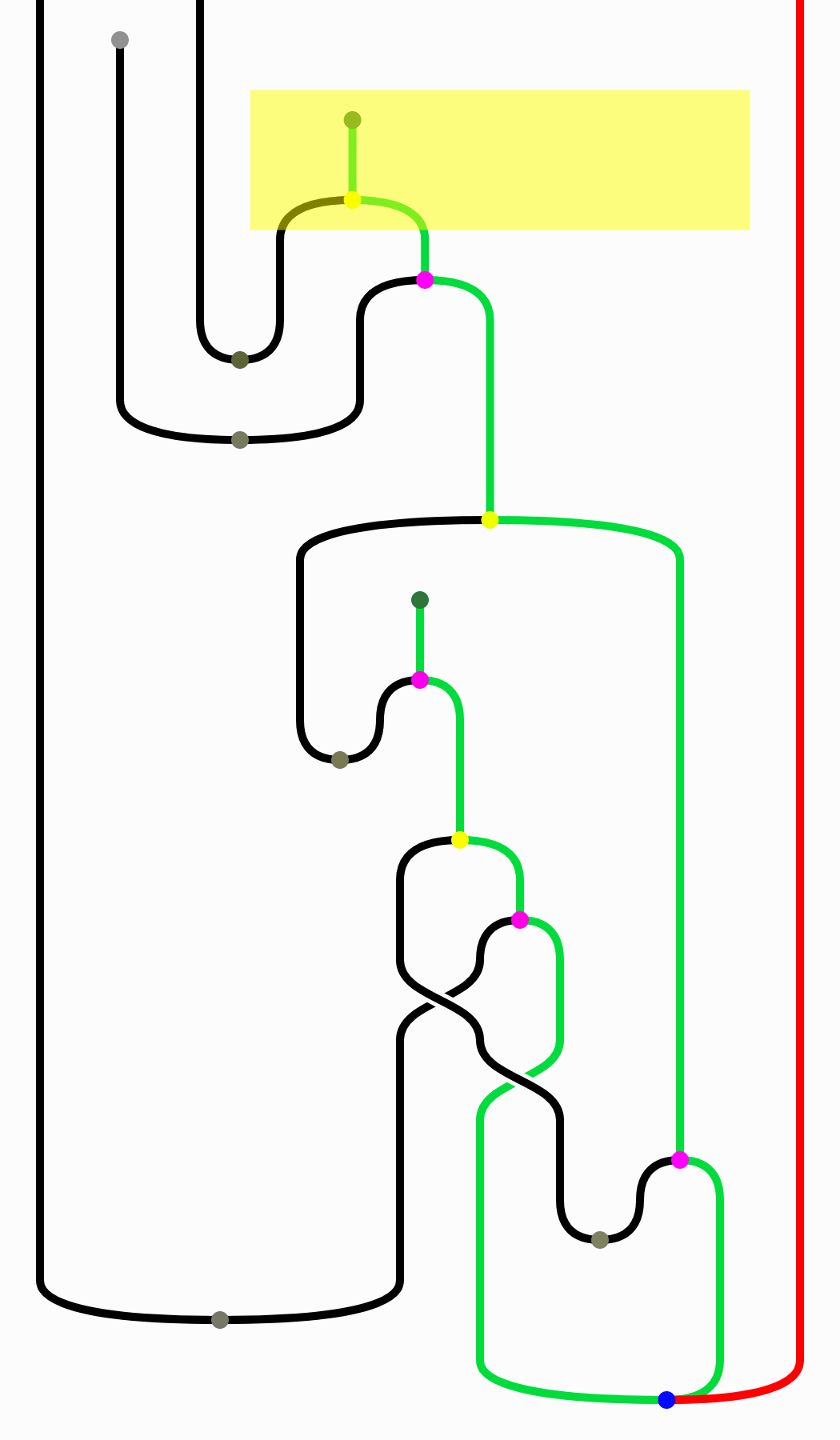}
      &
        \includegraphics[height=3cm]{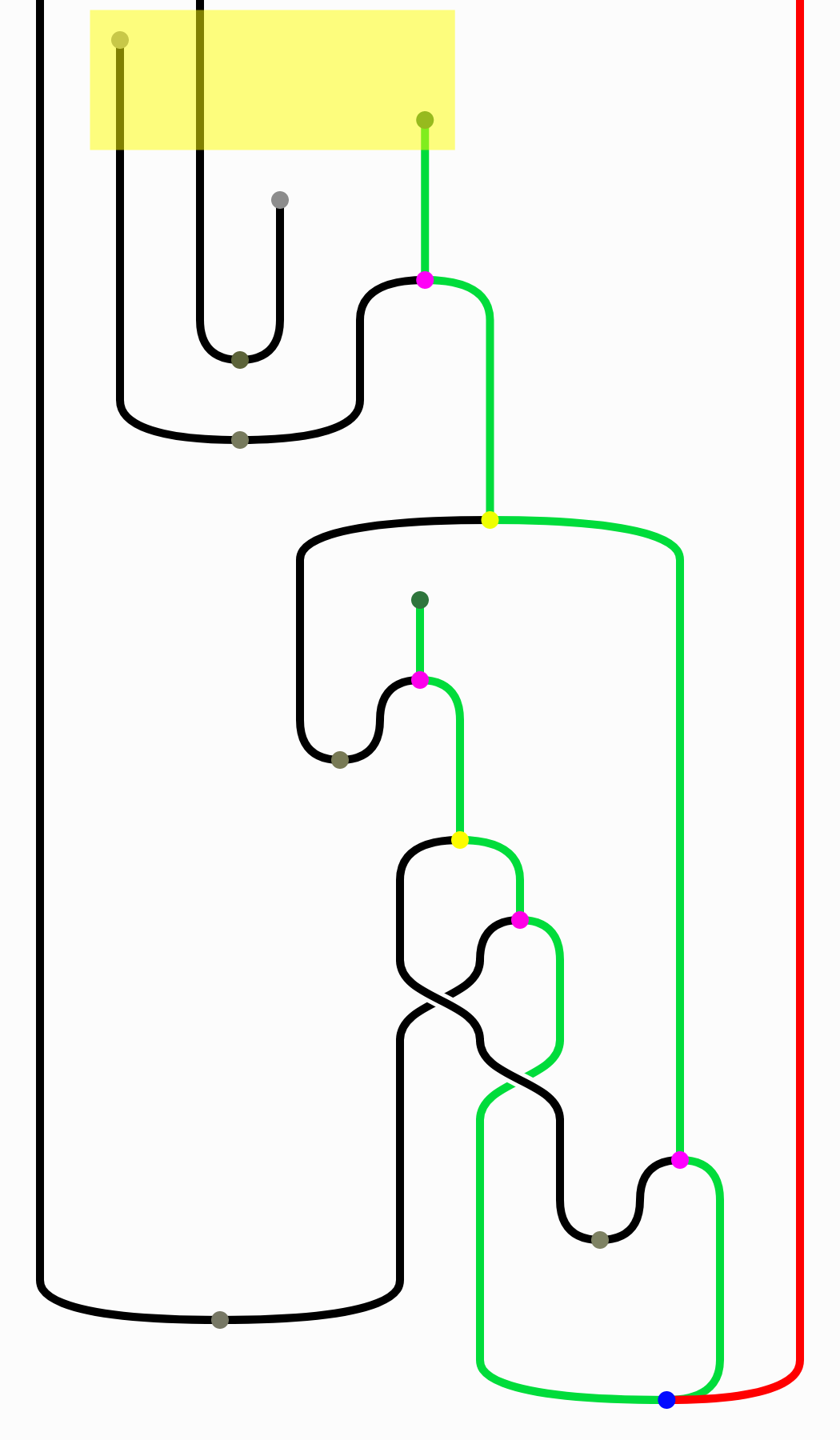}
      & ... &    
        \includegraphics[height=3cm]{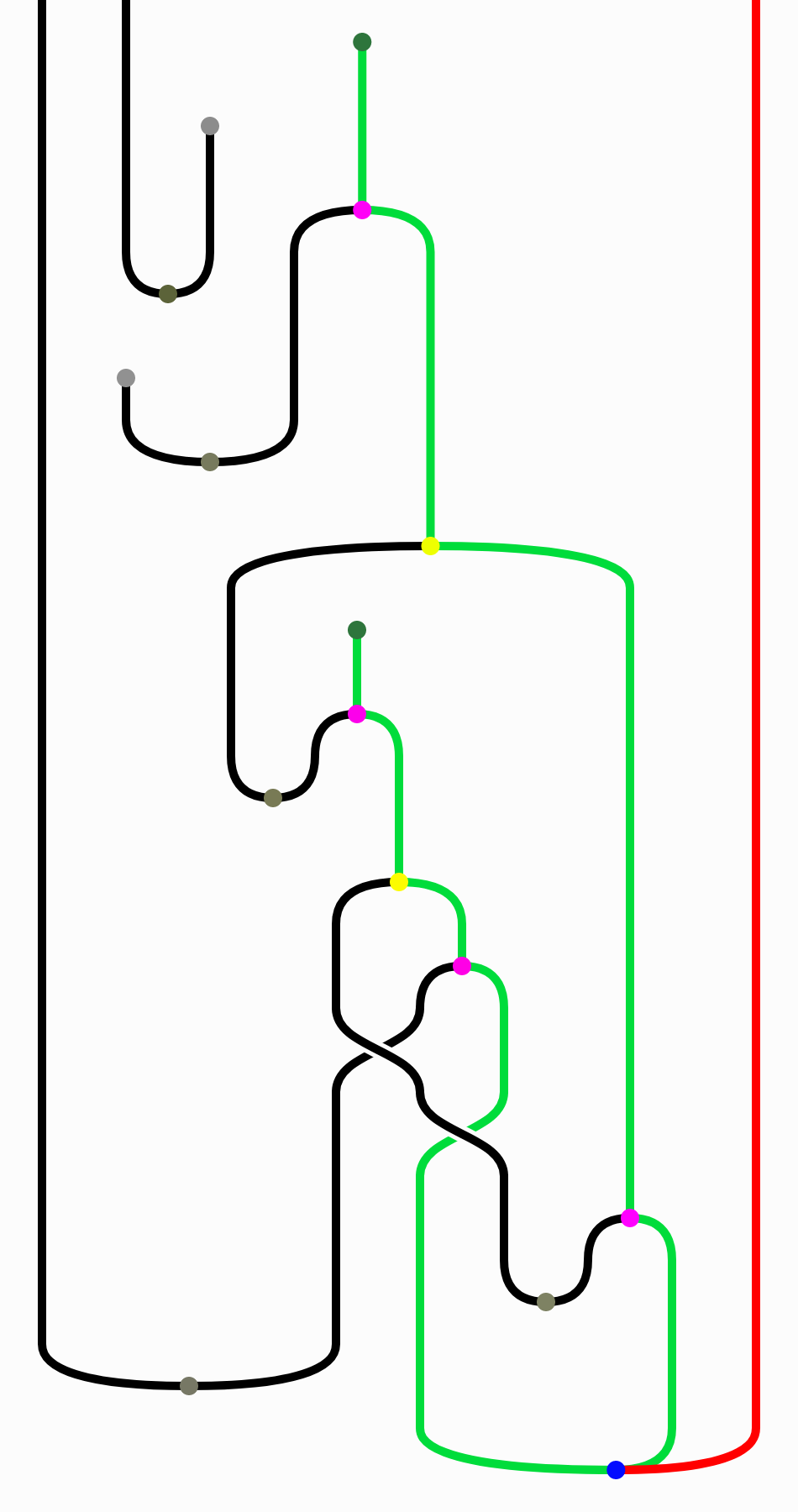}
      \\
      0 & 1 & 2 & 3 & ... & 55 & 56 & ... & 60 \\
    \end{tabular}
    \caption{Steps in a Globular proof of Theorem \ref{URE} for $k=2$}
    \label{globpics}
  \end{center}
\end{figure}

We found several significant benefits to building proofs in
Globular. As diagrams become more complex, our ability to manipulate
them with pen and paper is limited. Globular automates the management
of diagrams, allowing for easy reuse and undo.

The proof assistant also ``type-checks'' the user, admitting only
valid constructions and proof. This allows the user to explore the
space of possible diagrams and proofs without accidentally introducing
errors. This will be particularly important for learning because it
permits a focus on concepts over calculation.

More generally, formalization serves to identify gaps in our
reasoning. In this case, we first identified the graphical form of the
spot-check lemma (\Cref{lem:spotcheck}) and the final form of Theorem
\ref{URE}. By trying to prove the theorem in Globular we first
validated the back-and-forth approach described in the proof of our
theorem, but also showed that it was insufficient to yield our desired
result. This, in turn, helped to identify the role of causality in our
argument.

Overall, we found the Globular proof assistant to be quite helpful in
prototyping and managing our arguments, although some issues limit its
practical application. For example, three steps are needed to pass
from step 56 to step 60 (see above), despite the fact that the two
diagrams are more-or-less identical. This problem reifies as diagrams
become more complicated; in our proof, 19 steps of sliding operations
were needed between the second and third applications of the
spot-check protocol. Improvements to such tools, both in underlying
computation and representation and in user interface would be valuable
areas for future research.

\section{Conclusion}
\label{sec:conc}

Our graphical proof of unbounded expansion was based on two central
steps: one was the application of the Spot-Check Lemma
(Lemma~\ref{lem:spotcheck}), and the other was the principle of
\emph{causality}.  Causality is an elementary step in symbolic proofs
for quantum information, but in the case of our graphical proof it is
an important manipulation.

We have used the tools of categorical quantum mechanics to give a
streamlined proof that unbounded randomness expansion can be obtained
via the spot-checking protocol. We hope to have convinced the reader
of the usefulness and potential of graphical methods in quantum
cryptography for proof exposition.  Also, when graphical proofs are
appropriately created, they open the door to automated proof-checking.
Our experience using the Globular proof assistant can be seen as
interesting case study in the usefulness of the software and we hope
that our experience can motivate future work.

Our goal for later work is to develop a language for quantum
cryptography that allows a wide range of expositions of old results
and proofs of new results.  Some proofs (including unbounded
randomness expansion) seem easiest to understand in graphical form,
while others (such as Proposition~\ref{prop:approxprop}) may be most
accessible as algebraic proofs.  Thus, an ideal framework would allow
easy translation back and forth between algebraic and graphical
expositions.

\section*{Acknowledgements}

We are greatly indebted to David Spivak for providing some of the
original inspiration for this project and for helping us to get it
started.  CAM would also like to thank Brad Lackey for seminars at the
University of Maryland that deepened his understanding of axiomatic
quantum information. NJR is funded by the Department of Defense.  This
paper includes contributions from the U. S. National Institute of
Standards and Technology, and is not subject to copyright in the
United States.

\nocite{*} \bibliographystyle{apsrev4-1} \bibliography{graphicalcrypto}

\appendix
\section{Proof of Proposition~\ref{prop:approxprop}}
\label{approxpropsec}

In this section we revert to standard notation for quantum systems.
The state $\Psi$ can be written as a density operator
\begin{equation}
  \Psi = \sum_i \left| i \right> \left< i \right| \otimes M_i,
\end{equation}
where $M_i$ are positive semidefinite operators on $Q_1$, with $Q =
Q_1 \otimes Q_1$.  Then, the duplicated state of $\Psi$ is given by
\begin{eqnarray}
  \label{sympure}
  \Psi' & = & \sum_i \left| i \right> \left< i \right| \otimes
             \left| i \right> \left< i \right| \otimes 
             \left( \Vec \sqrt{M_i} \right) \left( \Vec \sqrt{M_i} \right)^*,
\end{eqnarray}
where the sum is taken over $i \in \{ 1, 2, \ldots, \dim ( C ) \}$,
and where $\Vec ( X )$ denotes the vector $\sum_{ij} x_{ij} \left| i
\right> \otimes \left| j \right>$ for any $X = \sum_{ij} x_{ij} \left|
i \right> \left< j \right|$.  If $\Phi$ is such that
\begin{equation}
  \Psi =_\epsilon \Phi,
\end{equation}
then
\begin{equation}
  \left\| \Psi - \Phi \right\|_1 \leq 2 \epsilon,
\end{equation}
and therefore if we let
\begin{equation}
  \Phi = \sum_i \left| i \right> \left< i \right| \otimes M'_i,
\end{equation}
we have
\begin{equation}
  \sum_i \left\| M_i - M'_i \right\|_1 \leq 2 \epsilon.
\end{equation}
By Lemma 3.37 from \cite{Wat17}, we have
\begin{equation}
  \sum_i \left\| \sqrt{M_i} - \sqrt{M'_i} \right\|_2^2 \leq 2
  \epsilon.
\end{equation}
For any $i$,
\begin{eqnarray}
  \lefteqn{\left\| (\Vec \sqrt{M_i} ) ( \Vec \sqrt{M_i} )^* - 
  (\Vec \sqrt{M_i} ) ( \Vec \sqrt{M_i} )^* \right\|_1}  \\
& \leq & \left\| (\Vec \sqrt{M_i} ) ( \Vec \sqrt{M_i} )^* - 
         (\Vec \sqrt{M_i} ) ( \Vec \sqrt{M'_i} )^* \right\|_1 \\
& & + \left\| (\Vec \sqrt{M_i} ) ( \Vec \sqrt{M'_i} )^* - 
    (\Vec \sqrt{M'_i} ) ( \Vec \sqrt{M'_i} )^* \right\|_1 \\
& \leq & \left\| \Vec \sqrt{M_i} \right\| \left\| \Vec \sqrt{M_i} - 
      \Vec \sqrt{M'_i} \right\| + \left\| \Vec \sqrt{M'_i} \right\| 
      \left\| \Vec \sqrt{M_i} - \Vec \sqrt{M'_i} \right\|
\end{eqnarray}
Therefore, applying the Cauchy-Schwartz inequality, the trace distance
between the duplicated states of $\Phi$ and $\Psi$ is upper bounded by
\begin{eqnarray}
  \lefteqn{\sum_i  \left( \left\| \Vec \sqrt{M_i} \right\| + 
  \left\| \Vec \sqrt{M'_i } \right\| \right)  
  \left\| \Vec \sqrt{M_i} - \Vec \sqrt{M'_i} \right\|} \\
& \leq & \sqrt{ \left[ \sum_i \left( \left\| \Vec \sqrt{M_i} \right\| 
         + \left\| \Vec \sqrt{M'_i } \right\| \right)^2 \right]
         \left[ \sum_i \left\| \Vec \sqrt{M_i} - 
         \Vec \sqrt{M'_i} \right\|^2 \right] } \\
& \leq & \sqrt{ \left[ \sum_i \left( 2 \left\| 
         \Vec \sqrt{M_i} \right\|^2 + 
         2 \left\| \Vec \sqrt{M'_i } \right\|^2 \right) \right]
         \left[ \sum_i \left\| \Vec \sqrt{M_i} - 
         \Vec \sqrt{M'_i} \right\|^2 \right] } \\
& \leq & \sqrt{4 \cdot 2 \epsilon },
\end{eqnarray}
as desired.

\section{Formal Justification of Theorem~\ref{startingthm}}

\label{justificationapp}

Theorem~\ref{startingthm} is a special case of known results on
randomness expansion \cite{vazirani2012, MY14-1, MY14-2, Arnon:2016}.
In this appendix we demonstrate one way to derive
Theorem~\ref{startingthm}.

We first define the conditional min-entropy of a classical-quantum
state.  Any subnormalized classical-quantum state
\begin{equation}
\label{copiedstate}
\mpp{0.5}{\scalebox{0.8}{
		\begin{tikzpicture}
		\node (out1) at (-.3,1) {};
		\node (out2) at (.7,1) {};
		\node[style=qstate] (psi) at (0, 0) {~~~~~$\psi$~~~~~~};
		\draw[style=cedge] (out1) to node[left] {$C$} (out1|-psi.north);
		\draw[style=qedge] (out2) to node[left] {$Q$} (out2|-psi.north);
		\end{tikzpicture}}}
\end{equation}
where $Q = V \otimes V$, can be expressed in conventional form as an
operator on $C \otimes V$:
\begin{eqnarray}
\Psi & = & \sum_i \left| i \right> \left< i \right| \otimes M_i,
\end{eqnarray}
where each $M_i$ is a positive semidefinite matrix on $V$.  Then, the
smooth min-entropy of $C$ conditioned on $V$ is given by
\begin{eqnarray}
H_{min} ( C \mid M )_{\Psi} & = & - \log \left[ \min_{\sigma \geq M_i}
  \Tr ( \sigma ) \right],
\end{eqnarray}
where the minimum is taken over all Hermitian operators $\sigma$ on
$Q$ that satisfy $\sigma \geq M_i$ for all $i$.  (If $\Psi$ is
normalized, this quantity is the negative log of the optimal
probability with which an adversary can guess the value of $C$ given
$Q$.)

Let
\begin{equation}
\label{copiedstate2}
\mpp{0.5}{\scalebox{0.8}{
		\begin{tikzpicture}
		\node (out1) at (-.3,2) {};
		\node (out2) at (.7,2) {};
		\node[style=qstate] (psi) at (0, 0) {~~$H_{min} ( C \mid M ) \geq K$~~~};
		\draw[style=cedge] (out1) to node[left] {$C$} (out1|-psi.north);
		\draw[style=qedge] (out2) to node[left] {$Q$} (out2|-psi.north);
		\end{tikzpicture}}}
\end{equation}
denote the set of all \textit{sub}normalized classical-quantum states
$\psi$ of $CQ$ satisfying
\begin{eqnarray}
H_{min} ( C \mid M )_\psi & \geq & K.
\end{eqnarray}

\subsection{Achieving high min-entropy from a uniform seed}

We will work with Protocol $R_{gen}$ from Figure 2 in \cite{MY14-2}.
Fix the game $G$ to be the CHSH game, and let the score threshold
$\chi$ be equal to $0.85$.  Let $T, A_1, A_2, X_1, X_2$ denote
classical bit registers, and let $I = (T, A_1, A_2)$ and $O = (X_1,
X_2)$.  For any real number $q$ with $0 < q < 1$, let $B_q$ denote the
distribution on $(T, A_1, A_2)$ given by
\begin{eqnarray}
\mathbf{P} ( t, a_1, a_2 ) & = \left\{ \begin{array}{cll} 
q/4 & & \textnormal{ if } t = 1 \\
(1-q) & & \textnormal{ if } t = a_1 = a_2 = 0 \\
0 && \textnormal{ otherwise.} \end{array}  \right. 
\end{eqnarray}
(This is the distribution used in a single round of Protocol $R_{gen}$
when the game $G$ is the CHSH game.)  The protocol $R_{gen}$ can be
expressed diagrammatically as follows.
\begin{equation}
\label{appsecstatement}
\mpp{0.5}{\scalebox{0.8}{
		\begin{tikzpicture}
		\def\lineA{-.8}
		\def\lineB{.8}
		\node (out1) at (\lineA,1.5) {};
		\node[style=qstate] (dstate) at (\lineB,-1.5) {~~~\rule{0pt}{1ex}~~~};
		\node[style=cstate] (seed) at (\lineA,-1.5) {$B_q^{\otimes M}$};
		\node[style=qprocess] (re) at (0,0) {~~~$U ( M , q) $~~~};
		\node[style=terminal] (end) at (\lineB,1) {};
		\draw[style=qedge] (dstate) to (dstate|-re.south);
		\draw[style=cedge] (seed) to node[left] {$I^{\otimes M}$} (seed|-re.south);
		\draw[style=cedge] (out1) to node[left,near start] {$O^{\otimes M}$} (out1|-re.north);
		\draw[style=qedge] (end) to (end|-re.north);
		\end{tikzpicture}
}}
\end{equation}
where $U ( M, q )$ denotes the process described in steps $1-6$ in the
device-independent Protocol $R_{gen}$.

\begin{remark}
\label{completenessrem}
The highest possible quantum winning probability for the CHSH game is
$\frac{1}{2} + \frac{\sqrt{2}}{4}$, which is strictly greater than
$\chi = 0.85$.  Therefore, an honest quantum device employing the
optimal strategy at each round will succeed at Protocol $R_{gen}$ with
probability $1 - 2^{-\Omega ( q N )}$.
\end{remark}

Applying Theorem 1.1 from \cite{MY14-2} with $b := q$, and making use
of formula (5.18) from \cite{MY14-2}, we obtain the following.

\begin{theorem}
\label{altsecthm}
There exist device-independent protocols $\{ U ( M, q ) \mid M \in \{
1, 2, \ldots \}, 0 < q < 1 \}$ such that the following holds
(soundness):
\begin{equation}
\label{appsecstatementB}
\mpp{0.5}{\scalebox{0.8}{
		\begin{tikzpicture}
		\def\lineA{-.8};
		\def\lineB{.8};
		\def\lineAA{-1.5};
		\def\lineBB{1.5};
		\node (out1) at (\lineA,1.5) {};
		\node (out0) at (\lineAA,1.5) {};
		\node (out2) at (\lineBB,1.5) {};
		\node[style=qstate] (dstate) at ({(\lineB+\lineBB)/2},-1.5) {~~~~\rule{0pt}{1ex}~~~~};
		\node[style=cstate] (seed) at ({(\lineA+\lineAA)/2},-1.5) {$B_q^{\otimes M}$};
		\node[style=qprocess] (re) at (0,0) {~~~$U ( M , q) $~~~};
		\node[style=terminal] (end) at (\lineB,1) {};
		\draw[style=cedge] (out0) to node[left,near start] {$\mathbf{I}$} (out0|-seed.north);
		\draw[style=cedge] (out1) to node[right,near start] {$\mathbf{O}$} (out1|-re.north);
		\draw[style=cedge] (out1|-re.south) to node[right] {$\mathbf{I}$} (out1|-seed.north);
		\draw[style=qedge] (end) to (end|-re.north);
		\draw[style=qedge] (end|-re.south) to (end|-dstate.north);
		\draw[style=qedge] (out2) to node[right,near start] {$E$} (out2|-dstate.north);
		\end{tikzpicture}
}}
~~ {\subseteq_{2^{-\Omega ( q^2)M } }} ~~
\mpp{0.5}{\scalebox{0.8}{
		\begin{tikzpicture}
		\def\lineA{-1.5};
		\def\lineB{.1};
		\def\lineC{1.7};
		\node (out1) at (\lineA,1.5) {};
		\node (out2) at (\lineB,1.5) {};
		\node (out3) at (\lineC,1.5) {};
		\node[style=qstate] (dstate) at (0,0) {$H_{min} ( \mathbf{O} \mid \mathbf{I} E ) \geq M \cdot 0.005$};
		\draw[style=cedge] (out1) to node[left,near start] {$\mathbf{I}$} (out1|-dstate.north);
		\draw[style=cedge] (out2) to node[left,near start] {$\mathbf{O}$} (out2|-dstate.north);
		\draw[style=qedge] (out3) to node[left,near start] {$E$} (out3|-dstate.north);
		\end{tikzpicture}
}}
\end{equation}
where we have written $\mathbf{I}$ and $\mathbf{O}$ for $I^{\otimes
  M}$ and $O^{\otimes M}$ respectively, and the following also holds
(completeness):
\begin{equation}
\mpp{0.5}{\scalebox{0.8}{
		\begin{tikzpicture}
		\node[style=scalar] (scale) {$1 - 2^{{-\Omega ( q N )}}$};
\end{tikzpicture}}}
~~ {\in} ~~
\mpp{0.5}{\scalebox{0.8}{
		\begin{tikzpicture}
		\def\lineA{-.8};
		\def\lineB{.8};
		\def\lineAA{-1.5};
		\def\lineBB{1.5};
		\node[style=terminal] (out1) at (\lineA,1.5) {};
		\node[style=terminal] (out0) at (\lineAA,1.5) {};
		\node[style=terminal] (out2) at (\lineBB,1.5) {};
		\node[style=qstate] (dstate) at ({(\lineB+\lineBB)/2},-1.5) {~~~~\rule{0pt}{1ex}~~~~};
		\node[style=cstate] (seed) at ({(\lineA+\lineAA)/2},-1.5) {$B_q^{\otimes M}$};
		\node[style=qprocess] (re) at (0,0) {~~~$U ( M , q) $~~~};
		\node[style=terminal] (end) at (\lineB,1) {};
		\draw[style=cedge] (out0) to node[left,near start] {$\mathbf{I}$} (out0|-seed.north);
		\draw[style=cedge] (out1) to node[right,near start] {$\mathbf{O}$} (out1|-re.north);
		\draw[style=cedge] (out1|-re.south) to node[right] {$\mathbf{I}$} (out1|-seed.north);
		\draw[style=qedge] (end) to (end|-re.north);
		\draw[style=qedge] (end|-re.south) to (end|-dstate.north);
		\draw[style=qedge] (out2) to node[right,near start] {$E$} (out2|-dstate.north);
		\end{tikzpicture}
}}
\end{equation}
\end{theorem}
(The figure $0.005$ in diagram (\ref{appsecstatementB}) above is
somewhat arbitrary -- any figure that is less than $\pi ( 0.85 - 0.75
) \approx 0.0096$ could be used in its place and the theorem statement
would still hold true.)

Next we address the source distribution $B_q^{\otimes M}$.  The
following definition will be useful.
\begin{definition}
Let $p$ be a subnormalized probability distribution on a finite set
$X$, and let $\ell$ be an integer.  Then, $p$ is
\textnormal{$2^\ell$-rational} if $2^\ell p ( x)$ is an integer for
all $x$.
\end{definition}
Note that the condition above is equivalent to the condition that $p$
can be expressed in the form $F ( U_{2^\ell} )$, where $U$ is a
uniform random variable on a set of size $2^\ell$ and $F$ is a
deterministic process.

\begin{proposition}
\label{asB}
For any $q \in (0, 1/4)$ and $M \in \{ 1, 2, \ldots, \}$, there is a
subnormalized probability distribution $B$ on $\{ 0, 1 \}^{3M}$ such
that
\begin{enumerate}
\item The distribution $B$ is $2^{O ( q \log (1/q) M )}$-rational, and
\item \label{secondB} The statistical distance between $B$ and
  $B_q^{\otimes M}$ is $2^{-\Omega ( q M ) }$.
\end{enumerate}
\end{proposition}

\begin{proof}
Let $B'$ be the subnormalized probability distribution which assigns
$0$ to all sequences
\begin{eqnarray}
((t^1, a_1^1, a_2^1), (t^2, a_1^2, a_2^2), \ldots, (t^M, a_1^M, a_2^M) )
\end{eqnarray}
satisfying $\sum_{i=1}^M t^i > 2qM$, and assigns the same value as
$B_q^{\otimes M}$ to all other sequences.  Then, the statistical
distance between $B'$ and $B_q^{\otimes M}$ is precisely the
probability that the sum $\sum_{i=1}^M t_i$ exceeds $2qM$.  By
elementary probability arguments, this probability is $2^{-\Omega ( q
  M )}$.  Additionally, the size of the support of $B'$ is no more
than $2^{H ( 2q) M} \cdot 2^{2qM } = 2^{O ( q \log (1/q ) M)}$.
(Here, $H ( t ) = t \log (1/t) + (1-t) \log (1/(1-t))$ denotes the
binary Shannon entropy function.)

Let $t = \lceil qM + \log \left| \textnormal{Supp } B' \right|
\rceil$, and let $B$ be the subnormalized probability distribution
which assigns to each sequence $(\mathbf{t}, \mathbf{a}_1,
\mathbf{a}_2)$ the value
\begin{eqnarray}
2^{-t} \lfloor 2^{t} \mathbf{P}_{B'} (\mathbf{t}, \mathbf{a}_1,
\mathbf{a}_2) \rfloor.
\end{eqnarray}
Then, the statistical distance between $B$ and $B'$ is no more than
$2^{-t} \left| \textnormal{Supp } B' \right| \leq 2^{-qM}$, and $B$ is
$2^t$-rational.  This implies the desired result.
\end{proof}

The previous proposition asserts that we can simulate the distribution
$B_q^{\otimes M}$ up to error $2^{-\Omega ( q M )}$ by applying a
deterministic process to $O ( q \log ( 1/q) M )$ uniformly random
bits.  When we fix $q \in (0, 1/4)$ to be sufficiently small so that
the function the function represented by $O ( q \log (1/q) M )$ in
Proposition~\ref{asB} is upper bounded by $N := M/1000$ as $M \to
\infty$, and also so that the function represented by $\Omega ( q^2 )$
in Theorem~\ref{altsecthm} is positive, we obtain the following
reformulation of Theorem~\ref{altsecthm}.

\begin{theorem}
\label{altaltsecthm}
There exist device-independent protocols $T(1), T(2), T(3), \ldots$
such that the following holds for any $N \in \{ 1, 2, \ldots \}$:
\begin{equation}
\label{appsecstatementC}
\mpp{0.5}{\scalebox{0.8}{
		\begin{tikzpicture}
		\def\lineA{-.5};
		\def\lineB{.5};
		\def\lineAA{-2};
		\def\lineBB{1.5};
		\node (out1) at (\lineA,1.5) {};
		\node (out0) at (\lineAA,1.5) {};
		\node (out2) at (\lineBB,1.5) {};
		\node[style=qstate] (dstate) at ({(\lineB+\lineBB)/2-.2},-1.5) {~~~~~\rule{0pt}{1ex}~~~~~};
		\node[style=uniform] (seed) at ({(\lineA+\lineAA)/2},-1.5) {};
		\node[style=qprocess] (re) at (0,0) {~~~~$T(N)$~~~~};
		\node[style=terminal] (end) at (\lineB,1) {};
		\draw[style=cedge] (out0) to node[left,near start] {$\mathbf{J}$} (out0|-re.south) to[out=-90,in=180] (seed);
		\draw[style=cedge] (out1) to node[right,near start] {$\mathbf{O}$} (out1|-re.north);
		\draw[style=cedge] (out1|-re.south) to[out=-90,in=0] node[right] {$\mathbf{J}$} (seed);
		\draw[style=qedge] (end) to (end|-re.north);
		\draw[style=qedge] (end|-re.south) to (end|-dstate.north);
		\draw[style=qedge] (out2) to node[right,near start] {$E$} (out2|-dstate.north);
		\end{tikzpicture}
}}
~~ {\subseteq_{2^{-\Omega(N)}}} ~~
\mpp{0.5}{\scalebox{0.8}{
		\begin{tikzpicture}
		\def\lineA{-1.5};
		\def\lineB{.1};
		\def\lineC{1.7};
		\node (out1) at (\lineA,1.5) {};
		\node (out2) at (\lineB,1.5) {};
		\node (out3) at (\lineC,1.5) {};
		\node[style=qstate] (dstate) at (0,0) {$H_{min} ( \mathbf{O} \mid \mathbf{J} E ) \geq 5N$};
		\draw[style=cedge] (out1) to node[left,near start] {$\mathbf{J}$} (out1|-dstate.north);
		\draw[style=cedge] (out2) to node[left,near start] {$\mathbf{O}$} (out2|-dstate.north);
		\draw[style=qedge] (out3) to node[left,near start] {$E$} (out3|-dstate.north);
		\end{tikzpicture}
}}
\end{equation}
where the register $\mathbf{J}$ has dimension $2^N$ and the register
$\mathbf{O}$ has dimension $2^{1000N}$.  Also, the following soundness
claim holds:
\begin{equation}
\mpp{0.5}{\scalebox{0.8}{
		\begin{tikzpicture}
		\node[style=scalar] (scale) {$1 - 2^{{-\Omega ( N )}}$};
\end{tikzpicture}}}
~~ {\in} ~~
\mpp{0.5}{\scalebox{0.8}{
		\begin{tikzpicture}
		\def\lineA{-.5};
		\def\lineB{.5};
		\def\lineAA{-2};
		\def\lineBB{1.5};
		\node[style=terminal] (out1) at (\lineA,1.5) {};
		\node[style=terminal] (out0) at (\lineAA,1.5) {};
		\node[style=terminal] (out2) at (\lineBB,1.5) {};
		\node[style=qstate] (dstate) at ({(\lineB+\lineBB)/2-.2},-1.5) {~~~~~\rule{0pt}{1ex}~~~~~};
		\node[style=uniform] (seed) at ({(\lineA+\lineAA)/2},-1.5) {};
		\node[style=qprocess] (re) at (0,0) {~~~~$T(N)$~~~~};
		\node[style=terminal] (end) at (\lineB,1) {};
		\draw[style=cedge] (out0) to node[left,near start] {$\mathbf{J}$} (out0|-re.south) to[out=-90,in=180] (seed);
		\draw[style=cedge] (out1) to node[right,near start] {$\mathbf{O}$} (out1|-re.north);
		\draw[style=cedge] (out1|-re.south) to[out=-90,in=0] node[right] {$\mathbf{J}$} (seed);
		\draw[style=qedge] (end) to (end|-re.north);
		\draw[style=qedge] (end|-re.south) to (end|-dstate.north);
		\draw[style=qedge] (out2) to node[right,near start] {$E$} (out2|-dstate.north);
		\end{tikzpicture}
}}
\end{equation}
\end{theorem}

\subsection{Randomness extraction}

A randomness extractor converts a high min-entropy source into a
near-uniformly random source, with the aid of a uniformly random seed.
We will make use of a known construction (Trevisan's extractor).

\begin{theorem}
\label{trevisanthm}
Let $a (N), b(N ), c ( N ), d ( N ) \in \Theta ( N )$ be functions and
suppose that $(c(N ) - d ( N)) \in \Theta ( N )$.  Then, there exist
deterministic processes $\{ S_N \mid N \in \{ 1, 2, 3, \ldots \} \}$
and a function $e ( N ) \in \Theta ( N )$ such the following relation
holds
\begin{equation}
\label{appsecstatementD}
\mpp{0.5}{\scalebox{0.8}{
		\begin{tikzpicture}
		\def\lineA{-.6};
		\def\lineB{.6};
		\def\lineAA{-2};
		\def\lineBB{1.5};
		\node (out1) at (\lineA,1.5) {};
		\node (out0) at (\lineAA,1.5) {};
		\node (out2) at (\lineBB,1.5) {};
		\node[style=qstate] (dstate) at ({(\lineB+\lineBB)/2-.2},-1.2) {~~~$y$~~~};
		\node[style=uniform] (seed) at ({(\lineA+\lineAA)/2},-1.2) {};
		\node (seedlabel) at ({(\lineA+\lineAA)/2},-1.6) {$a(N)$};
		\node[style=qprocess] (re) at (0,0) {~~~~$S_N$~~~~};
		\node (end) at (\lineB,1) {};
		\draw[style=cedge] (out0) to  (out0|-re.south) to[out=-90,in=180] (seed);
		\draw[style=cedge] (re) to node[right,near end] {$d(N)$} (re|-out0);
		\draw[style=cedge] (out1|-re.south) to[out=-90,in=0] (seed);
		\draw[style=qedge] (end|-re.south) to node[left] {$b(N)$} (end|-dstate.north);
		\draw[style=qedge] (out2) to node[right,near start] {$E$} (out2|-dstate.north);
		\end{tikzpicture}
}}
~~ {=_{2^{-e(N)}}} ~~
\mpp{0.5}{\scalebox{0.8}{
		\begin{tikzpicture}
		\def\lineA{-1.2};
		\def\lineB{0};
		\def\lineC{1};
		\def\lineD{1.8};
		\node (out1) at (\lineA,1.5) {};
		\node (out2) at (\lineB,1.5) {};
		\node[style=terminal] (end) at (\lineC,1) {};
		\node (out3) at (\lineD,1.5) {};
		\node[style=uniform] (seed1) at (\lineA,0) {};
		\node[style=uniform] (seed2) at (\lineB,0) {};
		\node[style=qstate] (dstate) at (1.2,0) {~~~$y$~~~};
		\draw[style=cedge] (out1) to node[left,near start] {$a(N)$} (seed1);
		\draw[style=cedge] (out2) to node[left,near start] {$d(N)$} (seed2);
		\draw[style=qedge] (out3) to node[right,near start] {$E$} (out3|-dstate.north);
		\draw[style=qedge] (end) to (end|-dstate.north);
		\end{tikzpicture}
}}
\end{equation}
for any normalized state $y$ whose min-entropy conditioned on $E$ is
at least $c ( N)$.
\end{theorem}

\begin{proof}
This follows from Theorem 4.6, Lemma C.2, and Proposition C.5 in
\cite{De2012}, letting $r = 1 + \delta$ and $\epsilon = 2^{-\delta N}$
where $\delta$ is a sufficiently small constant.
\end{proof}

We next assert that Theorem~\ref{trevisanthm} continues to hold when
the phrase ``normalized state'' is replaced with ``subnormalized
state.''
\begin{corollary}
\label{trevisancor}
Let $a (N), b(N ), c ( N ), d ( N ) \in \Theta ( N )$ be functions and
suppose that $(c(N ) - d ( N)) \in \Theta ( N )$.  Then, there exist
deterministic processes $\{ V_N \mid N \in \{ 1, 2, 3, \ldots \} \}$
and a function $e ( N ) \in \Theta ( N )$ such the following relation
holds
\begin{equation}
\label{trevisanpic}
\mpp{0.5}{\scalebox{0.8}{
		\begin{tikzpicture}
		\def\lineA{-.6};
		\def\lineB{.6};
		\def\lineAA{-2};
		\def\lineBB{1.5};
		\node (out1) at (\lineA,1.5) {};
		\node (out0) at (\lineAA,1.5) {};
		\node (out2) at (\lineBB,1.5) {};
		\node[style=qstate] (dstate) at ({(\lineB+\lineBB)/2-.2},-1.2) {~~~$y$~~~};
		\node[style=uniform] (seed) at ({(\lineA+\lineAA)/2},-1.2) {};
		\node (seedlabel) at ({(\lineA+\lineAA)/2},-1.6) {$a(N)$};
		\node[style=qprocess] (re) at (0,0) {~~~~$V_N$~~~~};
		\node (end) at (\lineB,1) {};
		\draw[style=cedge] (out0) to  (out0|-re.south) to[out=-90,in=180] (seed);
		\draw[style=cedge] (re) to node[right,near end] {$d(N)$} (re|-out0);
		\draw[style=cedge] (out1|-re.south) to[out=-90,in=0] (seed);
		\draw[style=qedge] (end|-re.south) to node[left] {$b(N)$} (end|-dstate.north);
		\draw[style=qedge] (out2) to node[right,near start] {$E$} (out2|-dstate.north);
		\end{tikzpicture}
}}
~~ {{=}_{2^{-e(N)}}} ~~
\mpp{0.5}{\scalebox{0.8}{
		\begin{tikzpicture}
		\def\lineA{-1.2};
		\def\lineB{0};
		\def\lineC{1};
		\def\lineD{1.8};
		\node (out1) at (\lineA,1.5) {};
		\node (out2) at (\lineB,1.5) {};
		\node[style=terminal] (end) at (\lineC,1) {};
		\node (out3) at (\lineD,1.5) {};
		\node[style=uniform] (seed1) at (\lineA,0) {};
		\node[style=uniform] (seed2) at (\lineB,0) {};
		\node[style=qstate] (dstate) at (1.2,0) {~~~$y$~~~};
		\draw[style=cedge] (out1) to node[left,near start] {$a(N)$} (seed1);
		\draw[style=cedge] (out2) to node[left,near start] {$d(N)$} (seed2);
		\draw[style=qedge] (out3) to node[right,near start] {$E$} (out3|-dstate.north);
		\draw[style=qedge] (end) to (end|-dstate.north);
		\end{tikzpicture}
}}
\end{equation}
for any \textbf{subnormalized} state $y$ whose min-entropy conditioned
on $E$ is at least $c ( N)$.
\end{corollary}

\begin{proof}
Choose a function $c'(N)$ such that $c'(N) - d ( N ) \in \Theta ( N)$
and $c ( N ) - c'(N) \in \Theta ( N)$.  By Theorem~\ref{trevisanthm},
there exists $e'(N) \in \Theta ( N)$ and deterministic processes $\{
V_N \}$ such that the following holds for any normalized state $y$
whose min-entropy conditioned on $E$ is at least $c' ( N )$:
\begin{equation}
\mpp{0.5}{\scalebox{0.8}{
		\begin{tikzpicture}
		\def\lineA{-.6};
		\def\lineB{.6};
		\def\lineAA{-2};
		\def\lineBB{1.5};
		\node (out1) at (\lineA,1.5) {};
		\node (out0) at (\lineAA,1.5) {};
		\node (out2) at (\lineBB,1.5) {};
		\node[style=qstate] (dstate) at ({(\lineB+\lineBB)/2-.2},-1.2) {~~~$y$~~~};
		\node[style=uniform] (seed) at ({(\lineA+\lineAA)/2},-1.2) {};
		\node (seedlabel) at ({(\lineA+\lineAA)/2},-1.6) {$a(N)$};
		\node[style=qprocess] (re) at (0,0) {~~~~$V_N$~~~~};
		\node (end) at (\lineB,1) {};
		\draw[style=cedge] (out0) to  (out0|-re.south) to[out=-90,in=180] (seed);
		\draw[style=cedge] (re) to node[right,near end] {$d(N)$} (re|-out0);
		\draw[style=cedge] (out1|-re.south) to[out=-90,in=0] (seed);
		\draw[style=qedge] (end|-re.south) to node[left] {$b(N)$} (end|-dstate.north);
		\draw[style=qedge] (out2) to node[right,near start] {$E$} (out2|-dstate.north);
		\end{tikzpicture}
}}
~~ {=_{2^{-e'(N)}}} ~~
\mpp{0.5}{\scalebox{0.8}{
		\begin{tikzpicture}
		\def\lineA{-1.2};
		\def\lineB{0};
		\def\lineC{1};
		\def\lineD{1.8};
		\node (out1) at (\lineA,1.5) {};
		\node (out2) at (\lineB,1.5) {};
		\node[style=terminal] (end) at (\lineC,1) {};
		\node (out3) at (\lineD,1.5) {};
		\node[style=uniform] (seed1) at (\lineA,0) {};
		\node[style=uniform] (seed2) at (\lineB,0) {};
		\node[style=qstate] (dstate) at (1.2,0) {~~~$y$~~~};
		\draw[style=cedge] (out1) to node[left,near start] {$a(N)$} (seed1);
		\draw[style=cedge] (out2) to node[left,near start] {$d(N)$} (seed2);
		\draw[style=qedge] (out3) to node[right,near start] {$E$} (out3|-dstate.north);
		\draw[style=qedge] (end) to (end|-dstate.north);
		\end{tikzpicture}
}}
\end{equation}
Let $e(N) \in \Theta ( N )$ be a function that is less than or equal
to both $(c(N) - c'(N))$ and $e'(N)/2$ for all $N$.  Let $y'$ be an
arbitrary \textit{sub}normalized state whose min-entropy conditioned
on $E$ is at least $c ( N)$ (rather than $c' ( N )$).  If the trace of
$y'$ is at least $2^{-e ( N )}$, then the next relation follows easily
from the previous one (with $y := y'/Tr ( y')$):
\begin{equation}
\label{subconnect}
\mpp{0.5}{\scalebox{0.8}{
		\begin{tikzpicture}
		\def\lineA{-.6};
		\def\lineB{.6};
		\def\lineAA{-2};
		\def\lineBB{1.5};
		\node (out1) at (\lineA,1.5) {};
		\node (out0) at (\lineAA,1.5) {};
		\node (out2) at (\lineBB,1.5) {};
		\node[style=qstate] (dstate) at ({(\lineB+\lineBB)/2-.2},-1.2) {~~~$y'$~~~};
		\node[style=uniform] (seed) at ({(\lineA+\lineAA)/2},-1.2) {};
		\node (seedlabel) at ({(\lineA+\lineAA)/2},-1.6) {$a(N)$};
		\node[style=qprocess] (re) at (0,0) {~~~~$V_N$~~~~};
		\node (end) at (\lineB,1) {};
		\draw[style=cedge] (out0) to  (out0|-re.south) to[out=-90,in=180] (seed);
		\draw[style=cedge] (re) to node[right,near end] {$d(N)$} (re|-out0);
		\draw[style=cedge] (out1|-re.south) to[out=-90,in=0] (seed);
		\draw[style=qedge] (end|-re.south) to node[left] {$b(N)$} (end|-dstate.north);
		\draw[style=qedge] (out2) to node[right,near start] {$E$} (out2|-dstate.north);
		\end{tikzpicture}
}}
~~ {=_{2^{-e(N)}}} ~~
\mpp{0.5}{\scalebox{0.8}{
		\begin{tikzpicture}
		\def\lineA{-1.2};
		\def\lineB{0};
		\def\lineC{1};
		\def\lineD{1.8};
		\node (out1) at (\lineA,1.5) {};
		\node (out2) at (\lineB,1.5) {};
		\node[style=terminal] (end) at (\lineC,1) {};
		\node (out3) at (\lineD,1.5) {};
		\node[style=uniform] (seed1) at (\lineA,0) {};
		\node[style=uniform] (seed2) at (\lineB,0) {};
		\node[style=qstate] (dstate) at (1.2,0) {~~~$y'$~~~};
		\draw[style=cedge] (out1) to node[left,near start] {$a(N)$} (seed1);
		\draw[style=cedge] (out2) to node[left,near start] {$d(N)$} (seed2);
		\draw[style=qedge] (out3) to node[right,near start] {$E$} (out3|-dstate.north);
		\draw[style=qedge] (end) to (end|-dstate.north);
		\end{tikzpicture}
}}
\end{equation}

On the other hand, if the trace of $y'$ is less than $2^{-e ( N ) }$,
then both diagrams in relation (\ref{subconnect}) have trace less than
$2^{-e ( N )}$ and so the relation obviously holds.  This completes
the proof.
\end{proof}

\subsection{Randomness expansion}

Now we combine the results of the previous two subsections to prove
Theorem~\ref{startingthm}. For any positive integer $M$, let $T(M)$
denote the protocol defined in Theorem~\ref{altaltsecthm}.  Define a
device-independent protocol $R ( M )$ by
\begin{equation}
\label{defofr}
\mpp{0.5}{\scalebox{0.8}{
		\begin{tikzpicture}
		\node (in1) at (-.5,-1.5) {};
		\node (in2) at (.5,-1.5) {};
		\node (out1) at (-.5,1.5) {};
		\node (out2) at (.5,1.5) {};
		\node[style=qprocess] (re) at (0,0) {~~~$R(M)$~~~};
		\draw[style=cedge] (in1) to node[left,near start] {$M$} (in1|-re.south);
		\draw[style=qedge] (in2) to (in2|-re.south);
		\draw[style=cedge] (out1) to node[left,near start] {$2M$}  (out1|-re.north);
		\draw[style=qedge] (out2) to (out2|-re.north);
		\end{tikzpicture}
}}
~~ = ~~
\mpp{0.5}{\scalebox{0.8}{
		\begin{tikzpicture}
		\node (in1) at (-.8,-1.5) {};
		\node (in1a) at (-.2,-1.5) {};
		\node (in2) at (1,-1.5) {};
		\node (out1) at (-.5,1.5) {};
		\node (out2) at (1,1.5) {};
		\node[style=qprocess] (re) at (0,0) {~~~~~~$R(M)$~~~~~~};
		\draw[style=cedge] (in1) to node[left,near start] {$\lfloor \frac{M}{2}\rfloor$} (in1|-re.south);
		\draw[style=cedge] (in1a) to node[right,near start] {$\lceil\frac{M}{2}\rceil$} (in1a|-re.south);
		\draw[style=qedge] (in2) to (in2|-re.south);
		\draw[style=cedge] (out1) to node[left,near start] {$2M$}  (out1|-re.north);
		\draw[style=qedge] (out2) to (out2|-re.north);
		\end{tikzpicture}
}}
~~ := ~~
\mpp{0.5}{\scalebox{0.8}{
		\begin{tikzpicture}
		\def\inY{-1}
		\def\outY{3.5}
		\node (in1) at (-1.2,\inY) {};
		\node (in1a) at (-.2,\inY) {};
		\node (in2) at (1.6,\inY) {};
		\node (out1) at (-.6,\outY) {};
		\node (out2) at (1.6,\outY) {};
		\node[style=qprocess] (T) at (.5,.5) {~~~~$T\left(\lceil\frac{M}{2}\rceil\right)$~~~~};
		\node[style=qprocess] (V) at (-.6,2.5) {~~~~~$V_M$~~~~~};
		\draw[style=cedge] (in1a) to node[right,near start] {$\lceil \frac{M}{2} \rceil$} (in1a|-T.south);
		\draw[style=qedge] (in2) to (in2|-T.south);
		\draw[style=qedge] (out2) to (out2|-T.north);
		\draw[style=cedge] (in1) to node[left, very near start] {$\lfloor \frac{M}{2} \rfloor$} (in1|-V.south);
		\draw[style=cedge] (in1a|-T.north) to node {~~~~~~~~~~~~$1000\lceil\frac{M}{2}\rceil$} (in1a|-V.south);
		\draw[style=cedge] (out1) to node[left,near start] {$2M$} (V);
		\end{tikzpicture}
}}
\end{equation}
where $V_M$ denotes the process defined in
Corollary~\ref{trevisancor}, with $c(M)=\frac{5M}{2}$.

Then for any $r\in R(N)$ and normalized state $\Gamma$, there is some
$t \in T ( \lceil \frac{M}{2 }\rceil )$ satisfying the following (the
twist in the upper-left maintains the order of the original register
$M$):
\begin{equation}
\mpp{0.5}{\scalebox{.8}{\begin{tikzpicture}
		\node (out1) at (-.2,2.8) {};
		\node (out2) at (-1.4,2.8) {};
		\node (out3) at (1.2,2.8) {};
		\node[style=terminal] (end1) at (0.5,2) {};
		\node[style=qstate] (dstate1) at (0.7,0) {~~$\Gamma$~~};
		\node[style=uniform] (seed1) at (-.8,0) {};
		\node at (-.8,-.4) {$M$};
		\node[style=qprocess] (re1) at (0,1) {~~~~$r$~~~~};
		\draw[style=qedge] (end1|-dstate1.north) to (end1|-re1.south);
		\draw[style=cedge] (seed1) to[out=0,in=-90] (out1|-re1.south);
		\draw[style=cedge] (out1) to node[left,near start] { $2M$}(out1|-re1.north);
		\draw[style=cedge] (out2) to[out=-90,in=90]  (-1.4,1) to[out=-90,in=180] (seed1);
		\draw[style=qedge] (out3) to (out3|-dstate1.north);
		\draw[style=qedge] (end1) to (end1|-re1.north);
		\end{tikzpicture}}}
~~ = ~~
\mpp{0.5}{\scalebox{.8}{\begin{tikzpicture}
		\node (out1) at (-.2,2.8) {};
		\node (out2) at (-1.4,2.8) {};
		\node (out2a) at (-2,2.8) {};	  
		\node (out3) at (1.2,2.8) {};
		\node[style=terminal] (end1) at (0.5,2) {};
		\node[style=qstate] (dstate1) at (0.7,-.5) {~~$\Gamma$~~};
		\node[style=uniform] (seed1) at (-.8,.2) {};
		\node at (-.8,-.2) {$\lfloor \frac{M}{2}\rfloor$};
		\node[style=uniform] (seed2) at (-.8,-.8) {};
		\node at (-.8,-1.2) {$\lceil \frac{M}{2}\rceil$};
		\node[style=qprocess] (re1) at (0,1) {~~~~$r$~~~~};
		\draw[style=qedge] (end1|-dstate1.north) to (end1|-re1.south);
		\draw[style=cedge] (seed1) to[out=0,in=-90] (re1.-145);
		\draw[style=cedge] (seed2) to[out=0,in=-90] (0,.2) to (re1.south);
		\draw[style=cedge] (out1) to node[left,near start] { $2M$}(out1|-re1.north);
		\draw[style=cedge] (out2a) to (out2a|-end1.south) to[out=-90,in=90] (out2|-re1.north) to (-1.4,1) to[out=-90,in=180] (seed1);
		\draw[style=cedge] (out2) to (out2|-end1.south) to[out=-90,in=90] (out2a|-re1.north) to (-2,.2)  to[out=-90,in=180] (seed2);
		\draw[style=qedge] (out3) to (out3|-dstate1.north);
		\draw[style=qedge] (end1) to (end1|-re1.north);
		\end{tikzpicture}}}
~~ = ~~
\mpp{0.5}{\scalebox{0.8}{
		\begin{tikzpicture}
		\def\inY{-1.8}
		\def\outY{3.5}
		\node (in1) at (-1.2,\inY) {};
		\node (in1a) at (-.2,\inY) {};
		\node (in2) at (1.2,\inY) {};
		\node (out1) at (-.6,\outY) {};
		\node (out2) at (1.6,\outY) {};
		\node (out3) at (-2.8,\outY) {};
		\node (out4) at (-3.5,\outY) {};
		\node[style=terminal] (end) at (.8,1.5) {};
		\node[style=qprocess] (T) at (.2,.5) {~~~~~~$t$~~~~~~};
		\node[style=qprocess] (V) at (-.6,2) {~~~~~$V_M$~~~~~};
		\node[style=qstate] (gamma) at (1,-1) {~~$\Gamma$~~};
		\node[style=uniform] (seed1) at (-2,0) {};
		\node at (-2,-.4) {$\lfloor \frac{M}{2} \rfloor$};
		\node[style=uniform] (seed2) at (-2,-1.5) {};
		\node at (-2,-1.9) {$\lceil \frac{M}{2} \rceil$} ;
		\draw[style=cedge] (in1a|-T.south)  to[out=-90,in=0]  (seed2);
		\draw[style=cedge] (out3) to (out3|-V.north) to[out=-90,in=90] (out4|-end.south) to (out4|-end.south) to[out=-90,in=180] (seed2);
		\draw[style=cedge] (out4) to (out4|-V.north) to[out=-90,in=90] (out3|-end.south) to[out=-90,in=180] (seed1);
		\draw[style=qedge] (end|-gamma.north) to (end|-T.south);
		\draw[style=qedge] (end) to (end|-T.north);
		\draw[style=qedge] (out2) to (out2|-gamma.north);
		\draw[style=cedge] (in1|-V.south) to[out=-90,in=0]  (seed1);
		\draw[style=cedge] (in1a|-T.north) to (in1a|-V.south);
		\draw[style=cedge] (out1) to node[left,near start] {$2M$} (V);
		\end{tikzpicture}
}}
\end{equation}    
Now we can apply Theorem~\ref{altaltsecthm} and
Corollary~\ref{trevisancor} to conclude that

\begin{equation}
\mpp{0.5}{\scalebox{0.8}{
		\begin{tikzpicture}
		\def\inY{-1.8}
		\def\outY{3.5}
		\node (in1) at (-1.2,\inY) {};
		\node (in1a) at (-.2,\inY) {};
		\node (in2) at (1.2,\inY) {};
		\node (out1) at (-.6,\outY) {};
		\node (out2) at (2.4,\outY) {};
		\node (out3) at (-2.8,\outY) {};
		\node (out4) at (-3.5,\outY) {};
		\node[style=terminal] (end) at (1.2,1.5) {};
		\node[style=qprocess] (T) at (.5,.5) {~~~~$T\left(\lceil\frac{M}{2}\rceil\right)$~~~~};
		\node[style=qprocess] (V) at (-.6,2) {~~~~~$V_M$~~~~~};
		\node[style=qstate] (gamma) at (1.5,-1) {~~~~$\Gamma$~~~~};
		\node[style=uniform] (seed1) at (-2,.75) {};
		\node at (-2,.35) {$\lfloor \frac{M}{2} \rfloor$};
		\node[style=uniform] (seed2) at (-2,-1.5) {};
		\node at (-2,-1.9) {$\lceil \frac{M}{2} \rceil$} ;
		\draw[style=cedge] (in1a|-T.south)  to[out=-90,in=0] node[right,near start] {$\lceil \frac{M}{2} \rceil$} (seed2);
		\draw[style=cedge] (out3) to (out3|-V.north) to[out=-90,in=90] (out4|-end.south) to (out4|-end.south) to[out=-90,in=180] (seed2);
		\draw[style=cedge] (out4) to (out4|-V.north) to[out=-90,in=90] (out3|-end.south) to[out=-90,in=180] (seed1);
		\draw[style=qedge] (in2|-gamma.north) to (in2|-T.south);
		\draw[style=qedge] (end) to (end|-T.north);
		\draw[style=qedge] (out2) to (out2|-gamma.north);
		\draw[style=cedge] (in1|-V.south) to[out=-90,in=0]  (seed1);
		\draw[style=cedge] (in1a|-T.north) to (in1a|-V.south);
		\draw[style=cedge] (out1) to node[left,near start] {$2M$} (V);
		\end{tikzpicture}
}}
~~ =_{2^{-\Omega(M)}}
\mpp{0.5}{\scalebox{0.8}{
		\begin{tikzpicture}
		\def\inY{-1.8}
		\def\outY{3.5}
		\node (in1) at (-1.2,\inY) {};
		\node (in1a) at (-.2,\inY) {};
		\node (in2) at (1.2,\inY) {};
		\node (out1) at (-.6,\outY) {};
		\node (out2) at (2.4,\outY) {};
		\node (out3) at (-2.8,\outY) {};
		\node (out4) at (-3.5,\outY) {};
		\node[style=terminal] (end) at (1.2,1.5) {};
		\node[style=terminal] (end2) at (-.2,1.5) {};
		\node[style=qprocess] (T) at (.5,.5) {~~~~$T\left(\lceil\frac{M}{2}\rceil\right)$~~~~};
		\node (V) at (-.6,2.2) {};
		\node[style=qstate] (gamma) at (1.5,-1) {~~~~$\Gamma$~~~~};
		\node[style=uniform] (seed1) at (-2.4,.5) {};
		\node at (-2.4,.1) {$\lfloor \frac{M}{2} \rfloor$};
		\node[style=uniform] (seed2) at (-2,-1.5) {};
		\node at (-2,-1.9) {$\lceil \frac{M}{2} \rceil$} ;
		\node[style=uniform] (seed3) at (-1.6,.5) {};
		\draw[style=cedge] (in1a|-T.south)  to[out=-90,in=0] (seed2);
		\draw[style=cedge] (out3) to (out3|-V.north) to[out=-90,in=90] (out4|-T.north) to (out4|-T.south) to[out=-90,in=180] (seed2);
		\draw[style=cedge] (out4) to (out4|-V.north) to[out=-90,in=90] (seed1|-T.north) to (seed1);
		\draw[style=qedge] (in2|-gamma.north) to (in2|-T.south);
		\draw[style=qedge] (end) to (end|-T.north);
		\draw[style=qedge] (out2) to (out2|-gamma.north);
		\draw[style=cedge] (end2) to (end2|-T.north);
		\draw[style=cedge] (seed3|-out1) to node[left,near start] {$2M$} (seed3);
		\end{tikzpicture}
}}
\end{equation}   
from which it follows easily that
\begin{equation}
\mpp{0.5}{\scalebox{1.0}{\begin{tikzpicture}
		\node (out1) at (-.2,2.8) {};
		\node (out2) at (-1.4,2.8) {};
		\node (out3) at (1.2,2.8) {};
		\node[style=terminal] (end1) at (0.5,2) {};
		\node[style=qstate] (dstate1) at (0.7,0) {~~$\Gamma$~~};
		\node[style=uniform] (seed1) at (-.8,0) {};
		\node at (-.8,-.4) {$M$};
		\node[style=qprocess] (re1) at (0,1) {~~~~$r$~~~~};
		\draw[style=qedge] (end1|-dstate1.north) to (end1|-re1.south);
		\draw[style=cedge] (seed1) to[out=0,in=-90] (out1|-re1.south);
		\draw[style=cedge] (out1) to node[left,near start] { $2M$}(out1|-re1.north);
		\draw[style=cedge] (out2) to[out=-90,in=90]  (-1.4,1) to[out=-90,in=180] (seed1);
		\draw[style=qedge] (out3) to (out3|-dstate1.north);
		\draw[style=qedge] (end1) to (end1|-re1.north);
		\end{tikzpicture}}}
~~ =_{2^{-\Omega(M)}} ~~
\mpp{0.5}{\scalebox{1.0}{\begin{tikzpicture}
		\node[style=qstate] (dstate2) at (5,0) {~~$\Gamma$~~};
		\node[style=uniform] (seed2) at (3.4,0) {};
		\node at (3.4,-.4) {$M$};
		\node[style=qprocess] (re2) at (4.2,1) {~~~~$r$~~~~};
		\node (out4) at (4,3) {};
		\node (out5) at (2.8,3) {};
		\node (out6) at (5.4,3) {};
		\node[style=terminal] (end2) at (4.7,1.9) {};
		\node[style=uniform] (seed3) at (4,2.3) {};
		\node[style=terminal] (end3) at (4,1.9) {};
		\draw[style=qedge] (end2|-dstate2.north) to (end2|-re2.south);
		\draw[style=cedge] (seed2) to[out=0,in=-90]  (out4|-re2.south);
		\draw[style=cedge] (out5) to[out=-90,in=90] (2.8,1) to[out=-90,in=180] (seed2);
		\draw[style=qedge] (out6) to (out6|-dstate2.north);
		\draw[style=qedge] (end2) to  (end2|-re2.north);
		\draw[style=cedge] (out4) to node[left,near start] {$2M$} (seed3);
		\draw[style=cedge] (end3) to node[left] {$2M$} (end3|-re2.north);
		\draw[style=qedge] (end2) to (end2|-re2.north);
		\end{tikzpicture}}}
\end{equation}    
This implies the soundness claim in Theorem~\ref{startingthm}.  The
completeness claim in Theorem~\ref{startingthm} follows easily from
the completeness claim in Theorem~\ref{altaltsecthm}.

\end{document}